\documentclass{article}

\PassOptionsToPackage{numbers,compress}{natbib}

\usepackage[final]{neurips_2021}

\usepackage[utf8]{inputenc} 
\usepackage[T1]{fontenc}    
\usepackage[hidelinks, colorlinks=true, urlcolor=blue, linkcolor=black, citecolor=blue]{hyperref}       
\usepackage{url}            
\usepackage{booktabs}       
\usepackage{amsfonts}       
\usepackage{nicefrac}       
\usepackage{microtype}      
\usepackage{xcolor}         

\usepackage[pdftex]{graphicx}
\usepackage{enumitem}
\usepackage{float}
\usepackage{physics}
\usepackage{bm}
\usepackage{bbm}
\usepackage{textcomp}
\usepackage[caption=false]{subfig}
\usepackage[export]{adjustbox}

\usepackage{amsthm}
\newtheorem{theorem}{Theorem}
\newtheorem{definition}{Definition}
\newtheorem{corollary}{Corollary}
\usepackage{thmtools}
\usepackage{thm-restate}
\declaretheorem[name=Lemma]{lemma}

\usepackage{amssymb}

\usepackage{qcircuit}
\usepackage{tikz}
\usepackage{array}
\usepackage{xcolor}
\usepackage{mathtools}
\usepackage[algo2e, linesnumbered, ruled, vlined]{algorithm2e}
\usepackage{etoolbox}

\newcommand{\params}{\ensuremath{\bm{\theta}}}
\newcommand{\policy}{\ensuremath{\pi_{\params}}}
\newcommand{\paramsV}{\ensuremath{\bm{\omega}}}
\newcommand{\approxV}{\ensuremath{\widetilde{V}_{\paramsV}}}

\newcommand{\phis}{\ensuremath{\bm{\phi}}}
\newcommand{\lambdas}{\ensuremath{\bm{\lambda}}}
\newcommand{\weights}{\ensuremath{\bm{w}}}

\newcommand{\CZ}[1]{%
\begin{tikzpicture}[#1]%
\draw[thick] (0ex,-0.5ex) -- (0ex, 0.5ex);%
\filldraw [black] (0ex, -0.5ex) circle (0.7pt);
\filldraw [black] (0ex, 0.5ex) circle (0.7pt);
\end{tikzpicture}%
}

\title{Parametrized Quantum Policies \\for Reinforcement Learning}

\author{%
  Sofiene Jerbi \\
  Institute for Theoretical Physics,\\
  University of Innsbruck \\
  \texttt{sofiene.jerbi@uibk.ac.at} \\
  \And
  Casper Gyurik \\
  LIACS, \\
  Leiden University \\
  \And
  Simon C. Marshall \\
  LIACS, \\
  Leiden University \\
  \And
  Hans J. Briegel \\
  Institute for Theoretical Physics,\\
  University of Innsbruck\\
  \And
  Vedran Dunjko \\
  LIACS, \\
  Leiden University \\
}

\begin{document}

\maketitle

\begin{abstract}

With the advent of real-world quantum computing, the idea that parametrized quantum computations can be used as hypothesis families in a quantum-classical machine learning system is gaining increasing traction. Such hybrid systems have already shown the potential to tackle real-world tasks in supervised and generative learning, and recent works have established their provable advantages in special artificial tasks. Yet, in the case of reinforcement learning, which is arguably most challenging and where learning boosts would be extremely valuable, no proposal has been successful in solving even standard benchmarking tasks, nor in showing a theoretical learning advantage over classical algorithms. In this work, we achieve both. We propose a hybrid quantum-classical reinforcement learning model using very few qubits, which we show can be effectively trained to solve several standard benchmarking environments. Moreover, we demonstrate, and formally prove, the ability of parametrized quantum circuits to solve certain learning tasks that are intractable to classical models, including current state-of-art deep neural networks, under the widely-believed classical hardness of the discrete logarithm problem.

\end{abstract}

\section{\label{sec:intro}Introduction}

Hybrid quantum machine learning models constitute one of the most promising applications of near-term quantum computers \cite{preskill18,bharti21}. In these models, parametrized and data-dependent quantum computations define a hypothesis family for a given learning task, and a classical optimization algorithm is used to train them. For instance, parametrized quantum circuits (PQCs) \cite{benedetti19} have already proven successful in classification \cite{farhi18,schuld20b,havlivcek19,schuld19b,peters21}, generative modeling \cite{liu18,zhu19} and clustering \cite{otterbach17} problems. Moreover, recent results have shown proofs of their learning advantages in artificially constructed tasks \cite{havlivcek19,huang20}, some of which are based on widely believed complexity-theoretic assumptions \cite{huang20,du20,liu20,sweke20}. All these results, however, only consider supervised and generative learning settings.

Arguably, the largest impact quantum computing can have is by providing enhancements to the hardest learning problems. From this perspective, reinforcement learning (RL) stands out as a field that can greatly benefit from a powerful hypothesis family. This is showcased by the boost in learning performance that deep neural networks (DNNs) have provided to RL \cite{mnih15}, which enabled systems like AlphaGo \cite{silver17}, among other achievements \cite{berner19,mirowski18}. Nonetheless, the true potential of near-term quantum approaches in RL remains very little explored. The few existing works \cite{chen20,lockwood20,wu20,jerbi19} have failed so far at solving classical benchmarking tasks using PQCs and left open the question of their ability to provide a learning advantage.

\paragraph{Contributions}In this work, we demonstrate the potential of policies based on PQCs in solving classical RL environments. To do this, we first propose new model constructions, describe their learning algorithms, and show numerically the influence of design choices on their learning performance. In our numerical investigation, we consider benchmarking environments from OpenAI Gym \cite{brockman16}, for which good and simple DNN policies are known, and in which we demonstrate that PQC policies can achieve comparable performance. Second, inspired by the classification task of Havlí\v{c}ek \emph{et al.}\ \cite{havlivcek19}, conjectured to be classically hard by the authors, we construct analogous RL environments where we show an empirical learning advantage of our PQC policies over standard DNN policies used in deep RL. In the same direction, we construct RL environments with a provable gap in performance between a family of PQC policies and any efficient classical learner. These environments essentially build upon the work of Liu \emph{et al.}\ \cite{liu20} by embedding into a learning setting the discrete logarithm problem (DLP), which is the problem solved by Shor's celebrated quantum algorithm \cite{shor99} but widely believed to be classically hard to solve \cite{blum84}.

\begin{figure*}
  \makebox{\includegraphics[width=\linewidth, valign=c]{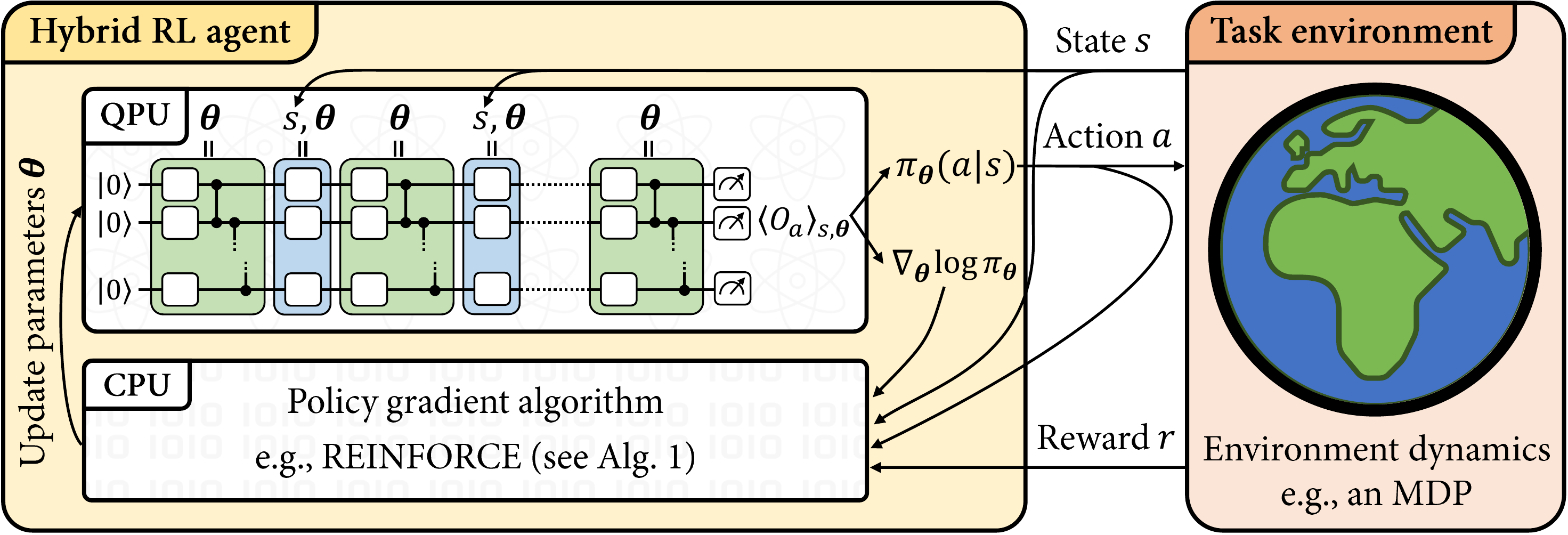}}
  \caption{\textbf{Training parametrized quantum policies for reinforcement learning. }We consider a quantum-enhanced RL scenario where a hybrid quantum-classical agent learns by interacting with a classical environment. For each state $s$ it perceives, the agent samples its next action $a$ from its policy $\policy(a|s)$ and perceives feedback on its behavior in the form of a reward $r$. For our hybrid agents, the policy $\policy$ is specified by a PQC (see Def.\ \ref{def:raw-softmax-PQC}) evaluated (along with the gradient $\nabla_{\params} \log \policy$) on a quantum processing unit (QPU). The training of this policy is performed by a classical learning algorithm, such as the \textsc{REINFORCE} algorithm (see Alg.\ \ref{alg}), which uses sample interactions and policy gradients to update the policy parameters $\params$.}
  \label{fig:pqc-rl}

\end{figure*}

\paragraph{Related work}Recently, a few works have been exploring hybrid quantum approaches for RL. Among these, Refs.\ \cite{chen20,lockwood20} also trained PQC-based agents in classical RL environments. However, these take a value-based approach to RL, meaning that they use PQCs as value-function approximators instead of direct policies. The learning agents in these works are also tested on OpenAI Gym environments (namely, a modified FrozenLake and CartPole), but do not achieve sufficiently good performance to be solving them, according to the Gym specifications. Ref.~\cite{skolik21b} shows that, using some of our design choices for PQCs in RL (i.e., data re-uploading circuits \cite{perez20} with trainable observable weights and input scaling parameters), one can also solve these environments using a value-based approach. An actor-critic approach to QRL was introduced in Ref.\ \cite{wu20}, using both a PQC actor (or policy) and a PQC critic (or value-function approximator). In contrast to our work, these are trained in quantum environments (e.g., quantum-control environments), that provide a quantum state to the agent, which acts back with a continuous classical action. These aspects make it a very different learning setting to ours. Ref.\ \cite{jerbi19} also describes a hybrid quantum-classical algorithm for value-based RL. The function-approximation models on which this algorithm is applied are however not PQCs but energy-based neural networks (e.g., deep and quantum Boltzmann machines). Finally, our work provides an alternative approach to take advantage of quantum effects in designing QRL agents compared to earlier approaches \cite{dong08,paparo14,dunjko16,crawford18,neukart18}, which are mainly based on (variations of) Grover's search algorithm \cite{grover96} or quantum annealers \cite{johnson11} to speed up sampling routines.

\paragraph{Code}An accompanying tutorial \cite{tfq21}, implemented as part of the quantum machine learning library TensorFlow Quantum \cite{broughton20}, provides the code required to reproduce our numerical results and explore different settings. It also implements the Q-learning approach for PQC-based RL of Skolik \emph{et al.} \cite{skolik21b}

\section{Parametrized quantum policies: definitions and learning algorithm\label{sec:pqc-policies}}

In this section, we give a detailed construction of our parametrized quantum policies and describe their associated training algorithms. We start however with a short introduction to the basic concepts of quantum computation, introduced in more detail in \cite{nielsen00,dewolf19}. 

\subsection{Quantum computation: a primer} 

A quantum system composed of $n$ qubits is represented by a $2^n$-dimensional complex Hilbert space $\mathcal{H}=(\mathbb{C}^2)^{\otimes n}$. Its quantum state is described by a vector $\ket{\psi} \in \mathcal{H}$ of unit norm $\braket{\psi}=1$, where we adopt the bra-ket notation to describe vectors $\ket{\psi}$, their conjugate transpose $\bra{\psi}$ and inner-products $\braket{\psi}{\psi'}$ in $\mathcal{H}$. Single-qubit computational basis states are given by $\ket{0}=(1,0)^T, \ket{1}=(0,1)^T$, and their tensor products describe general computational basis states, e.g., $\ket{10} = \ket{1}\otimes\ket{0} = (0,0,1,0)$. 

A quantum gate is a unitary operation $U$ acting on $\mathcal{H}$. When a gate $U$ acts non-trivially only on a subset $S \subseteq [n]$ of qubits, we identify it to the operation $U\otimes \mathbbm{1}_{[n]\backslash S}$. In this work, we are mainly interested in the single-qubit Pauli gates $Z,Y$ and their associated rotations $R_z,R_y$:
\begin{equation}
	Z = \begin{pmatrix} 1 & 0 \\ 0 & -1 \end{pmatrix}, R_z(\theta) = \exp(-i \frac{\theta}{2} Z), \quad Y = \begin{pmatrix} 0 & -i \\ i & 0 \end{pmatrix}, R_y(\theta) = \exp(-i \frac{\theta}{2} Y),
\end{equation}
for rotation angles $\theta\in\mathbb{R}$, and the 2-qubit Ctrl-$Z$ gate $\CZ{scale=1.5} = \text{diag}(1,1,1,-1)$.

A projective measurement is described by a Hermitian operator $O$ called an observable. Its spectral decomposition $O=\sum_m \alpha_m P_m$ in terms of eigenvalues $\alpha_m$ and orthogonal projections $P_m$ defines the outcomes of this measurement, according to the Born rule: a measured state $\ket{\psi}$ gives the outcome $\alpha_m$ and gets projected onto the state $P_m \ket{\psi}/\sqrt{p(m})$ with probability $p(m) = \bra{\psi} P_m \ket{\psi} = \expval{P_m}_\psi$. The expectation value of the observable $O$ with respect to $\ket{\psi}$ is $\mathbb{E}_\psi[O] = \sum_m p(m) \alpha_m = \expval{O}_\psi$.

\subsection{The \textsc{raw-PQC} and \textsc{softmax-PQC} policies\label{sec:model-def}}

\begin{figure}[t]
\makebox[\linewidth]{\hspace{3em}
\Qcircuit @C=0.4em @R=.6em {
& &  & & U_\text{var}(\phis_0) & & & & & & U_\text{enc}(s,\lambdas_0)  & & &  &  & &    &   \\
& &  & & & && &&& &&&  & &  && &  &  & &    &   \\
\lstick{\ket{0}_0} & \gate{H}&\qw & \gate{R_z(\phi_{0,0})} & \gate{R_y(\phi_{0,2})}  &\qw &\ctrl{1} &\qw & \qw &\gate{R_y(\lambda_{0,0}s_0)} & \qw & \gate{R_z(\lambda_{0,2}s_0)} & \qw & \multigate{1}{U_\text{var}(\phis_1)} &\qw & \meter \gategroup{3}{4}{4}{7}{1.0em}{--}   \gategroup{3}{10}{4}{12}{1.0em}{--}   \\
\lstick{\ket{0}_1} & \gate{H}&\qw & \gate{R_z(\phi_{0,1})} & \gate{R_y(\phi_{0,3})}  &\qw &\ctrl{-1} & \qw & \qw  & \gate{R_y(\lambda_{0,1}s_1)} & \qw & \gate{R_z(\lambda_{0,3}s_1)} &\qw &\ghost{U_\text{var}(\phis_1)} & \qw & \meter\\
}}
\caption{\textbf{PQC architecture for $n=2$ qubits and depth $D_\text{enc}=1$. }This architecture is composed of alternating layers of encoding unitaries $U_\text{enc}(s,\lambdas_i)$ taking as input a state vector $s = (s_0, \ldots, s_{d-1})$ and scaling parameters $\lambdas_i$ (part of a vector $\bm{\lambda} \in \mathbb{R}^\abs{\lambdas}$ of dimension $\abs{\lambdas}$), and variational unitaries $U_\text{var}(\phis_i)$ taking as input rotation angles $\phis_i$ (part of a vector $\phis \in [0,2\pi]^\abs{\phis}$ of dimension $\abs{\phis}$).}
\label{fig:pqc-architecture}
\end{figure}
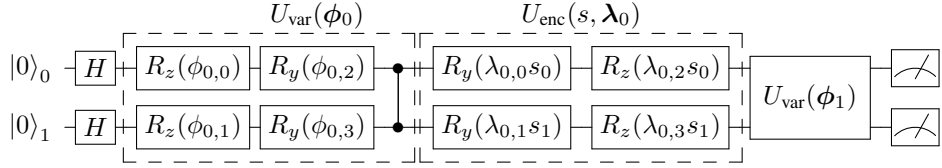

At the core of our parametrized quantum policies is a PQC defined by a unitary $U(s,\params)$ that acts on a fixed $n$-qubit state (e.g., $\ket{0^{\otimes n}}$). This unitary encodes an input state $s \in \mathbb{R}^d$ and is parametrized by a trainable vector $\params$. Although different choices of PQCs are possible, throughout our numerical experiments (Sec.\ \ref{sec:benchmark} and \ref{sec:PQC-env}), we consider so-called hardware-efficient PQCs \cite{kandala17} with an alternating-layered architecture \cite{perez20,schuld20}. This architecture is depicted in Fig.\ \ref{fig:pqc-architecture} and essentially consists in an alternation of $D_\text{enc}$ encoding unitaries $U_\text{enc}$ (composed of single-qubit rotations $R_z,R_y$) and $D_\text{enc}+1$ variational unitaries $U_\text{var}$ (composed of single-qubit rotations $R_z,R_y$ and entangling Ctrl-$Z$ gates $\CZ{scale=1.5}$ ).

For any given PQC, we define two families of policies, differing in how the final quantum states $\ket{\psi_{s,\params}} = U(s,\params)\ket{0^{\otimes n}}$ are used. In the \textsc{raw-PQC} model, we exploit the probabilistic nature of quantum measurements to define an RL policy. For $\abs{A}$ available actions to the RL agent, we partition $\mathcal{H}$ in $|A|$ disjoint subspaces (e.g., spanned by computational basis states) and associate a projection $P_a$ to each of these subspaces. The projective measurement associated to the observable $O = \sum_{a} a P_a$ then defines our \textsc{raw-PQC} policy $\policy(a|s) = \expval{P_a}_{s,\params}$. A limitation of this policy family however is that it does not have a directly adjustable greediness (i.e., a control parameter that makes the policy more peaked). This consideration arises naturally in an RL context where an agent's policy needs to shift from an exploratory behavior (i.e., close to uniform distribution) to a more exploitative behavior (i.e., a peaked distribution). To remedy this limitation, we define the \textsc{softmax-PQC} model, that applies an adjustable $\text{softmax}_\beta$ non-linear activation function on the expectation values $\expval{P_a}_{s,\params}$ measured on $\ket{\psi_{s,\params}}$. Since the softmax function normalizes any real-valued input, we can generalize the projections $P_a$ to be arbitrary Hermitian operators $O_a$. We also generalize these observables one step further by assigning them trainable weights. The two models are formally defined below.

\begin{definition}[\textsc{raw-} and \textsc{softmax-PQC}]\label{def:raw-softmax-PQC}
Given a PQC acting on $n$ qubits, taking as input a state $s \in \mathbb{R}^d$, rotation angles $\phis \in [0,2\pi]^\abs{\phis}$ and scaling parameters $\bm{\lambda} \in \mathbb{R}^\abs{\lambdas}$, such that its corresponding unitary $U(s,\phis,\lambdas)$ produces the quantum state $\ket{\psi_{s,\phis,\lambdas}} = U(s,\phis,\lambdas) \ket{0^{\otimes n}}$, we define its associated \textsc{raw-PQC} policy as:
	\begin{equation}
	\policy (a|s) = \expval{P_a}_{s,\params}
	\end{equation}
where $\expval{P_a}_{s,\params} = \ev{P_{a}}{\psi_{s,\phis,\lambdas}}$ is the expectation value of a projection $P_{a}$ associated to action $a$, such that $\sum_{a} P_a = I$ and $P_aP_{a'} = \delta_{a,a'}$. $\params = (\phis,\lambdas)$ constitute all of its trainable parameters.\\
Using the same PQC, we also define a \textsc{softmax-PQC} policy as:
	\begin{equation}\label{eq:softmax-PQC}
	\policy (a|s) = \frac{e^{\beta \expval{O_a}_{s,\params}}}{\sum_{a'} e^{\beta \expval{O_{a'}}_{s,\params}}}
	\end{equation}
where $\expval{O_a}_{s,\params} = \ev{\sum_i w_{a,i} H_{a,i}}{\psi_{s,\phis,\lambdas}}$ is the expectation value of the weighted Hermitian operators $H_{a,i}$ associated to action $a$, $\beta \in \mathbb{R}$ is an inverse-temperature parameter and $\params = (\phis,\lambdas,\weights)$.
\end{definition}
Note that we adopt here a very general definition for the observables $O_a$ of our \textsc{softmax-PQC} policies. As we discuss in more detail in Appendix \ref{sec:trainable-observables}, very expressive trainable observables can in some extreme cases take over all training of the PQC parameters $\phis, \lambdas$ and render the role of the PQC in learning trivial. However, in practice, as well as in our numerical experiments, we only consider very restricted observables $O_a = \sum_{i} w_{a,i} H_{a,i}$, where $H_{a,i}$ are (tensor products of) Pauli matrices or high-rank projections on computational basis states, which do not allow for these extreme scenarios.

In our PQC construction, we include trainable \emph{scaling parameters} $\lambdas$, used in every encoding gate to re-scale its input components. This modification to the standard data encoding in PQCs comes in light of recent considerations on the structure of PQC functions \cite{schuld19}. These additional parameters allow to represent functions with a wider and richer spectrum of frequencies, and hence provide shallow PQCs with more expressive power.

\subsection{Learning algorithm}

In order to analyze the properties of our PQC policies without the interference of other learning mechanisms \cite{weng18}, we train these policies using the basic Monte Carlo policy gradient algorithm \textsc{REINFORCE} \cite{sutton98,williams92} (see Alg.\ \ref{alg}). This algorithm consists in evaluating Monte Carlo estimates of the value function $V_{\policy}(s_0) = \mathbb{E}_{\policy}\left[\sum_{t=0}^{H-1}\gamma^t r_t \right]$, $\gamma \in [0,1]$, using batches of interactions with the environment, and updating the policy parameters $\params$ via a gradient ascent on $V_{\policy}(s_0)$. The resulting updates (see line $8$ of Alg.\ \ref{alg}) involve the gradient of the log-policy $\nabla_{\params} \log{\policy(a|s)}$, which we therefore need to compute for our policies. We describe this computation in the following lemma.

\begin{lemma}\label{lem:gradient-exp}
	Given a \textsc{softmax-PQC} policy $\policy$, the gradient of its logarithm is given by:
	\begin{equation}
		\nabla_{\params} \log{\policy(a|s)} = \beta \Big(\nabla_{\params}\expval{O_a}_{s,\params} - \sum\nolimits_{a'} \policy(a'|s) \nabla_{\params}\expval{O_{a'}}_{s,\params}\Big).
	\end{equation} 
	Partial derivatives with respect to observable weights are trivially given by $\partial_{w_{a,i}}\langle O_a\rangle_{s,\params} = \ev{H_{a,i}}{\psi_{s,\phis,\lambdas}}$ (see Def.\ \ref{def:raw-softmax-PQC}), while derivatives with respect to rotation angles $\partial_{\phi_i}\langle O_a\rangle_{s,\params}$ and scaling parameters\footnote{Note that the parameters $\lambdas$ do not act as rotation angles. To compute the derivatives $\partial_{\lambda_{i,j}} \langle O_a\rangle_{s,\params}$, one should compute derivatives w.r.t.\ $s_j\lambda_{i,j}$ instead and apply the chain rule: $\partial_{\lambda_{i,j}} \langle O_a\rangle_{s,\params} = s_j \partial_{s_j \lambda_{i,j}} \langle O_a\rangle_{s,\params}$.} $\partial_{\lambda_i}\langle O_a\rangle_{s,\params}$ can be estimated via the parameter-shift rule \cite{mitarai18,schuld19}:
	\begin{equation}\label{eq:psr}
		\partial_i \expval{O_a}_{s,\params} = \frac{1}{2}\big(\expval{O_a}_{s,\params+\frac{\pi}{2} \bm{e_i}} - \expval{O_a}_{s,\params-\frac{\pi}{2} \bm{e_i}}\big),
	\end{equation}
	i.e., using the difference of two expectation values $\expval{O_a}_{s,\params'}$ with a single angle shifted by $\pm\frac{\pi}{2}$.\\
	For a \textsc{raw-PQC} policy $\policy$, we have instead:
	\begin{equation}\label{eq:raw-PQC-gradient}
		\nabla_{\params} \log{\policy(a|s)} = \nabla_{\params}\expval{P_a}_{s,\params} / \expval{P_a}_{s,\params} 
	\end{equation} 
	where the partial derivatives $\partial_{\phi_i}\langle P_a\rangle_{s,\params}$ and $\partial_{\lambda_i}\langle P_a\rangle_{s,\params}$ can be estimated similarly to above.
\end{lemma}

In some of our environments, we additionally rely on a linear value-function baseline to reduce the variance of the Monte Carlo estimates \cite{greensmith04}. We choose it to be identical to that of Ref.\ \cite{duan16}.

\begin{figure}[t]
  \centering
\begin{minipage}{.8\linewidth}
\DecMargin{0.2em}
\begin{algorithm2e}[H]
\caption{\textsc{REINFORCE} with PQC policies and value-function baselines}\label{alg}
 \KwIn{a PQC policy $\policy$ from Def.\ \ref{def:raw-softmax-PQC}; a value-function approximator $\approxV$}
 Initialize parameters $\params$ and $\paramsV$\;
 \While{True}{
  Generate $N$ episodes $\left\{(s_0, a_0, r_1, \ldots, s_{H-1}, a_{H-1}, r_H) \right\}_i$ following $\policy$\;
  \For{episode $i$ in batch}{
    Compute the returns $G_{i,t} \leftarrow \sum_{t'=1}^{H-t} \gamma^{t'} r^{(i)}_{t+t'}$\;
    Compute the gradients $\nabla_{\params}\log{\policy(a^{(i)}_t|s^{(i)}_t)}$ using Lemma \ref{lem:gradient-exp}\;}
  Fit $\big\{\approxV(s_t^{(i)})\big\}_{i,t}$ to the returns $\left\{G_{i,t}\right\}_{i,t}$\;
  Compute $\Delta\params  = \dfrac{1}{N}\sum\limits_{i=1}^{N} \sum\limits_{t=0}^{H-1} \nabla_{\params}\log{\policy(a^{(i)}_t|s^{(i)}_t)} \left( G_{i,t} - \approxV(s_t^{(i)}) \right);$\BlankLine
  Update $\params \leftarrow \params + \alpha\Delta\params$\;
 }
\end{algorithm2e}
\end{minipage}
\end{figure}

\subsection{Efficient policy sampling and policy-gradient evaluation}

A natural consideration when it comes to the implementation of our PQC policies is whether one can efficiently (in the number of executions of the PQC on a quantum computer) sample and train them.

By design, sampling from our \textsc{raw-PQC} policies can be done with a single execution (and measurement) of the PQC: the projective measurement corresponding to the observable $O = \sum_{a} a P_a$ naturally samples a basis state associated to action $a$ with probability $\expval{P_a}_{s,\params}$. However, as Eq.~(\ref{eq:raw-PQC-gradient}) indicates, in order to train these policies using REINFORCE, one is nonetheless required to estimate the expectation values $\expval{P_a}_{s,\params}$, along with the gradients $\nabla_{\params}\expval{P_a}_{s,\params}$. Fortunately, these quantities can be estimated efficiently up to some additive error $\varepsilon$, using only $\mathcal{O}(\varepsilon^{-2})$ repeated executions and measurements on a quantum computer.

In the case of our \textsc{softmax-PQC} policies, it is less clear whether similar noisy estimates $\langle\widetilde{O_a}\rangle_{s,\params}$ of the expectation values $\expval{O_a}_{s,\params}$ are sufficient to evaluate policies of the form of Eq.~(\ref{eq:softmax-PQC}). We show however that, using these noisy estimates, we can compute a policy $\widetilde{\pi}_{\bm{\theta}}$ that produces samples close to that of the true policy $\policy$. We state our result formally in the following lemma, proven in Appendix~\ref{sec:proofs-softmax-pqc}.

\begin{restatable}[]{lemma}{tvpolicy}
\label{lemma:tv-policy}
	For a \textsc{softmax-PQC} policy $\policy$ defined by a unitary $U(s,\params)$ and observables $O_a$, call $\langle\widetilde{O_a}\rangle_{s,\params}$ approximations of the true expectation values $\expval{O_a}_{s,\params}$ with at most $\varepsilon$ additive error. Then the approximate policy $\widetilde{\pi}_{\bm{\theta}} = \textnormal{softmax}_\beta(\langle\widetilde{O_a}\rangle_{s,\params})$ has total variation distance $\mathcal{O}(\beta\varepsilon)$ to $\policy = \textnormal{softmax}_\beta(\expval{O_a}_{s,\params})$. Since expectation values can be efficiently estimated to additive error on a quantum computer, this implies efficient approximate sampling from $\policy$. 
\end{restatable}
We also obtain a similar result for the log-policy gradient of \textsc{softmax-PQC}s (see Lemma \ref{lem:gradient-exp}), that we show can be efficiently estimated to additive error in $\ell_\infty$-norm (see Appendix \ref{sec:proofs-softmax-pqc} for a proof).

\section{Performance comparison in benchmarking environments\label{sec:benchmark}}

In the previous section, we have introduced our quantum policies and described several of our design choices. We defined the \textsc{raw-PQC} and \textsc{softmax-PQC} models and introduced two original features for PQCs: trainable observables at their output and trainable scaling parameters for their input. In this section, we evaluate the influence of these design choices on learning performance through numerical simulations. We consider three classical benchmarking environments from the OpenAI Gym library \cite{brockman16}: CartPole, MountainCar and Acrobot. All three have continuous state spaces and discrete action spaces (see Appendix \ref{sec:env-spec-hyper} for their specifications). Moreover, simple NN-policies, as well as simple closed-form policies, are known to perform very well in these environments \cite{gym20}, which makes them an excellent test-bed to benchmark PQC policies.

\subsection{\textsc{raw-PQC} v.s.\ \textsc{softmax-PQC}\label{sec:raw-vs-softmax-sim}}

\begin{figure*}
	\subfloat{\label{fig:cartpole}\includegraphics[width=0.33\linewidth]{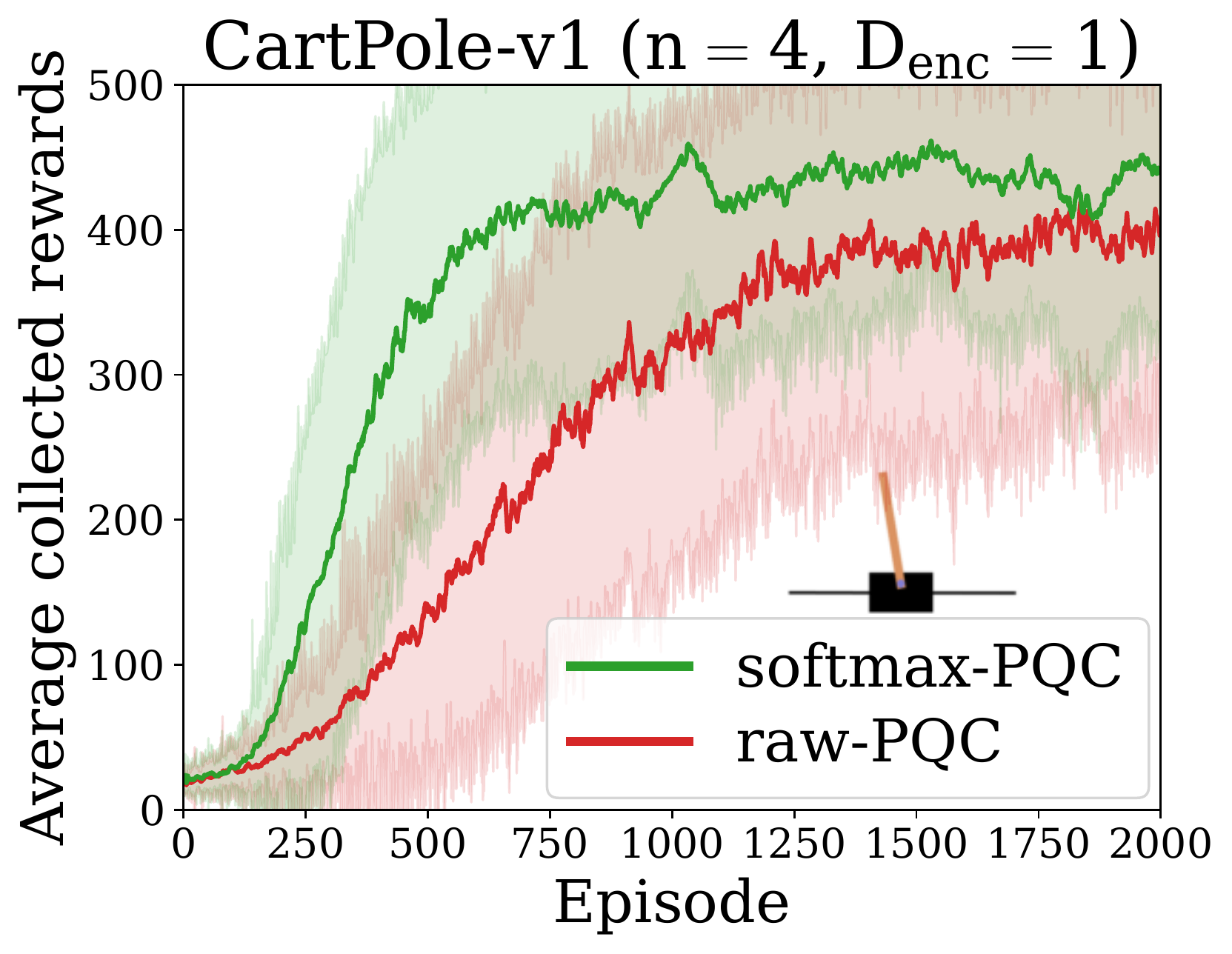}}\hspace{0em}%
	\subfloat{\label{fig:mountaincar}\includegraphics[width=0.328\linewidth]{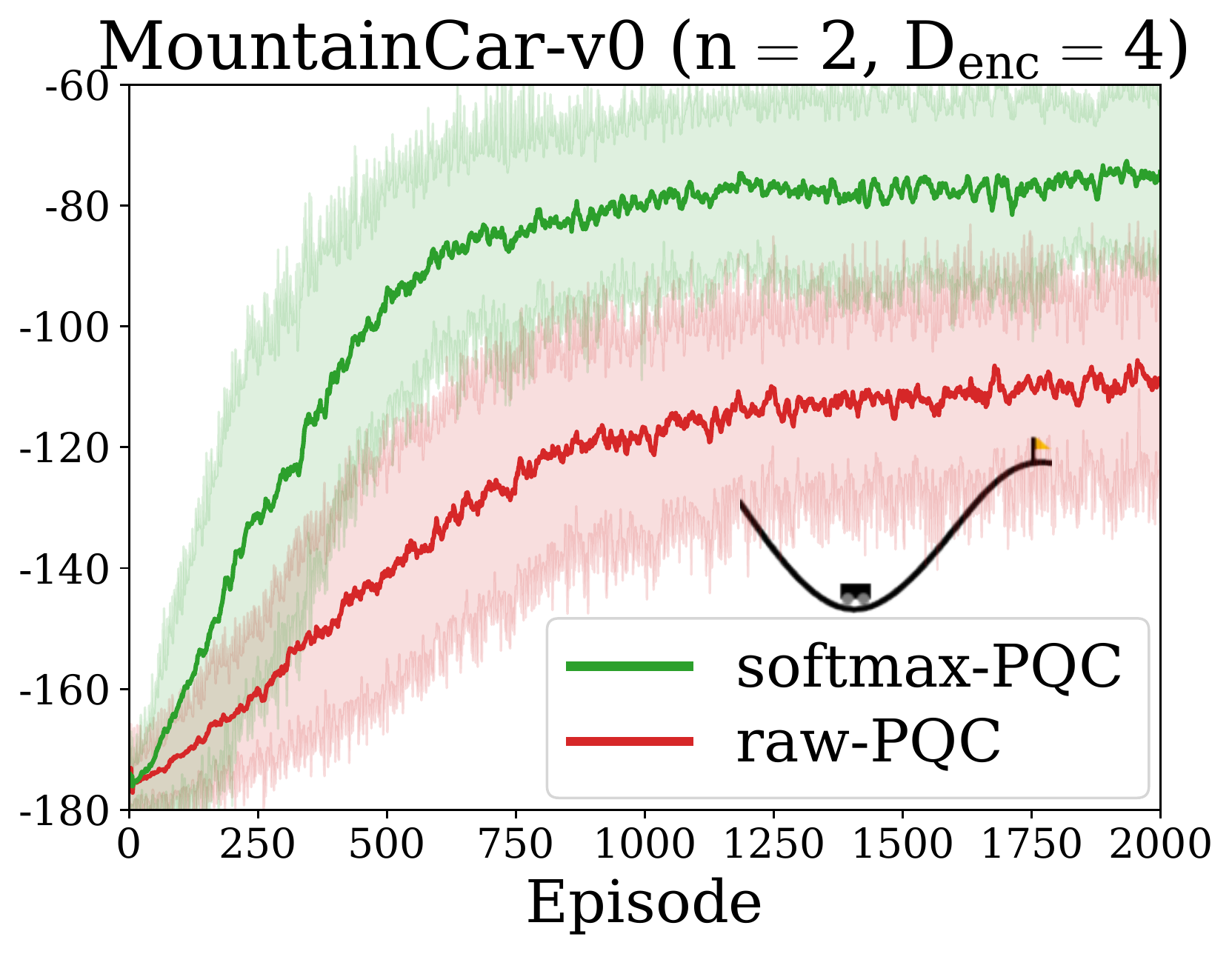}}\hspace{0em}%
	\subfloat{\label{fig:acrobot}\includegraphics[width=0.328\linewidth]{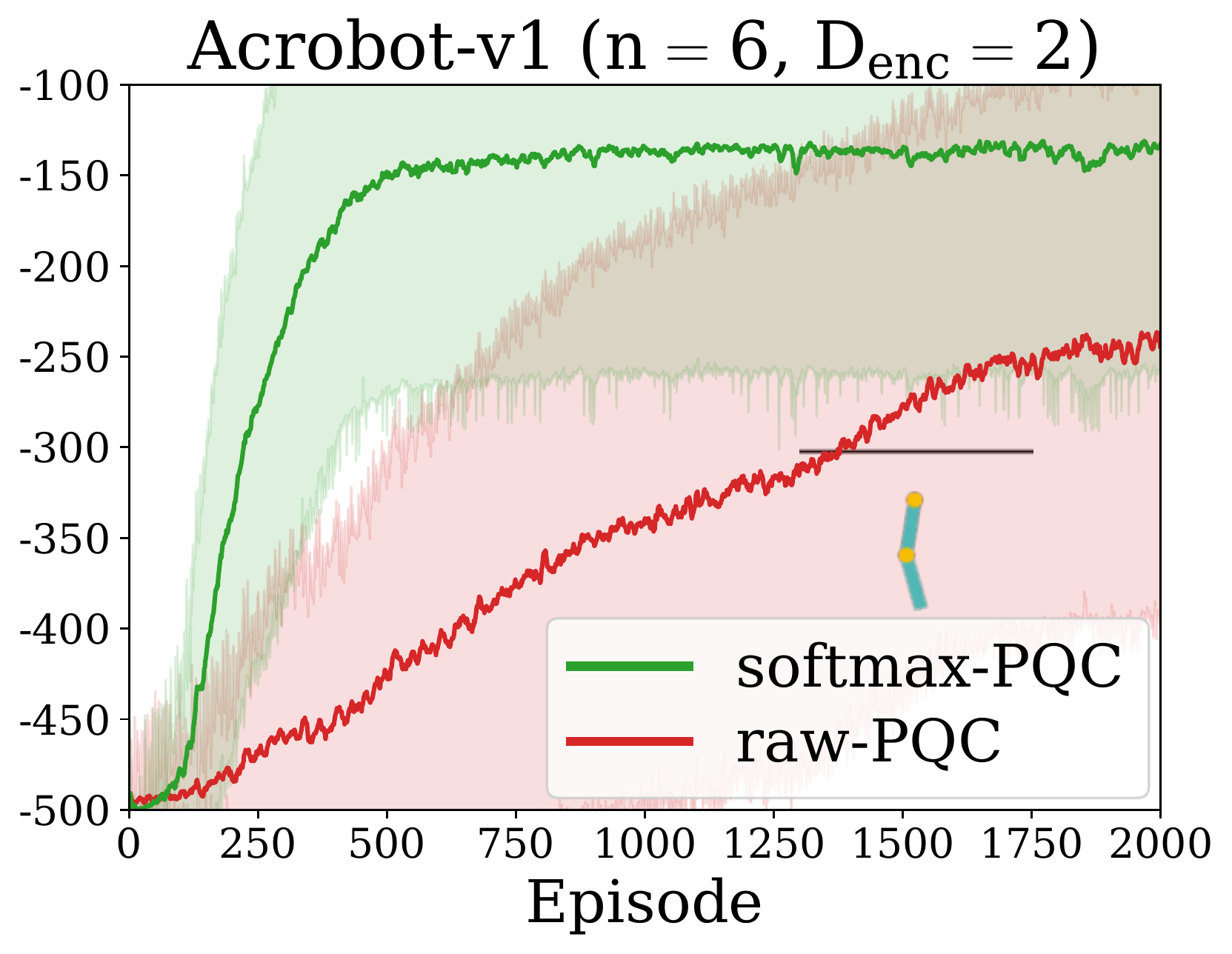}}\hspace{0em}%
  \caption{\textbf{Numerical evidence of the advantage of \textsc{softmax-PQC} over \textsc{raw-PQC} in benchmarking environments. }The learning curves ($20$ agents per curve) of randomly-initialized \textsc{softmax-PQC} agents (green curves) and \textsc{raw-PQC} agents (red curves) in OpenAI Gym environments: CartPole-v1, MountainCar-v0, and Acrobot-v1. Each curve is temporally averaged with a time window of $10$ episodes. All agents have been trained using the \textsc{REINFORCE} algorithm (see Alg.\ \ref{alg}), with value-function baselines for the MountainCar and Acrobot environments.}
  \label{fig:softmax-vs-raw}
\end{figure*}
\vspace{-0.5em}
\begin{figure*}
	\subfloat{\label{fig:cartpole-softmax}\includegraphics[width=0.33\linewidth, valign=c]{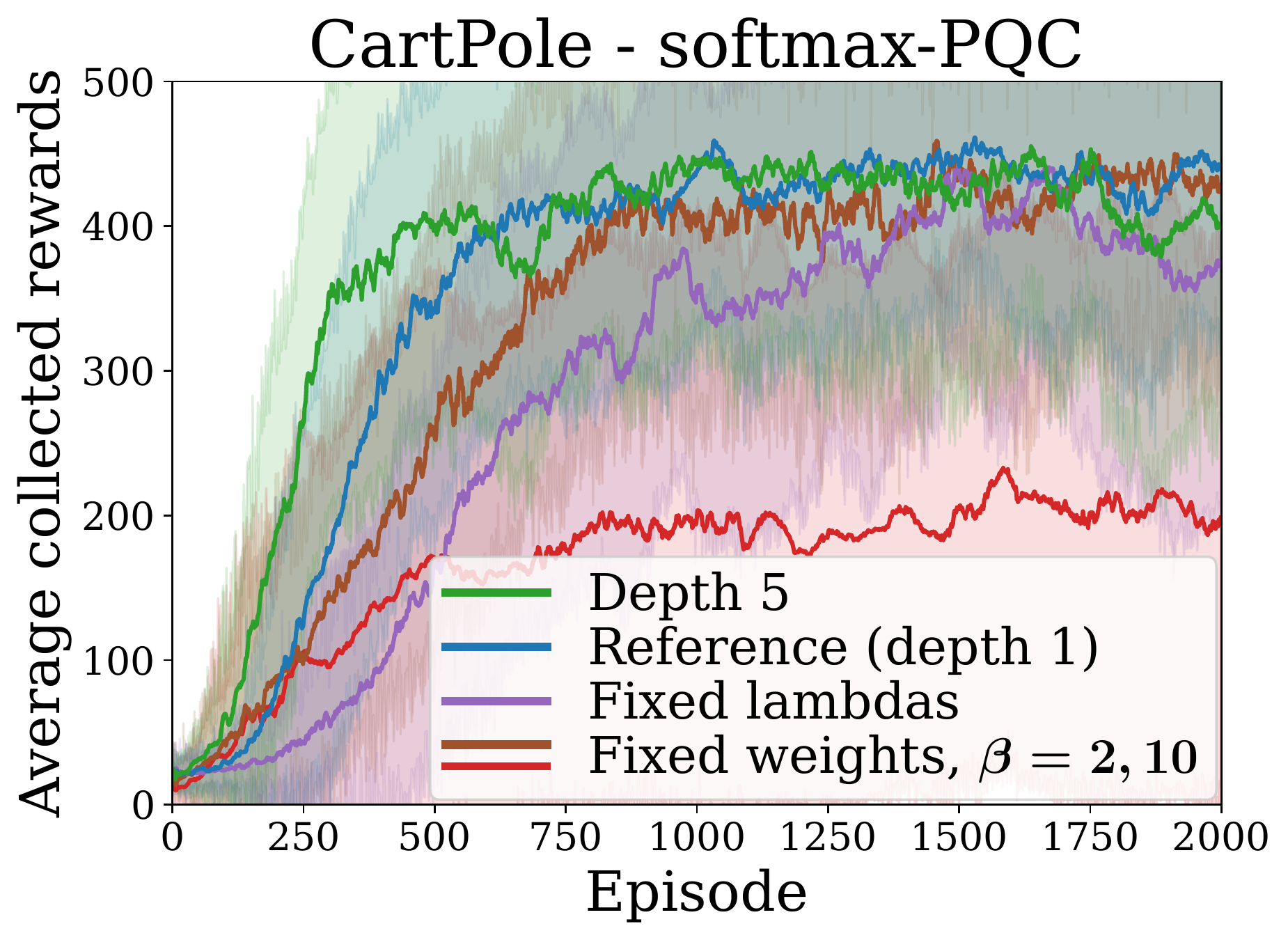}}\hspace{0em}%
	\subfloat{\label{fig:mountaincar-softmax}\includegraphics[width=0.33\linewidth, valign=c]{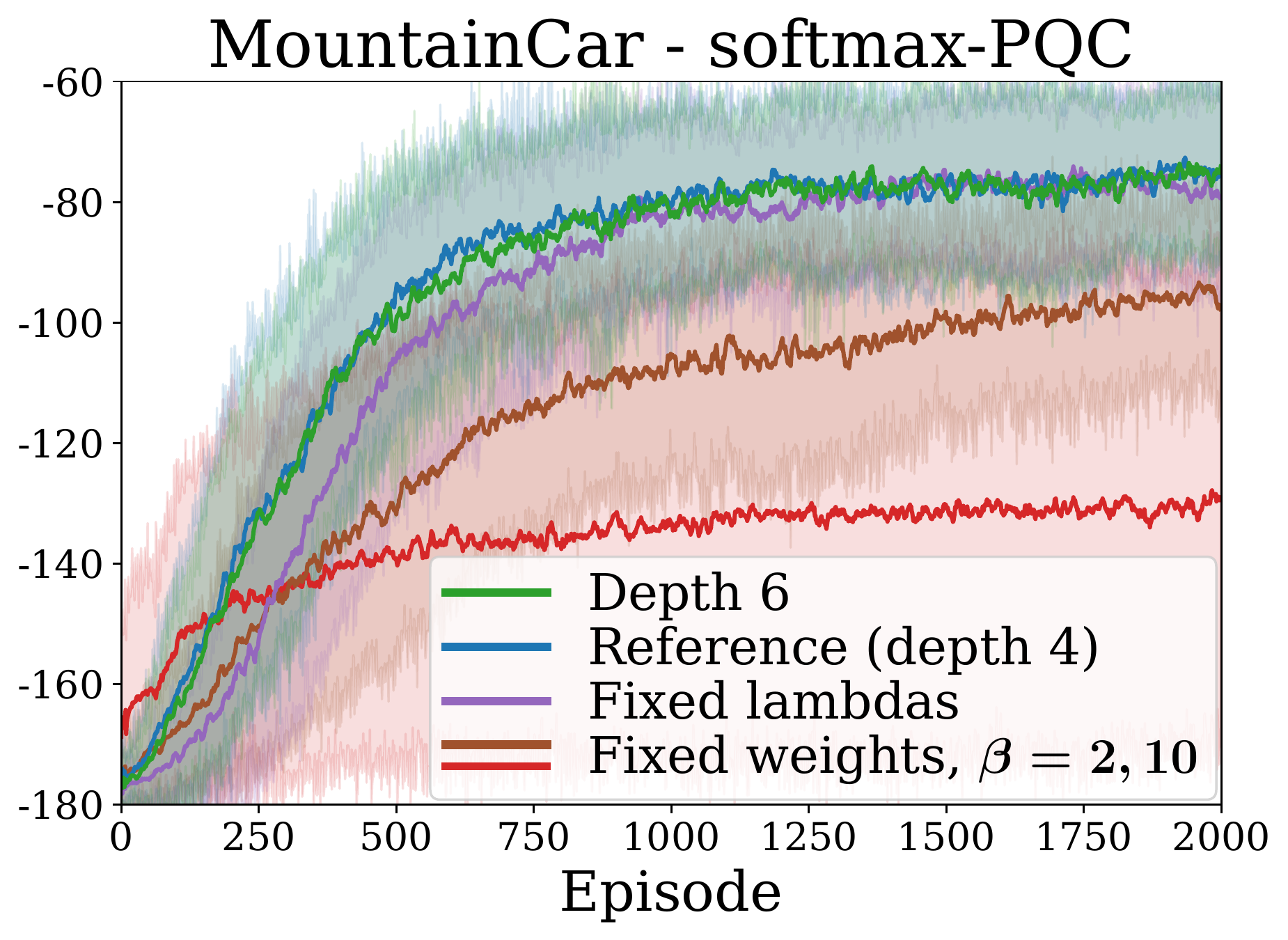}}\hspace{0em}%
	\subfloat{\label{fig:acrobot-softmax}\includegraphics[width=0.33\linewidth, valign=c]{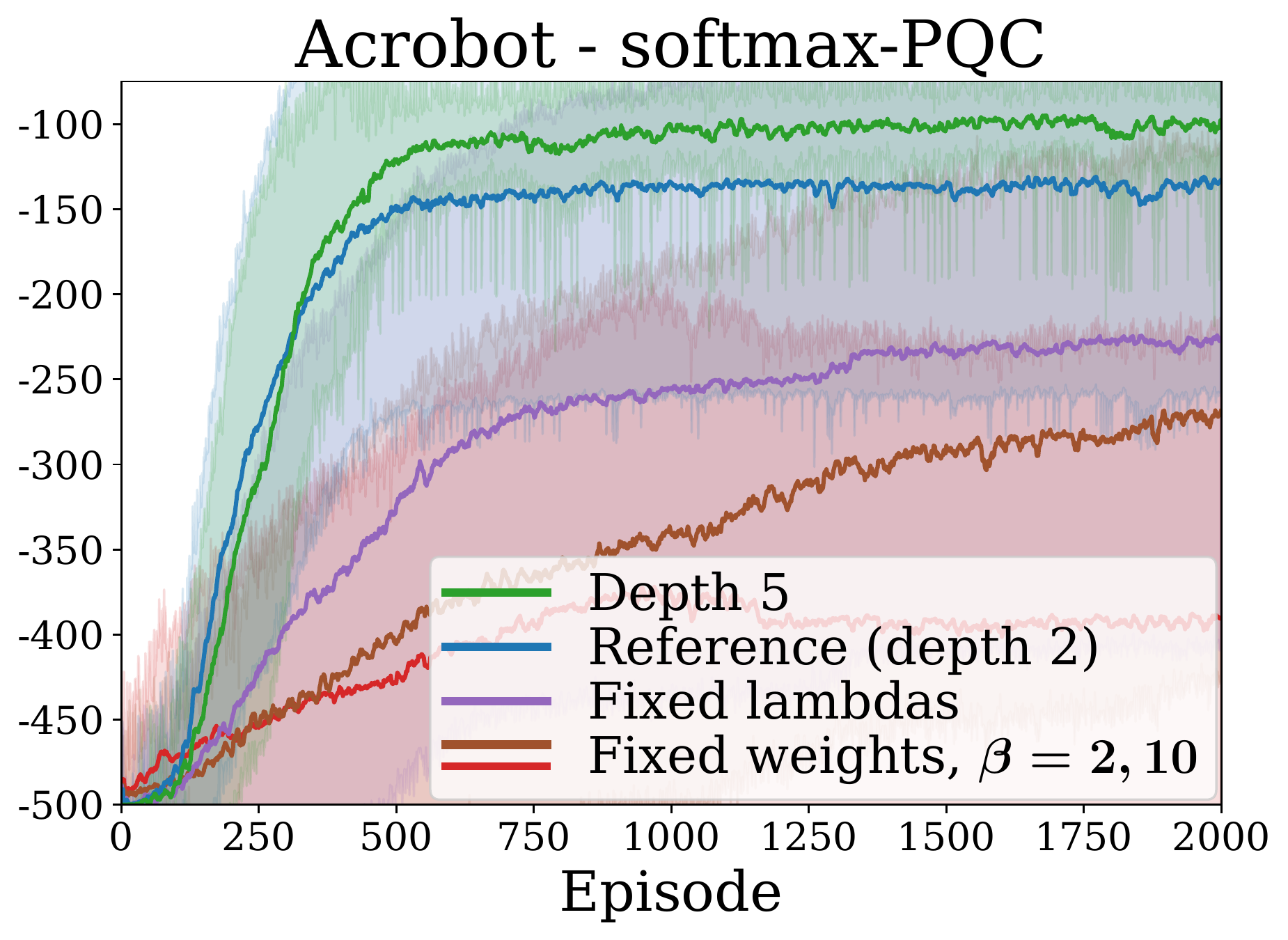}}\hspace{0em}
  \caption{\textbf{Influence of the model architecture for \textsc{softmax-PQC} agents. }The blue curves in each plot correspond to the learning curves from Fig.\ \ref{fig:softmax-vs-raw} and are taken as a reference. Other curves highlight the influence of individual hyperparameters. For \textsc{raw-PQC} agents, see Appendix \ref{sec:PQC-DNN}.}
  \label{fig:architecture-compare}
\end{figure*}\vspace{-0.5em}

In our first set of experiments, presented in Fig.\ \ref{fig:softmax-vs-raw}, we evaluate the general performance of our proposed policies. The aim of these experiments is twofold: first, to showcase that quantum policies based on shallow PQCs and acting on very few qubits can be trained to good performance in our selected environments; second, to test the advantage of \textsc{softmax-PQC} policies over \textsc{raw-PQC} policies that we conjectured in the Sec.\ \ref{sec:model-def}. To assess these claims, we take a similar approach for each of our benchmarking environments, in which we evaluate the average learning performance of $20$ \textsc{raw-PQC} and $20$ \textsc{softmax-PQC} agents. Apart from the PQC depth, the shared hyperparameters of these two models were jointly picked as to give the best overall performance for both; the hyperparameters specific to each model were optimized independently. As for the PQC depth $D_\text{enc}$, the latter was chosen as the minimum depth for which near-optimal performance was observed for either model. The simulation results confirm both our hypotheses: quantum policies can achieve good performance on the three benchmarking tasks that we consider, and we can see a clear separation between the performance of \textsc{softmax-PQC} and \textsc{raw-PQC} agents.

\subsection{Influence of architectural choices\label{sec:architecture-sim}}

The results of the previous subsection however do not indicate whether other design choices we have made in Sec.\ \ref{sec:model-def} had an influence on the performance of our quantum agents. To address this, we run a second set of experiments, presented in Fig.\ \ref{fig:architecture-compare}. In these simulations, we evaluate the average performance of our \textsc{softmax-PQC} agents after modifying one of three design choices: we either increment the depth of the PQC (until no significant increase in performance is observed), fix the input-scaling parameters $\lambdas$ to $\bm{1}$, or fix the observable weights $\weights$ to $\bm{1}$. By comparing the performance of these agents with that of the agents from Fig.\ \ref{fig:softmax-vs-raw}, we can make the following observations:
\begin{itemize}[leftmargin=4mm]
\item \textbf{Influence of depth:} Increasing the depth of the PQC generally improves (not strictly) the performance of the agents. Note that the maximum depth we tested was $D_\text{enc}=10$.
\item \textbf{Influence of scaling parameters $\lambdas$:} We observe that training these scaling parameters in general benefits the learning performance of our PQC policies, likely due to their increased expressivity.
\item \textbf{Influence of trainable observable weights $\weights$:} our final consideration relates to the importance of having a policy with ``trainable greediness'' in RL scenarios. For this, we consider \textsc{softmax-PQC} agents with fixed observables $\beta O_a$ throughout training. We observe that this has the general effect of decreasing the performance and/or the speed of convergence of the agents. We also see that policies with fixed high $\beta$ (or equivalently, a large observable norm $\beta\norm{O_a}$) tend to have a poor learning performance, likely due to their lack of exploration in the RL environments.
\end{itemize}
Finally, note that all the numerical simulations performed here did not include any source of noise in the PQC evaluations. It would be an interesting research direction to assess the influence of (simulated or hardware-induced) noise on the learning performance of PQC agents.

\section{Quantum advantage of PQC agents in RL environments\label{sec:quantum-advantage}}

The proof-of-concept experiments of the previous section show that our PQC agents can learn in basic classical environments, where they achieve comparable performance to standard DNN policies. This observation naturally raises the question of whether there exist RL environments where PQC policies can provide a learning advantage over standard classical policies. In this section, we answer this question in the affirmative by constructing: a) environments with a provable separation in learning performance between quantum and any classical (polynomial-time) learners, and b) environments where our PQC policies of Sec.\ \ref{sec:pqc-policies} show an empirical learning advantage over standard DNN policies.

\subsection{Quantum advantage of PQC policies over any classical learner\label{sec:DLP-RL}}

In this subsection, we construct RL environments with theoretical guarantees of separation between quantum and classical learning agents. These constructions are predominantly based on the recent work of Liu \emph{et al.} \cite{liu20}, which defines a classification task out of the discrete logarithm problem (DLP), i.e., the problem solved in the seminal work of Shor \cite{shor99}. 
In broad strokes, this task can be viewed as an encryption of an easy-to-learn problem. For an ``un-encrypted'' version, one defines a labeling $f_s$ of integers between $0$ and $p-2$ (for a large prime $p$), where the integers are labeled positively if and only if they lie in the segment $[s, s+(p-3)/2]$ ($\text{mod}\ p-1$). Since this labeling is linearly separable, the concept class $\{f_s\}_s$ is then easy to learn. To make it hard, the input integers $x$ (now between $1$ and $p-1$) are first encrypted using modular exponentiation, i.e., the secure operation performed in the Diffie–Hellman key exchange protocol. In the encrypted problem, the logarithm of the input integer $\log_g(x)$ (for a generator $g$ of $\mathbbmss{Z}_p^*$, see Appendix \ref{sec:DLP-task}) hence determines the label of $x$. Without the ability to decrypt by solving DLP, which is widely believed to be classically intractable, the numbers appear randomly labeled. Moreover, Liu \emph{et al.}\ show that achieving non-trivial labeling accuracy $1/2 + 1/\text{poly}(n)$ (for $n=\log(p)$, i.e., slightly better than random guessing) with a classical polynomial-time algorithm using $\text{poly}(n)$ examples would lead to an efficient classical algorithm that solves DLP \cite{liu20}. In contrast, the same authors construct a family of quantum learners based on Shor's algorithm, that can achieve a labeling accuracy larger than $0.99$ with high probability.

\paragraph{SL-DLP}
Our objective is to show that analogous separations between classical and quantum learners can be established for RL environments, in terms of their attainable value functions. We start by pointing out that supervised learning (SL) tasks (and so the classification problem of Liu \emph{et al.}) can be trivially embedded into RL environments \cite{dunjko17b}: for a given concept $f_s$, the states $x$ are datapoints, an action $a$ is an agent's guess on the label of $x$, an immediate reward specifies if it was correct (i.e., $f_s(x)=a$), and subsequent states are chosen uniformly at random. In such settings, the value function is trivially related to the testing accuracy of the SL problem, yielding a direct reduction of the separation result of Liu \emph{et al.} \cite{liu20} to an RL setting. We call this family of environments SL-DLP.

\paragraph{Cliffwalk-DLP}
In the SL-DLP construction, we made the environment fully random in order to simulate the process of obtaining i.i.d.\ samples in an SL setting. It is an interesting question whether similar results can be obtained for environments that are less random, and endowed with temporal structure, which is characteristic of RL. In our second family of environments (Cliffwalk-DLP), we supplement the SL-DLP construction with next-state transitions inspired by the textbook ``cliff walking'' environment of Sutton \& Barto \cite{sutton98}: all states are ordered in a chain and some actions of the agent can lead to immediate episode termination. We keep however stochasticity in the environment  by allowing next states to be uniformly sampled, with a certain probability $\delta$ (common in RL to ensure that an agent is not simply memorizing a correct sequence of actions). This allows us to show that, as long as sufficient randomness is provided, we still have a simple classical-quantum separation. 

\paragraph{Deterministic-DLP}
In the two families constructed above, each environment instance provided the randomness needed for a reduction from the SL problem. This brings us to the question of whether separations are also possible for fully deterministic environments. In this case, it is clear that for any given environment, there exists an efficient classical agent which performs perfectly over any polynomial horizon (a lookup-table will do). However, we show in our third family of environments (Deterministic-DLP) that a separation can still be attained by moving the randomness to the choice of the environment itself: assuming an efficient classical agent is successful in most of exponentially-many randomly generated (but otherwise deterministic) environments, implies the existence of a classical efficient algorithm for DLP.

We summarize our results in the following theorem, detailed and proven in Appendices \ref{sec:proof-thm-separations} through \ref{sec:proof-deterministic-DLP}.

\begin{theorem}\label{thm:separations-DLP}
There exist families of reinforcement learning environments which are: i) fully random (i.e., subsequent states are independent from the previous state and action); ii) partially random (i.e., the previous moves determine subsequent states, except with a probability $\delta$ at least $0.86$ where they are chosen uniformly at random), and iii) fully deterministic; such that there exists a separation in the value functions achievable by a given quantum polynomial-time agent and any classical polynomial-time agent. Specifically, the value of the initial state for the quantum agent $V_q(s_0)$ is $\varepsilon-$close to the optimal value function (for a chosen $\varepsilon$, and with probability above 2/3). Further, if there exists a classical efficient learning agent that achieves a value $V_c(s_0)$ better than $V_\textnormal{rand}(s_0)+\varepsilon'$ (for a chosen $\varepsilon'$, and with probability above 0.845), then there exists a classical efficient algorithm to solve DLP. Finally, we have $V_q(s_0)-V_c(s_0)$ larger than some constant, which depends on the details of the environment.
\end{theorem}
The remaining point we need to address here is that the learning agents of Liu \emph{et al.} do not rely on PQCs but rather support vector machines (SVMs) based on quantum kernels \cite{havlivcek19,schuld19b}. Nonetheless, using a connection between these quantum SVMs and PQCs \cite{schuld19b}, we construct PQC policies which are as powerful in solving the DLP environments as the agents of Liu \emph{et al.} (even under similar noise considerations). We state our result in the following informal theorem, that we re-state formally, along with the details of our construction in Appendices \ref{sec:PQC-agent-DLP} and \ref{sec:training-SL-DLP}.
\begin{theorem}[informal version]
Using a training set of size polynomial in $n = \log(p)$ and a number of (noisy) quantum circuit evaluations also polynomial in $n$, we can train a PQC classifier on the DLP task of Liu \emph{et al.} of size $n$ that achieves a testing accuracy arbitrarily close to optimal, with high probability. This PQC classifier can in turn be used to construct close-to-optimal quantum agents in our DLP environments, as prescribed by Theorem \ref{thm:separations-DLP}.
\end{theorem}

\subsection{Quantum advantage of PQC policies over DNN policies\label{sec:PQC-env}}

While the DLP environments establish a proof of the learning advantage PQC policies can have in theory, these environments remain extremely contrived and artificial. They are based on algebraic properties that agents must explicitly decrypt in order to perform well. Instead, we would like to consider environments that are less tailored to a specific decryption function, which would allow more general agents to learn. To do this, we take inspiration from the work of Havlí\v{c}ek \emph{et al.} \cite{havlivcek19}, who, in order to test their PQC classifiers, define a learning task generated by similar quantum circuits.

\subsubsection{PQC-generated environments}

We generate our RL environments out of random \textsc{raw-PQC}s. To do so, we start by uniformly sampling a \textsc{raw-PQC} that uses the alternating-layer architecture of Fig.\ \ref{fig:pqc-architecture} for $n=2$ qubits and depth $D_\text{enc} = 4$. We use this \textsc{raw-PQC} to generate a labeling function $f(s)$ by assigning a label $+1$ to the datapoints $s$ in $[0,2\pi]^2$ for which $\expval{ZZ}_{s,\params} \geq 0$ and a label $-1$ otherwise. We create a dataset $S$ of $10$ datapoints per label by uniformly sampling points in $[0,2\pi]^2$ for which $|\expval{ZZ}_{s,\params}| \geq \frac{\Delta}{2} = 0.15$. This dataset allows us to define two RL environments, similar to the SL-DLP and Cliffwalk-DLP environments of Sec.\ \ref{sec:DLP-RL}:

\begin{figure*}
	\subfloat{\label{fig:pqc-env}\includegraphics[width=0.315\linewidth, valign=c]{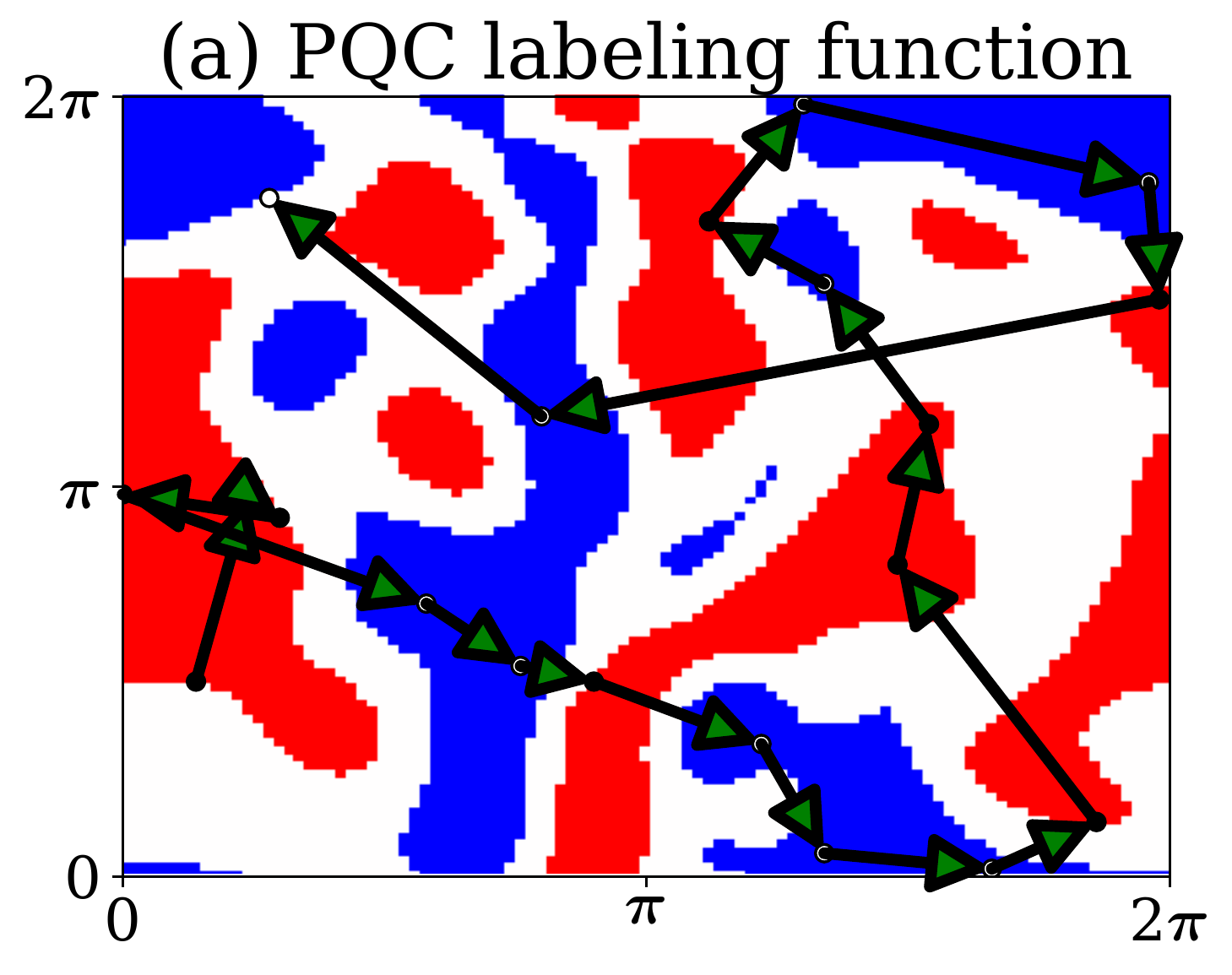}}\hspace{0em}%
	\subfloat{\label{fig:supervised}\includegraphics[width=0.34\linewidth, valign=c]{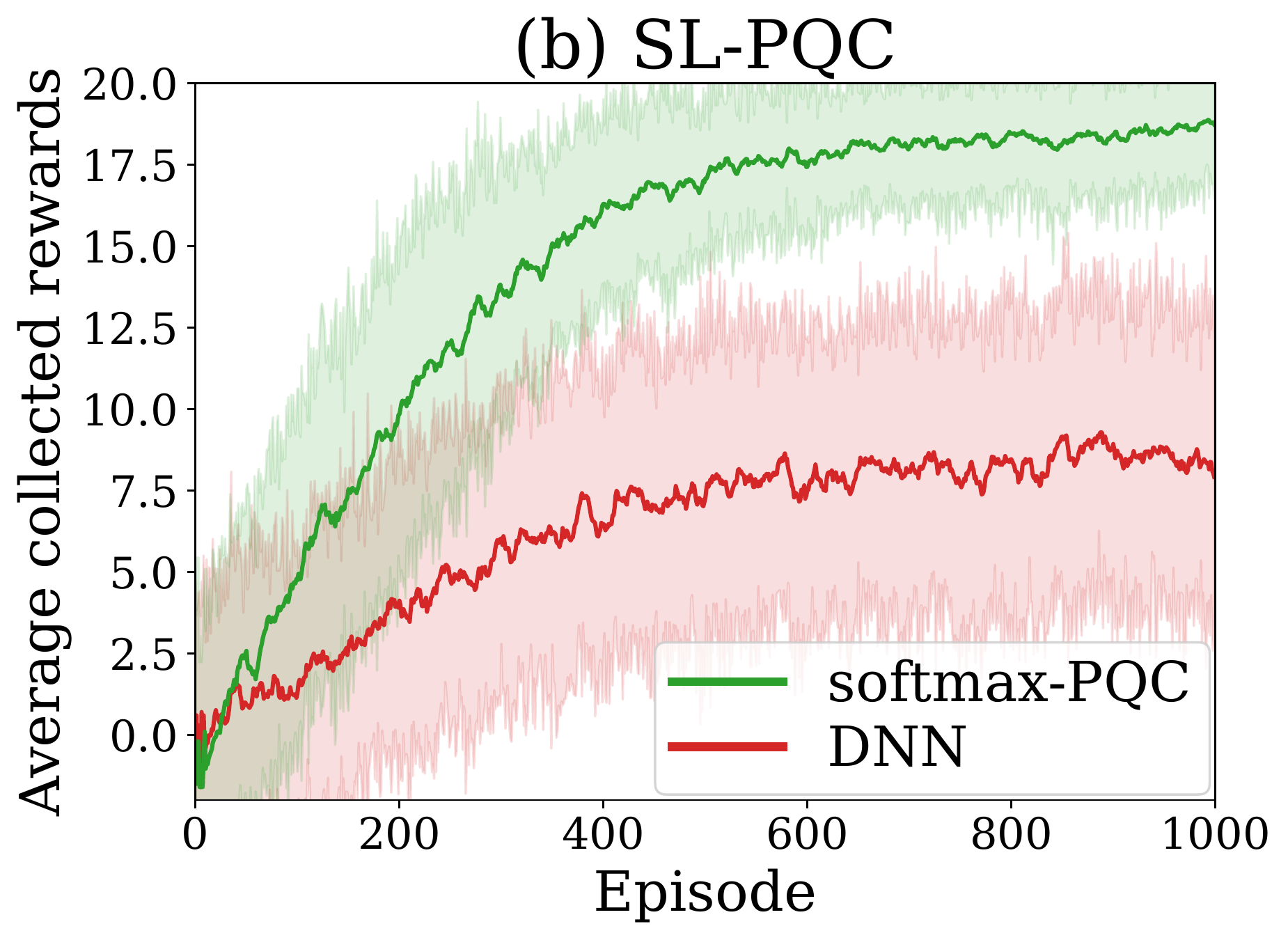}}\hspace{0em}%
	\subfloat{\label{fig:cliffwalk}\includegraphics[width=0.34\linewidth, valign=c]{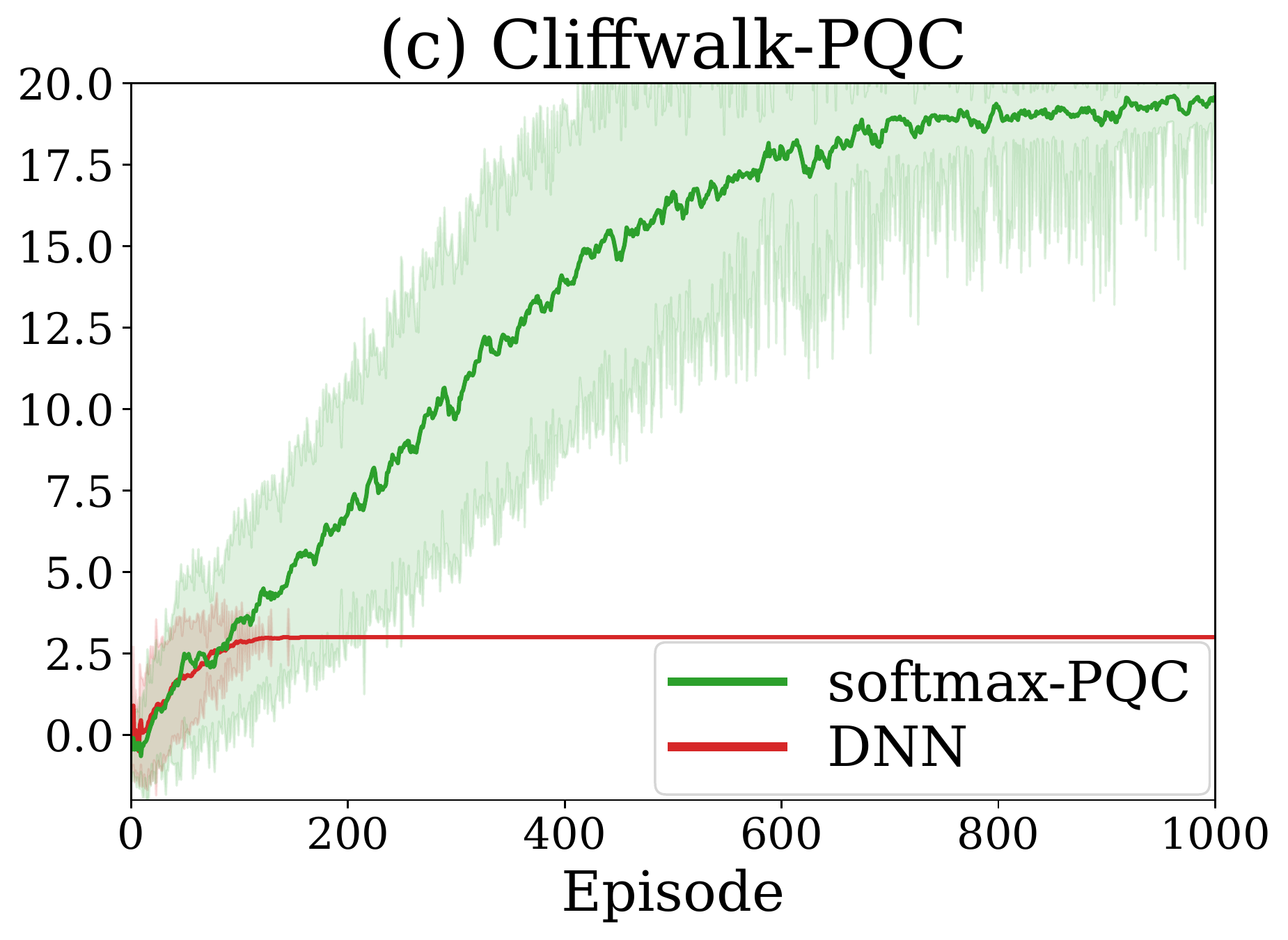}}\hspace{0em}%
  \caption{\textbf{Numerical evidence of the advantage of PQC policies over DNN policies in PQC-generated environments. }(a) Labeling function and training data used for both RL environments. The data labels (red for $+1$ label and blue for $-1$ label) are generated using a \textsc{raw-PQC} of depth $D_\text{enc} = 4$ with a margin $\Delta = 0.3$ (white areas). The training samples are uniformly sampled from the blue and red regions, and arrows indicate the rewarded path of the cliffwalk environment.  (b) and (c) The learning curves ($20$ agents per curve) of randomly-initialized \textsc{softmax-PQC} agents and DNN agents in RL environments where input states are (b) uniformly sampled from the dataset and (c) follow cliffwalk dynamics. Each curve is temporally averaged with a time window of $10$ episodes.}
  \label{fig:pqc-vs-nn}
\end{figure*}

\begin{itemize}[leftmargin=4mm]
\item \textbf{SL-PQC:} this degenerate RL environment encodes a classification task in an episodic RL environment: at each interaction step of a $20$-step episode, a sample state $s$ is uniformly sampled from the dataset $S$, the agent assigns a label $a=\pm1$ to it and receives a reward $\delta_{f(s),a}=\pm1$.

\item \textbf{Cliffwalk-PQC:} this environment essentially adds a temporal structure to SL-PQC: each episode starts from a fixed state $s_0 \in S$, and if an agent assigns the correct label to a state $s_i$, $0\leq i \leq 19$, it moves to a fixed state $s_{i+1}$ and receives a $+1$ reward, otherwise the episode is instantly terminated and the agent gets a $-1$ reward. Reaching $s_{20}$ also causes termination.\\
\end{itemize}

\subsubsection{Performance comparison\label{sec:PQC-env-sim}}

Having defined our PQC-generated environments, we now evaluate the performance of \textsc{softmax-PQC} and DNN policies in these tasks. The particular models we consider are \textsc{softmax-PQC}s with PQCs sampled from the same family as that of the \textsc{raw-PQC}s generating the environments (but with re-initialized parameters $\params$), and DNNs using Rectified Linear Units (ReLUs) in their hidden layers. In our hyperparameter search, we evaluated the performance of DNNs with a wide range of depths (number of hidden layers between $2$ to $10$) and widths (number of units per hidden layer between $8$ and $64$), and kept the architecture with the best average performance (depth $4$, width $16$).

Despite this hyperparametrization, we find (see Fig.\ \ref{fig:pqc-vs-nn}, and Fig.\ \ref{fig:pqc-vs-nn-2} in Appendix \ref{sec:PQC-DNN} for different environment instances) that the performance of DNN policies on these tasks remains limited compared to that of \textsc{softmax-PQC}s, that learn close-to-optimal policies on both tasks. Moreover, we observe that the separation in performance gets boosted by the cliffwalk temporal structure. This is likely do to the increased complexity of this task, as, in order to move farther in the cliffwalk, the policy family should allow learning new labels without ``forgetting'' the labels of earlier states. In these particular case studies, the \textsc{softmax-PQC} policies exhibited sufficient flexibility in this sense, whereas the DNNs we considered did not (see Appendix \ref{sec:PQC-DNN} for a visualization of these policies). Note that these results do not reflect the difficulty of our tasks at the sizes we consider (a look-up table would perform optimally) but rather highlight the inefficacy of these DNNs at learning PQC functions.

\newpage

\section{Conclusion}

In this work, we have investigated the design of quantum RL agents based on PQCs. We proposed several constructions and showed the impact of certain design choices on learning performance. In particular, we introduced the \textsc{softmax-PQC} model, where a softmax policy is computed from expectation values of a PQC with both trainable observable weights and input scaling parameters. These added features to standard PQCs used in ML (e.g., as quantum classifiers) enhance both the expressivity and flexibility of PQC policies, which allows them to achieve a learning performance on benchmarking environments comparable to that of standard DNNs. We additionally demonstrated the existence of task environments, constructed out of PQCs, that are very natural for PQC agents, but on which DNN agents have a poor performance. To strengthen this result, we constructed several RL environments, each with a different degree of degeneracy (i.e., closeness to a supervised learning task), where we showed a rigorous separation between a class of PQC agents and any classical learner, based on the widely-believed classical hardness of the discrete logarithm problem. We believe that our results constitute strides toward a practical quantum advantage in RL using near-term quantum~devices.

\section{Broad impact\label{sec:broad-impact}}

We expect our work to have an overall positive societal impact. Notably, we believe that our approach to QRL could be beneficial in the two following ways:
\begin{itemize}[leftmargin=4mm]
	\item Modern-day RL is known to be very resource-heavy in terms of compute power and energy consumption (see, e.g., the resources needed to train AlphaGo Zero \cite{silver17}). In other computational problems, e.g., the quantum supremacy problem of Google \cite{arute19}, it was shown that, because of their computational advantages, quantum computers could save many orders of magnitude in energy consumption compared to classical supercomputers \cite{villalonga20}. Therefore, a quantum learning advantage as showcased in our work could potentially alleviate the computational demands of RL, making it more economically appealing and environmentally-friendly.
	\item Aside from the game-based problems that we consider in our work, the areas of application of RL are constantly increasing \cite{kober13,mahmud18,yu19}. The learning advantages of QRL could potentially make these existing applications more accessible technologically and economically, but also unlock new applications, e.g., in problems in quantum information \cite{albarran18,wu20} or quantum chemistry \cite{peruzzo14}.
\end{itemize}

At the same time, our work may have certain negative consequences. Notably, QRL will inherit many of the problems that are already present in classical RL and ML in general. For instance, it is not clear whether the question of interpretability of learning models \cite{linardatos21} will be negatively or positively impacted by switching to quantum models. One could argue that the inability to fully access the quantum Hilbert spaces in which quantum computers operate can turn learning models even further into ``black-boxes'' than existing classical models. Also, similarly to the fact that current state-of-the-art ML/RL requires supercomputers that are not accessible to everyone, private and select access to quantum computers could emphasize existing inequalities in developing and using AI.

\begin{ack}
The authors would like to thank Srinivasan Arunachalam for clarifications on the testing accuracy of their quantum classifier in the DLP classification task. The authors would also like to thank Andrea Skolik and Arjan Cornelissen for helpful discussions and comments. CG thanks Thomas Moerland for discussions in the early phases of this project. SJ and HJB acknowledge support from the Austrian Science Fund (FWF) through the projects DK-ALM:W1259-N27 and SFB BeyondC F7102. SJ also acknowledges the Austrian Academy of Sciences as a recipient of the DOC Fellowship. This work was in part supported by the Dutch Research Council (NWO/OCW), as part of the Quantum Software Consortium program (project number 024.003.037). VD and SM acknowledge the support by the project NEASQC funded from the European Union’s Horizon 2020 research and innovation programme (grant agreement No 951821). VD and SM also acknowledge partial funding by an unrestricted gift from Google Quantum AI. The computational results presented here have been achieved in part using the LEO HPC infrastructure of the University of Innsbruck.
\end{ack}

\bibliographystyle{unsrtnat}
\bibliography{references}

\begin{thebibliography}{62}
\providecommand{\natexlab}[1]{#1}
\providecommand{\url}[1]{\texttt{#1}}
\expandafter\ifx\csname urlstyle\endcsname\relax
  \providecommand{\doi}[1]{doi: #1}\else
  \providecommand{\doi}{doi: \begingroup \urlstyle{rm}\Url}\fi

\bibitem[Preskill(2018)]{preskill18}
John Preskill.
\newblock Quantum computing in the nisq era and beyond.
\newblock \emph{Quantum}, 2:\penalty0 79, 2018.

\bibitem[Bharti et~al.(2021)Bharti, Cervera-Lierta, Kyaw, Haug, Alperin-Lea,
  Anand, Degroote, Heimonen, Kottmann, Menke, et~al.]{bharti21}
Kishor Bharti, Alba Cervera-Lierta, Thi~Ha Kyaw, Tobias Haug, Sumner
  Alperin-Lea, Abhinav Anand, Matthias Degroote, Hermanni Heimonen, Jakob~S
  Kottmann, Tim Menke, et~al.
\newblock Noisy intermediate-scale quantum (nisq) algorithms.
\newblock \emph{arXiv preprint arXiv:2101.08448}, 2021.

\bibitem[Benedetti et~al.(2019)Benedetti, Lloyd, Sack, and
  Fiorentini]{benedetti19}
Marcello Benedetti, Erika Lloyd, Stefan Sack, and Mattia Fiorentini.
\newblock Parameterized quantum circuits as machine learning models.
\newblock \emph{Quantum Science and Technology}, 4\penalty0 (4):\penalty0
  043001, 2019.

\bibitem[Farhi and Neven(2018)]{farhi18}
Edward Farhi and Hartmut Neven.
\newblock Classification with quantum neural networks on near term processors.
\newblock \emph{arXiv preprint arXiv:1802.06002}, 2018.

\bibitem[Schuld et~al.(2020)Schuld, Bocharov, Svore, and Wiebe]{schuld20b}
Maria Schuld, Alex Bocharov, Krysta~M Svore, and Nathan Wiebe.
\newblock Circuit-centric quantum classifiers.
\newblock \emph{Physical Review A}, 101\penalty0 (3):\penalty0 032308, 2020.

\bibitem[Havl{\'\i}{\v{c}}ek et~al.(2019)Havl{\'\i}{\v{c}}ek, C{\'o}rcoles,
  Temme, Harrow, Kandala, Chow, and Gambetta]{havlivcek19}
Vojt{\v{e}}ch Havl{\'\i}{\v{c}}ek, Antonio~D C{\'o}rcoles, Kristan Temme,
  Aram~W Harrow, Abhinav Kandala, Jerry~M Chow, and Jay~M Gambetta.
\newblock Supervised learning with quantum-enhanced feature spaces.
\newblock \emph{Nature}, 567\penalty0 (7747):\penalty0 209--212, 2019.

\bibitem[Schuld and Killoran(2019)]{schuld19b}
Maria Schuld and Nathan Killoran.
\newblock Quantum machine learning in feature hilbert spaces.
\newblock \emph{Physical review letters}, 122\penalty0 (4):\penalty0 040504,
  2019.

\bibitem[Peters et~al.(2021)Peters, Caldeira, Ho, Leichenauer, Mohseni, Neven,
  Spentzouris, Strain, and Perdue]{peters21}
Evan Peters, Joao Caldeira, Alan Ho, Stefan Leichenauer, Masoud Mohseni,
  Hartmut Neven, Panagiotis Spentzouris, Doug Strain, and Gabriel~N Perdue.
\newblock Machine learning of high dimensional data on a noisy quantum
  processor.
\newblock \emph{arXiv preprint arXiv:2101.09581}, 2021.

\bibitem[Liu and Wang(2018)]{liu18}
Jin-Guo Liu and Lei Wang.
\newblock Differentiable learning of quantum circuit born machines.
\newblock \emph{Physical Review A}, 98\penalty0 (6):\penalty0 062324, 2018.

\bibitem[Zhu et~al.(2019)Zhu, Linke, Benedetti, Landsman, Nguyen, Alderete,
  Perdomo-Ortiz, Korda, Garfoot, Brecque, et~al.]{zhu19}
Daiwei Zhu, Norbert~M Linke, Marcello Benedetti, Kevin~A Landsman, Nhung~H
  Nguyen, C~Huerta Alderete, Alejandro Perdomo-Ortiz, Nathan Korda, A~Garfoot,
  Charles Brecque, et~al.
\newblock Training of quantum circuits on a hybrid quantum computer.
\newblock \emph{Science advances}, 5\penalty0 (10):\penalty0 eaaw9918, 2019.

\bibitem[Otterbach et~al.(2017)Otterbach, Manenti, Alidoust, Bestwick, Block,
  Bloom, Caldwell, Didier, Fried, Hong, et~al.]{otterbach17}
JS~Otterbach, R~Manenti, N~Alidoust, A~Bestwick, M~Block, B~Bloom, S~Caldwell,
  N~Didier, E~Schuyler Fried, S~Hong, et~al.
\newblock Unsupervised machine learning on a hybrid quantum computer.
\newblock \emph{arXiv preprint arXiv:1712.05771}, 2017.

\bibitem[Huang et~al.(2021)Huang, Broughton, Mohseni, Babbush, Boixo, Neven,
  and McClean]{huang20}
Hsin-Yuan Huang, Michael Broughton, Masoud Mohseni, Ryan Babbush, Sergio Boixo,
  Hartmut Neven, and Jarrod~R McClean.
\newblock Power of data in quantum machine learning.
\newblock \emph{Nature communications}, 12\penalty0 (1):\penalty0 1--9, 2021.

\bibitem[Du et~al.(2020)Du, Hsieh, Liu, and Tao]{du20}
Yuxuan Du, Min-Hsiu Hsieh, Tongliang Liu, and Dacheng Tao.
\newblock Expressive power of parametrized quantum circuits.
\newblock \emph{Physical Review Research}, 2\penalty0 (3):\penalty0 033125,
  2020.

\bibitem[Liu et~al.(2021)Liu, Arunachalam, and Temme]{liu20}
Yunchao Liu, Srinivasan Arunachalam, and Kristan Temme.
\newblock A rigorous and robust quantum speed-up in supervised machine
  learning.
\newblock \emph{Nature Physics}, 17\penalty0 (9):\penalty0 1013--1017, 2021.

\bibitem[Sweke et~al.(2021)Sweke, Seifert, Hangleiter, and Eisert]{sweke20}
Ryan Sweke, Jean-Pierre Seifert, Dominik Hangleiter, and Jens Eisert.
\newblock On the quantum versus classical learnability of discrete
  distributions.
\newblock \emph{Quantum}, 5:\penalty0 417, 2021.

\bibitem[Mnih et~al.(2015)Mnih, Kavukcuoglu, Silver, Rusu, Veness, Bellemare,
  Graves, Riedmiller, Fidjeland, Ostrovski, et~al.]{mnih15}
Volodymyr Mnih, Koray Kavukcuoglu, David Silver, Andrei~A Rusu, Joel Veness,
  Marc~G Bellemare, Alex Graves, Martin Riedmiller, Andreas~K Fidjeland, Georg
  Ostrovski, et~al.
\newblock Human-level control through deep reinforcement learning.
\newblock \emph{nature}, 518\penalty0 (7540):\penalty0 529--533, 2015.

\bibitem[Silver et~al.(2017)Silver, Schrittwieser, Simonyan, Antonoglou, Huang,
  Guez, Hubert, Baker, Lai, Bolton, et~al.]{silver17}
David Silver, Julian Schrittwieser, Karen Simonyan, Ioannis Antonoglou, Aja
  Huang, Arthur Guez, Thomas Hubert, Lucas Baker, Matthew Lai, Adrian Bolton,
  et~al.
\newblock Mastering the game of go without human knowledge.
\newblock \emph{Nature}, 550\penalty0 (7676):\penalty0 354, 2017.

\bibitem[Berner et~al.(2019)Berner, Brockman, Chan, Cheung, D{\k{e}}biak,
  Dennison, Farhi, Fischer, Hashme, Hesse, et~al.]{berner19}
Christopher Berner, Greg Brockman, Brooke Chan, Vicki Cheung, Przemys{\l}aw
  D{\k{e}}biak, Christy Dennison, David Farhi, Quirin Fischer, Shariq Hashme,
  Chris Hesse, et~al.
\newblock Dota 2 with large scale deep reinforcement learning.
\newblock \emph{arXiv preprint arXiv:1912.06680}, 2019.

\bibitem[Mirowski et~al.(2018)Mirowski, Grimes, Malinowski, Hermann, Anderson,
  Teplyashin, Simonyan, Zisserman, Hadsell, et~al.]{mirowski18}
Piotr Mirowski, Matt Grimes, Mateusz Malinowski, Karl~Moritz Hermann, Keith
  Anderson, Denis Teplyashin, Karen Simonyan, Andrew Zisserman, Raia Hadsell,
  et~al.
\newblock Learning to navigate in cities without a map.
\newblock \emph{Advances in Neural Information Processing Systems},
  31:\penalty0 2419--2430, 2018.

\bibitem[Chen et~al.(2020)Chen, Yang, Qi, Chen, Ma, and Goan]{chen20}
Samuel Yen-Chi Chen, Chao-Han~Huck Yang, Jun Qi, Pin-Yu Chen, Xiaoli Ma, and
  Hsi-Sheng Goan.
\newblock Variational quantum circuits for deep reinforcement learning.
\newblock \emph{IEEE Access}, 8:\penalty0 141007--141024, 2020.

\bibitem[Lockwood and Si(2020)]{lockwood20}
Owen Lockwood and Mei Si.
\newblock Reinforcement learning with quantum variational circuit.
\newblock In \emph{Proceedings of the AAAI Conference on Artificial
  Intelligence and Interactive Digital Entertainment}, volume~16, pages
  245--251, 2020.

\bibitem[Wu et~al.(2020)Wu, Jin, Wen, and Wang]{wu20}
Shaojun Wu, Shan Jin, Dingding Wen, and Xiaoting Wang.
\newblock Quantum reinforcement learning in continuous action space.
\newblock \emph{arXiv preprint arXiv:2012.10711}, 2020.

\bibitem[Jerbi et~al.(2021)Jerbi, Trenkwalder, Poulsen~Nautrup, Briegel, and
  Dunjko]{jerbi19}
Sofiene Jerbi, Lea~M. Trenkwalder, Hendrik Poulsen~Nautrup, Hans~J. Briegel,
  and Vedran Dunjko.
\newblock Quantum enhancements for deep reinforcement learning in large spaces.
\newblock \emph{PRX Quantum}, 2:\penalty0 010328, Feb 2021.

\bibitem[Brockman et~al.(2016)Brockman, Cheung, Pettersson, Schneider,
  Schulman, Tang, and Zaremba]{brockman16}
Greg Brockman, Vicki Cheung, Ludwig Pettersson, Jonas Schneider, John Schulman,
  Jie Tang, and Wojciech Zaremba.
\newblock Openai gym.
\newblock \emph{arXiv preprint arXiv:1606.01540}, 2016.

\bibitem[Shor(1999)]{shor99}
Peter~W Shor.
\newblock Polynomial-time algorithms for prime factorization and discrete
  logarithms on a quantum computer.
\newblock \emph{SIAM review}, 41\penalty0 (2):\penalty0 303--332, 1999.

\bibitem[Blum and Micali(1984)]{blum84}
Manuel Blum and Silvio Micali.
\newblock How to generate cryptographically strong sequences of pseudorandom
  bits.
\newblock \emph{SIAM journal on Computing}, 13\penalty0 (4):\penalty0 850--864,
  1984.

\bibitem[Skolik et~al.(2021)Skolik, Jerbi, and Dunjko]{skolik21b}
Andrea Skolik, Sofiene Jerbi, and Vedran Dunjko.
\newblock Quantum agents in the gym: a variational quantum algorithm for deep
  q-learning.
\newblock \emph{arXiv preprint arXiv:2103.15084}, 2021.

\bibitem[P{\'e}rez-Salinas et~al.(2020)P{\'e}rez-Salinas, Cervera-Lierta,
  Gil-Fuster, and Latorre]{perez20}
Adri{\'a}n P{\'e}rez-Salinas, Alba Cervera-Lierta, Elies Gil-Fuster, and
  Jos{\'e}~I Latorre.
\newblock Data re-uploading for a universal quantum classifier.
\newblock \emph{Quantum}, 4:\penalty0 226, 2020.

\bibitem[Dong et~al.(2008)Dong, Chen, Li, and Tarn]{dong08}
Daoyi Dong, Chunlin Chen, Hanxiong Li, and Tzyh-Jong Tarn.
\newblock Quantum reinforcement learning.
\newblock \emph{IEEE Transactions on Systems, Man, and Cybernetics, Part B
  (Cybernetics)}, 38\penalty0 (5):\penalty0 1207--1220, 2008.

\bibitem[Paparo et~al.(2014)Paparo, Dunjko, Makmal, Martin-Delgado, and
  Briegel]{paparo14}
Giuseppe~Davide Paparo, Vedran Dunjko, Adi Makmal, Miguel~Angel Martin-Delgado,
  and Hans~J Briegel.
\newblock Quantum speedup for active learning agents.
\newblock \emph{Physical Review X}, 4\penalty0 (3):\penalty0 031002, 2014.

\bibitem[Dunjko et~al.(2016)Dunjko, Taylor, and Briegel]{dunjko16}
Vedran Dunjko, Jacob~M Taylor, and Hans~J Briegel.
\newblock Quantum-enhanced machine learning.
\newblock \emph{Physical review letters}, 117\penalty0 (13):\penalty0 130501,
  2016.

\bibitem[Crawford et~al.(2018)Crawford, Levit, Ghadermarzy, Oberoi, and
  Ronagh]{crawford18}
Daniel Crawford, Anna Levit, Navid Ghadermarzy, Jaspreet~S Oberoi, and Pooya
  Ronagh.
\newblock Reinforcement learning using quantum boltzmann machines.
\newblock \emph{Quantum Information \& Computation}, 18\penalty0
  (1-2):\penalty0 51--74, 2018.

\bibitem[Neukart et~al.(2018)Neukart, Von~Dollen, Seidel, and
  Compostella]{neukart18}
Florian Neukart, David Von~Dollen, Christian Seidel, and Gabriele Compostella.
\newblock Quantum-enhanced reinforcement learning for finite-episode games with
  discrete state spaces.
\newblock \emph{Frontiers in Physics}, 5:\penalty0 71, 2018.

\bibitem[Grover(1996)]{grover96}
Lov~K Grover.
\newblock A fast quantum mechanical algorithm for database search.
\newblock In \emph{Proceedings of the twenty-eighth annual ACM symposium on
  Theory of computing}, pages 212--219, 1996.

\bibitem[Johnson et~al.(2011)Johnson, Amin, Gildert, Lanting, Hamze, Dickson,
  Harris, Berkley, Johansson, Bunyk, et~al.]{johnson11}
Mark~W Johnson, Mohammad~HS Amin, Suzanne Gildert, Trevor Lanting, Firas Hamze,
  Neil Dickson, Richard Harris, Andrew~J Berkley, Jan Johansson, Paul Bunyk,
  et~al.
\newblock Quantum annealing with manufactured spins.
\newblock \emph{Nature}, 473\penalty0 (7346):\penalty0 194--198, 2011.

\bibitem[Quantum(2021)]{tfq21}
TensorFlow Quantum.
\newblock Parametrized quantum circuits for reinforcement learning.
\newblock
  \textsc{url:}~\href{https://www.tensorflow.org/quantum/tutorials/quantum_reinforcement_learning}{\mbox{tensorflow.org/quantum/tutorials/quantum\_reinforcement\_learning}},
  2021.

\bibitem[Broughton et~al.(2020)Broughton, Verdon, McCourt, Martinez, Yoo,
  Isakov, Massey, Niu, Halavati, Peters, et~al.]{broughton20}
Michael Broughton, Guillaume Verdon, Trevor McCourt, Antonio~J Martinez,
  Jae~Hyeon Yoo, Sergei~V Isakov, Philip Massey, Murphy~Yuezhen Niu, Ramin
  Halavati, Evan Peters, et~al.
\newblock Tensorflow quantum: A software framework for quantum machine
  learning.
\newblock \emph{arXiv preprint arXiv:2003.02989}, 2020.

\bibitem[Nielsen and Chuang(2000)]{nielsen00}
Michael~A. Nielsen and Isaac~L. Chuang.
\newblock \emph{Quantum Computation and Quantum Information}.
\newblock Cambridge University Press, 2000.

\bibitem[De~Wolf(2019)]{dewolf19}
Ronald De~Wolf.
\newblock Quantum computing: Lecture notes.
\newblock \emph{arXiv preprint arXiv:1907.09415}, 2019.

\bibitem[Kandala et~al.(2017)Kandala, Mezzacapo, Temme, Takita, Brink, Chow,
  and Gambetta]{kandala17}
Abhinav Kandala, Antonio Mezzacapo, Kristan Temme, Maika Takita, Markus Brink,
  Jerry~M Chow, and Jay~M Gambetta.
\newblock Hardware-efficient variational quantum eigensolver for small
  molecules and quantum magnets.
\newblock \emph{Nature}, 549\penalty0 (7671):\penalty0 242--246, 2017.

\bibitem[Schuld et~al.(2021)Schuld, Sweke, and Meyer]{schuld20}
Maria Schuld, Ryan Sweke, and Johannes~Jakob Meyer.
\newblock Effect of data encoding on the expressive power of variational
  quantum-machine-learning models.
\newblock \emph{Physical Review A}, 103\penalty0 (3):\penalty0 032430, 2021.

\bibitem[Schuld et~al.(2019)Schuld, Bergholm, Gogolin, Izaac, and
  Killoran]{schuld19}
Maria Schuld, Ville Bergholm, Christian Gogolin, Josh Izaac, and Nathan
  Killoran.
\newblock Evaluating analytic gradients on quantum hardware.
\newblock \emph{Physical Review A}, 99\penalty0 (3):\penalty0 032331, 2019.

\bibitem[Weng(2018)]{weng18}
Lilian Weng.
\newblock Policy gradient algorithms.
\newblock
  \textsc{url:}~\href{https://lilianweng.github.io/lil-log/2018/04/08/policy-gradient-algorithms.html}{lilianweng.github.io/lil-log},
  2018.

\bibitem[Sutton et~al.(1998)Sutton, Barto, et~al.]{sutton98}
Richard~S Sutton, Andrew~G Barto, et~al.
\newblock \emph{Reinforcement learning: An introduction}.
\newblock 1998.

\bibitem[Williams(1992)]{williams92}
Ronald~J Williams.
\newblock Simple statistical gradient-following algorithms for connectionist
  reinforcement learning.
\newblock \emph{Machine learning}, 8\penalty0 (3-4):\penalty0 229--256, 1992.

\bibitem[Mitarai et~al.(2018)Mitarai, Negoro, Kitagawa, and Fujii]{mitarai18}
Kosuke Mitarai, Makoto Negoro, Masahiro Kitagawa, and Keisuke Fujii.
\newblock Quantum circuit learning.
\newblock \emph{Physical Review A}, 98\penalty0 (3):\penalty0 032309, 2018.

\bibitem[Greensmith et~al.(2004)Greensmith, Bartlett, and Baxter]{greensmith04}
Evan Greensmith, Peter~L Bartlett, and Jonathan Baxter.
\newblock Variance reduction techniques for gradient estimates in reinforcement
  learning.
\newblock \emph{Journal of Machine Learning Research}, 5\penalty0
  (Nov):\penalty0 1471--1530, 2004.

\bibitem[Duan et~al.(2016)Duan, Chen, Houthooft, Schulman, and Abbeel]{duan16}
Yan Duan, Xi~Chen, Rein Houthooft, John Schulman, and Pieter Abbeel.
\newblock Benchmarking deep reinforcement learning for continuous control.
\newblock In \emph{International conference on machine learning}, pages
  1329--1338. PMLR, 2016.

\bibitem[OpenAI(2020)]{gym20}
OpenAI.
\newblock Leaderboard of openai gym environments.
\newblock
  \textsc{url:}~\href{https://github.com/openai/gym/wiki/Leaderboard}{github.com/openai/gym/wiki},
  2020.

\bibitem[Dunjko et~al.(2017)Dunjko, Liu, Wu, and Taylor]{dunjko17b}
Vedran Dunjko, Yi-Kai Liu, Xingyao Wu, and Jacob~M Taylor.
\newblock Exponential improvements for quantum-accessible reinforcement
  learning.
\newblock \emph{arXiv preprint arXiv:1710.11160}, 2017.

\bibitem[Arute et~al.(2019)Arute, Arya, Babbush, Bacon, Bardin, Barends,
  Biswas, Boixo, Brandao, Buell, et~al.]{arute19}
Frank Arute, Kunal Arya, Ryan Babbush, Dave Bacon, Joseph~C Bardin, Rami
  Barends, Rupak Biswas, Sergio Boixo, Fernando~GSL Brandao, David~A Buell,
  et~al.
\newblock Quantum supremacy using a programmable superconducting processor.
\newblock \emph{Nature}, 574\penalty0 (7779):\penalty0 505--510, 2019.

\bibitem[Villalonga et~al.(2020)Villalonga, Lyakh, Boixo, Neven, Humble,
  Biswas, Rieffel, Ho, and Mandr{\`a}]{villalonga20}
Benjamin Villalonga, Dmitry Lyakh, Sergio Boixo, Hartmut Neven, Travis~S
  Humble, Rupak Biswas, Eleanor~G Rieffel, Alan Ho, and Salvatore Mandr{\`a}.
\newblock Establishing the quantum supremacy frontier with a 281 pflop/s
  simulation.
\newblock \emph{Quantum Science and Technology}, 5\penalty0 (3):\penalty0
  034003, 2020.

\bibitem[Kober et~al.(2013)Kober, Bagnell, and Peters]{kober13}
Jens Kober, J~Andrew Bagnell, and Jan Peters.
\newblock Reinforcement learning in robotics: A survey.
\newblock \emph{The International Journal of Robotics Research}, 32\penalty0
  (11):\penalty0 1238--1274, 2013.

\bibitem[Mahmud et~al.(2018)Mahmud, Kaiser, Hussain, and Vassanelli]{mahmud18}
Mufti Mahmud, Mohammed~Shamim Kaiser, Amir Hussain, and Stefano Vassanelli.
\newblock Applications of deep learning and reinforcement learning to
  biological data.
\newblock \emph{IEEE transactions on neural networks and learning systems},
  29\penalty0 (6):\penalty0 2063--2079, 2018.

\bibitem[Yu et~al.(2019)Yu, Liu, and Nemati]{yu19}
Chao Yu, Jiming Liu, and Shamim Nemati.
\newblock Reinforcement learning in healthcare: A survey.
\newblock \emph{arXiv preprint arXiv:1908.08796}, 2019.

\bibitem[Albarr{\'a}n-Arriagada et~al.(2018)Albarr{\'a}n-Arriagada, Retamal,
  Solano, and Lamata]{albarran18}
Francisco Albarr{\'a}n-Arriagada, Juan~C Retamal, Enrique Solano, and Lucas
  Lamata.
\newblock Measurement-based adaptation protocol with quantum reinforcement
  learning.
\newblock \emph{Physical Review A}, 98\penalty0 (4):\penalty0 042315, 2018.

\bibitem[Peruzzo et~al.(2014)Peruzzo, McClean, Shadbolt, Yung, Zhou, Love,
  Aspuru-Guzik, and O’brien]{peruzzo14}
Alberto Peruzzo, Jarrod McClean, Peter Shadbolt, Man-Hong Yung, Xiao-Qi Zhou,
  Peter~J Love, Al{\'a}n Aspuru-Guzik, and Jeremy~L O’brien.
\newblock A variational eigenvalue solver on a photonic quantum processor.
\newblock \emph{Nature communications}, 5\penalty0 (1):\penalty0 1--7, 2014.

\bibitem[Linardatos et~al.(2021)Linardatos, Papastefanopoulos, and
  Kotsiantis]{linardatos21}
Pantelis Linardatos, Vasilis Papastefanopoulos, and Sotiris Kotsiantis.
\newblock Explainable ai: A review of machine learning interpretability
  methods.
\newblock \emph{Entropy}, 23\penalty0 (1):\penalty0 18, 2021.

\bibitem[Goto et~al.(2021)Goto, Tran, and Nakajima]{goto21}
Takahiro Goto, Quoc~Hoan Tran, and Kohei Nakajima.
\newblock Universal approximation property of quantum machine learning models
  in quantum-enhanced feature spaces.
\newblock \emph{Physical Review Letters}, 127\penalty0 (9):\penalty0 090506,
  2021.

\bibitem[P{\'e}rez-Salinas et~al.(2021)P{\'e}rez-Salinas,
  L{\'o}pez-N{\'u}{\~n}ez, Garc{\'\i}a-S{\'a}ez, Forn-D{\'\i}az, and
  Latorre]{perez21}
Adri{\'a}n P{\'e}rez-Salinas, David L{\'o}pez-N{\'u}{\~n}ez, Artur
  Garc{\'\i}a-S{\'a}ez, Pol Forn-D{\'\i}az, and Jos{\'e}~I Latorre.
\newblock One qubit as a universal approximant.
\newblock \emph{Physical Review A}, 104\penalty0 (1):\penalty0 012405, 2021.

\bibitem[Google(2018)]{cirq18}
Google.
\newblock Cirq: A python framework for creating, editing, and invoking noisy
  intermediate scale quantum circuits.
\newblock
  \textsc{url:}~\href{https://github.com/quantumlib/Cirq}{github.com/quantumlib/Cirq},
  2018.

\bibitem[Suzuki et~al.(2020)Suzuki, Kawase, Masumura, Hiraga, Nakadai, Chen,
  Nakanishi, Mitarai, Imai, Tamiya, et~al.]{suzuki20}
Yasunari Suzuki, Yoshiaki Kawase, Yuya Masumura, Yuria Hiraga, Masahiro
  Nakadai, Jiabao Chen, Ken~M Nakanishi, Kosuke Mitarai, Ryosuke Imai, Shiro
  Tamiya, et~al.
\newblock Qulacs: a fast and versatile quantum circuit simulator for research
  purpose.
\newblock \emph{arXiv preprint arXiv:2011.13524}, 2020.

\end{thebibliography}

\newpage

\newpage

\vbox{%
    \hsize\textwidth
    \linewidth\hsize
    \hrule height 4pt
    \vskip 0.25in
    \vskip -\parskip%
    \centering
    {\LARGE\bf Supplementary Material for:\\ {\fontsize{15}{15} \selectfont Parametrized Quantum Policies for Reinforcement Learning}\par}
    \vskip 0.29in
    \vskip -\parskip
    \hrule height 1pt
    \vskip 0.3in%
  }
  
\vbox{
    \centering
    {\bf{\fontsize{11}{11} \selectfont Sofiene Jerbi, Casper Gyurik, Simon C.~Marshall, Hans J.~Briegel, Vedran Dunjko}\par}
    \vskip 0.1in
}

\appendix

\paragraph{Outline}
The Supplementary Material is organized as follows.
In Appendix \ref{sec:gradient-log-prob}, we derive the expression of the log-policy gradient for \textsc{softmax-PQC}s presented in Lemma \ref{lem:gradient-exp}. In Appendix~\ref{sec:proofs-softmax-pqc}, we prove Lemmas \ref{lemma:tv-policy} and \ref{lemma:log-pol-grad} on the efficient policy sampling and the efficient estimation of the log-policy gradient for \textsc{softmax-PQC} policies. In Appendix \ref{sec:trainable-observables}, we clarify the role of the trainable observables in our definition of \textsc{softmax-PQC} policies.  In Appendix \ref{sec:env-spec-hyper}, we give a specification of the environments considered in our numerical simulations, as well the hyperparameters we used to train all RL agents. In Appendix \ref{sec:PQC-DNN}, we present additional plots and numerical simulations that help our understanding and visualization of PQC polices. In Appendix \ref{sec:DLP-task}, we give a succinct description of the DLP classification task of Liu \emph{et al.} In Appendices \ref{sec:proof-thm-separations} to \ref{sec:proof-deterministic-DLP}, we prove our main Theorem \ref{thm:separations-DLP} on learning separations in DLP environments.  In appendix \ref{sec:PQC-agent-DLP}, we construct PQC agents with provable guarantees of solving the DLP environments, stated and proven in Theorem \ref{thm:perf-PQC}. 

\section{Derivation of the log-policy gradient\label{sec:gradient-log-prob}}

For a \textsc{softmax-PQC} defined in Def.\ \ref{def:raw-softmax-PQC}, we have:
\begin{align*}
	\nabla_{\params} \log\policy(a|s) &= \nabla_{\params} \log e^{\beta\expval{O_a}_{s,\params}} - \nabla_{\params} \log \sum_{a'} e^{\beta\expval{O_{a'}}_{s,\params}}\\
	&= \beta\nabla_{\params} \expval{O_a}_{s,\params} - \sum_{a'} \frac{e^{\beta\expval{O_{a'}}_{s,\params}}  \beta\nabla_{\params} \expval{O_{a'}}_{s,\params}}{\sum_{a''} e^{\beta\expval{O_{a''}}_{s,\params}}} \\
	&= \beta \left(\nabla_{\params}\expval{O_a}_{s,\params} - \sum_{a'} \policy(a'|s) \nabla_{\params}\expval{O_{a'}}_{s,\params}\right).
\end{align*}

\vspace{0.5em}

\section{Efficient implementation of \textsc{softmax-PQC} policies\label{sec:proofs-softmax-pqc}}

\subsection{Efficient approximate policy sampling}

In this section we prove Lemma \ref{lemma:tv-policy}, restated below:
\tvpolicy*
\begin{proof}
Consider $\abs{A}$ estimates $\left\{\langle\widetilde{O_a}\rangle_{s,\params}\right\}_{1\leq a \leq \abs{A}}$, obtained all to additive error $\varepsilon$, i.e., 
\begin{equation*}
\abs{\langle\widetilde{O_a}\rangle_{s,\params} - \expval{O_a}_{s,\params}} \leq \varepsilon, \quad \forall a
\end{equation*}
and used to compute an approximate policy 
\begin{equation*}
\widetilde{\pi}_{\bm{\theta}}(a|s) = \frac{e^{\beta \langle\widetilde{O_a}\rangle_{s,\params}}}{\sum_{a'} e^{\beta \langle\widetilde{O_{a'}}\rangle_{s,\params}}}.
\end{equation*}

Due to the monoticity of the exponential, we have, for all $a$:
\newlength{\tmp}
\settowidth{\tmp}{$\frac{e^{\beta \langle\widetilde{O_a}\rangle_{s,\params}}}{\sum_{a'} e^{\beta \langle\widetilde{O_{a'}}\rangle_{s,\params}}}$}
\begin{alignat}{2}\label{eq:approx-policy-ineq}
	\frac{e^{-\beta\varepsilon} e^{\beta \expval{O_a}_{s,\params}}}{e^{\beta\varepsilon}\sum_{a'}e^{\beta \expval{O_{a'}}_{s,\params}}} &\leq &\frac{e^{\beta \langle\widetilde{O_a}\rangle_{s,\params}}}{\sum_{a'} e^{\beta \langle\widetilde{O_{a'}}\rangle_{s,\params}}} &\leq \frac{e^{\beta\varepsilon} e^{\beta \expval{O_a}_{s,\params}}}{e^{-\beta\varepsilon}\sum_{a'}e^{\beta \expval{O_{a'}}_{s,\params}}} \nonumber\\
	\Leftrightarrow \quad\quad e^{-2\beta\varepsilon}\policy(a|s) &\leq &\makebox[\tmp]{$\widetilde{\pi}_{\bm{\theta}}(a|s)$} &\leq e^{2\beta\varepsilon}\policy(a|s).
\end{alignat}
Hence,
\allowdisplaybreaks
\begin{align*}
	\text{TV}(\policy,\widetilde{\pi}_{\bm{\theta}}) &= \sum_a \abs{\widetilde{\pi}_{\bm{\theta}}(a|s)-\policy(a|s)}\\
	&\leq \sum_a \abs{e^{2\beta\varepsilon}\policy(a|s)-e^{-2\beta\varepsilon}\policy(a|s)}\\
	&= \sum_a \abs{e^{2\beta\varepsilon}-e^{-2\beta\varepsilon}}\policy(a|s)\\
	&= 2\abs{\sinh(2\beta\varepsilon)} \mathrel{\underset{\beta\varepsilon \to 0^+}{=}} 4 \beta\varepsilon + \mathcal{O}\left((\beta\varepsilon)^3\right),
\end{align*}
where we used $\left\{\widetilde{\pi}_{\bm{\theta}}(a|s),\policy(a|s)\right\} \in [e^{-2\beta\varepsilon}\policy(a|s),e^{2\beta\varepsilon}\policy(a|s)]$ in the first inequality.
\end{proof}

\subsection{Efficient estimation of the log-policy gradient}
Using a similar approach to the proof of the previous section, we show the following lemma:
\begin{lemma}\label{lemma:log-pol-grad}
	For a \textsc{softmax-PQC} policy $\policy$ defined by a unitary $U(s,\params)$ and observables $O_a$, call $\partial_i\langle\widetilde{O_{a}}\rangle_{s,\params}$ approximations of the true derivatives $\partial_i\langle O_a\rangle_{s,\params}$ with at most $\varepsilon$ additive error, and $\langle\widetilde{O_a}\rangle_{s,\params}$ approximations of the true expectation values $\expval{O_a}_{s,\params}$ with at most $\varepsilon' = \varepsilon(4\beta\max_a\norm{O_a})^{-1}$ additive error. Then the approximate log-policy gradient $\nabla_{\params} \log{\widetilde{\policy}(a|s)} = \beta \big( \nabla_{\params}\langle\widetilde{O_{a}}\rangle_{s,\params} - \sum_{a'} \widetilde{\policy}(a'|s) \nabla_{\params}\langle\widetilde{O_{a'}}\rangle_{s,\params} \big)$ has distance $\mathcal{O}(\beta\varepsilon)$ to $\nabla_{\params} \log{\policy(a|s)}$ in $\ell_\infty$-norm.
\end{lemma}
\begin{proof}
	Call $x_{a,i} = \policy(a|s)\partial_i\langle O_a\rangle_{s,\params}$ and $\widetilde{x}_{a,i} = \widetilde{\pi}_{\bm{\theta}}(a|s) \partial_i\langle\widetilde{O_{a}}\rangle_{s,\params}$, such that:
\begin{equation*}
	\partial_i \log{\widetilde{\policy}(a|s)} =  \beta \Big( \partial_i \langle\widetilde{O_{a}}\rangle_{s,\params} - \sum\nolimits_{a'} \widetilde{x}_{a',i} \Big).
\end{equation*}
and similarly for $\partial_i \log{\policy(a|s)}$.\\
\settowidth{\tmp}{$\sum_a(a|s) \partial_i\langle\widetilde{O_{a}}\rangle_{s,\params}$}
	Using Eq.\ (\ref{eq:approx-policy-ineq}) and that $|\partial_i\langle O_a \rangle_{s,\params} - \partial_i\langle\widetilde{O_a}\rangle_{s,\params}| \leq \varepsilon, \forall a,i$, we have:
\begin{alignat*}{2}
	e^{-2\beta\varepsilon'}\policy(a|s)\left(\partial_i\langle O_a\rangle_{s,\params} - \varepsilon\right) &\leq &\widetilde{\pi}_{\bm{\theta}}(a|s) \partial_i\langle\widetilde{O_{a}}\rangle_{s,\params} &\leq e^{2\beta\varepsilon'}\policy(a|s)\left(\partial_i\langle O_a\rangle_{s,\params} + \varepsilon\right)\\
	\Rightarrow \quad\quad\quad e^{-2\beta\varepsilon'}\left(\sum_a x_{a,i} - \varepsilon\right) &\leq &\makebox[\tmp]{$\mathlarger{\sum\limits}_a\ \widetilde{x}_{a,i}$} &\leq e^{2\beta\varepsilon'}\left(\sum_a x_{a,i} + \varepsilon\right)
\end{alignat*}
where we summed the first inequalities over all $a$. Hence:
\begin{align}\label{eq:diff-xai}
	\abs{\sum_a x_{a,i} - \sum_a \widetilde{x}_{a,i}} &\leq \abs{e^{2\beta\varepsilon'}\left(\sum_a x_{a,i} + \varepsilon\right) - e^{-2\beta\varepsilon'}\left(\sum_a x_{a,i} - \varepsilon\right)} \nonumber\\
	&\leq \abs{(e^{2\beta\varepsilon'}+e^{-2\beta\varepsilon'})\varepsilon + (e^{2\beta\varepsilon'}-e^{-2\beta\varepsilon'})\sum_a x_{a,i}} \nonumber\\
	&\leq \abs{2\cosh(2\beta\varepsilon')\varepsilon + 2\sinh(2\beta\varepsilon')\sum_a x_{a,i}} \nonumber\\
	&\mathrel{\underset{\beta\varepsilon' \to 0^+}{=}} \abs{\varepsilon + 4\beta\varepsilon'\sum_a x_{a,i} + \mathcal{O}\left((\beta\varepsilon')^2\varepsilon\right) + \mathcal{O}\left((\beta\varepsilon')^3\right)}.
\end{align}
We also have 
\begin{equation*}
\abs{\sum_a x_{a,i}} = \abs{\sum_a \policy(a|s) \partial_i\langle O_a\rangle_{s,\params}} \leq \max_{a,i} \abs{\partial_i\langle O_a\rangle_{s,\params}}\leq \max_a \norm{O_a}
\end{equation*}
where the last inequality derives from the parameter-shift rule (Eq.\ (\ref{eq:psr})) formulation of $\partial_i\expval{O_{a}}$ for derivatives w.r.t.\ rotation angles of the PQC and the fact that $\partial_i\expval{O_{a}}$ are simply expectation values $\expval{H_{a,i}}$ with $\norm{H_{a,i}} \leq \norm{O_a}$ for observable weights.\\
Applying the triangular inequality on the right side of Eq.\ (\ref{eq:diff-xai}), we hence have:
\begin{align*}
	\abs{\sum_a x_{a,i} - \sum_a \widetilde{x}_{a,i}} \mathrel{\underset{\beta\varepsilon' \to 0^+}{\leq}} \varepsilon + 4\beta\varepsilon'\max_a\norm{O_a} + \mathcal{O}\left((\beta\varepsilon')^2\varepsilon\right) + \mathcal{O}\left((\beta\varepsilon')^3\right).
\end{align*}
For $\varepsilon' = \varepsilon(4\beta\max_a\norm{O_a})^{-1}$ and using $|\partial_i\langle O_a \rangle_{s,\params} - \partial_i\langle\widetilde{O_a}\rangle_{s,\params}| \leq \varepsilon, \forall a,i$, we finally have:
\begin{equation*} 
	\abs{\partial_i \log{\policy(a|s)} - \partial_i \log{\widetilde{\policy}(a|s)}} \mathrel{\underset{\beta\varepsilon \to 0^+}{\leq}} 3\beta\varepsilon + \mathcal{O}(\beta\varepsilon^3) \quad \forall i \qedhere
\end{equation*}
\end{proof}

\section{The role of trainable observables in \textsc{softmax-PQC} policies\label{sec:trainable-observables}}

In Sec.~\ref{sec:model-def}, we presented a general definition of the \textsc{softmax-PQC} observables $O_a = \sum_i w_{a,i} H_{a,i}$ in terms of an arbitrary weighted sum of Hermitian matrices $H_{a,i}$. In this appendix, we clarify the role of such a decomposition.

\subsection{Training the eigenbasis and the eigenvalues of an observable}

Consider a projective measurement defined by an observable $O = \sum_m \alpha_m P_m$, to be performed on a quantum state of the form $V(\params) \ket{\psi}$, where $V(\params)$ denotes a (variational) unitary. Equivalently, one could also measure the observable $V^{\dagger}(\params)OV(\params)$ on the state $\ket{\psi}$. Indeed, these two measurements have the same probabilities $p(m) = \bra{\psi} V^{\dagger}(\params) P_m V(\params) \ket{\psi}$ of measuring any outcome $\alpha_m$. Note also that the possible outcomes $\alpha_m$ (i.e., the eigenvalues of the observable $O$) remain unchanged.

From this observation, it is then clear that, by defining an observable $O = \sum_m \alpha_m P_m$ using projections $P_m$ on each computational basis state of the Hilbert space $\mathcal{H}$ and arbitrary eigenvalues $\alpha_m\in\mathbb{R}$,\break the addition of a \emph{universal} variational unitary $V(\params)$ prior to the measurement results in a family of observables $\{V^{\dagger}(\params) O V(\params)\}_{\bm{\params},\bm{\alpha}}$ that covers all possible Hermitian observables in $\mathcal{H}$. Moreover, in this setting, the parameters that define the eigenbasis of the observables $V^{\dagger}(\params) O V(\params)$ (i.e., $\params$) are completely distinct from the parameters that define their eigenvalues (i.e., $\bm{\alpha}$). This is not the case for observables that are expressed as linear combinations of non-commuting matrices, for instance.

In our simulations, we consider restricted families of observables. In particular, we take the Hermitian matrices $H_{a,i}$ to be diagonal in the computational basis (e.g., tensor products of Pauli-$Z$ matrices), which means they, as well as $O_a$, can be decomposed in terms of projections on the computational basis states. However, the resulting eigenvalues $\bm{\alpha}$ that we obtain from this decomposition are in our case degenerate, which means that the weights $\bm{w}_{a}$ underparametrize the spectrums of the observables $O_a$.\break Additionally, the last variational unitaries $V_\text{var}(\phis_{L})$ of our PQCs are far from universal, which restricts the accessible eigenbasis of all variational observables $V^\dagger_\text{var}(\phis_{L})O_aV_\text{var}(\phis_{L})$.

\subsection{The power of universal observables}

Equivalently to the universal family of observables $\{V^{\dagger}(\params) O V(\params)\}_{\bm{\params},\bm{\alpha}}$ that we defined in the previous section, one can construct a family of observables $\{O_{\bm{w}} = \sum_i w_{i} H_{i}\}_{\bm{w}}$ that parametrizes all Hermitian matrices in $\mathcal{H}$ (e.g., by taking $H_i$ to be single components of a Hermitian matrix acting on $\mathcal{H}$).\break Note that this family is covered by our definition of \textsc{softmax-PQC} observables. Now, given access to data-dependent quantum states $\ket{\psi_{s}}$ that are expressive enough (e.g., a binary encoding of the input $s$, or so-called universal quantum feature states \cite{goto21}), one can approximate arbitrary functions of $s$ using expectations values of the form $\bra{\psi_{s}}O_{\bm{w}}\ket{\psi_{s}}$. This is because the observables $O_{\bm{w}}$ can encode an arbitrary quantum computation. Hence, in the case of our \textsc{softmax-PQC}s, one could use such observables and such encodings $\ket{\psi_{s}}$ of the input states $s$ to approximate any policy $\pi(a|s)$ (using an additional softmax), without the need for any variational gates in the PQC generating $\ket{\psi_{s}}$.

As we mentioned in the previous section, the observables that we consider in this work are more restricted, and moreover, the way we encode the input states $s$ leads to non-trivial encodings $\ket{\psi_{s,\phis,\lambdas}}$ in general. This implies that the variational parameters $\phis,\lambdas$ of our PQCs have in general a non-trivial role in learning good policies. One can even show here that these degrees of freedom are sufficient to make such PQCs universal function approximators \cite{perez21}.

\section{Environments specifications and hyperpameters\label{sec:env-spec-hyper}}

In Table \ref{table:env-desc}, we present a specification of the environments we consider in our numerical simulations. These are standard benchmarking environments from the OpenAI Gym library \cite{brockman16}, described in Ref.\ \cite{gym20}, PQC-generated environments that we define in Sec.\ \ref{sec:PQC-env}, and the CognitiveRadio environment of Ref.~\cite{chen20} that we discuss in Appendix \ref{sec:PQC-DNN}.

\begin{table*}[h]
	\centering
	\includegraphics[width=\linewidth]{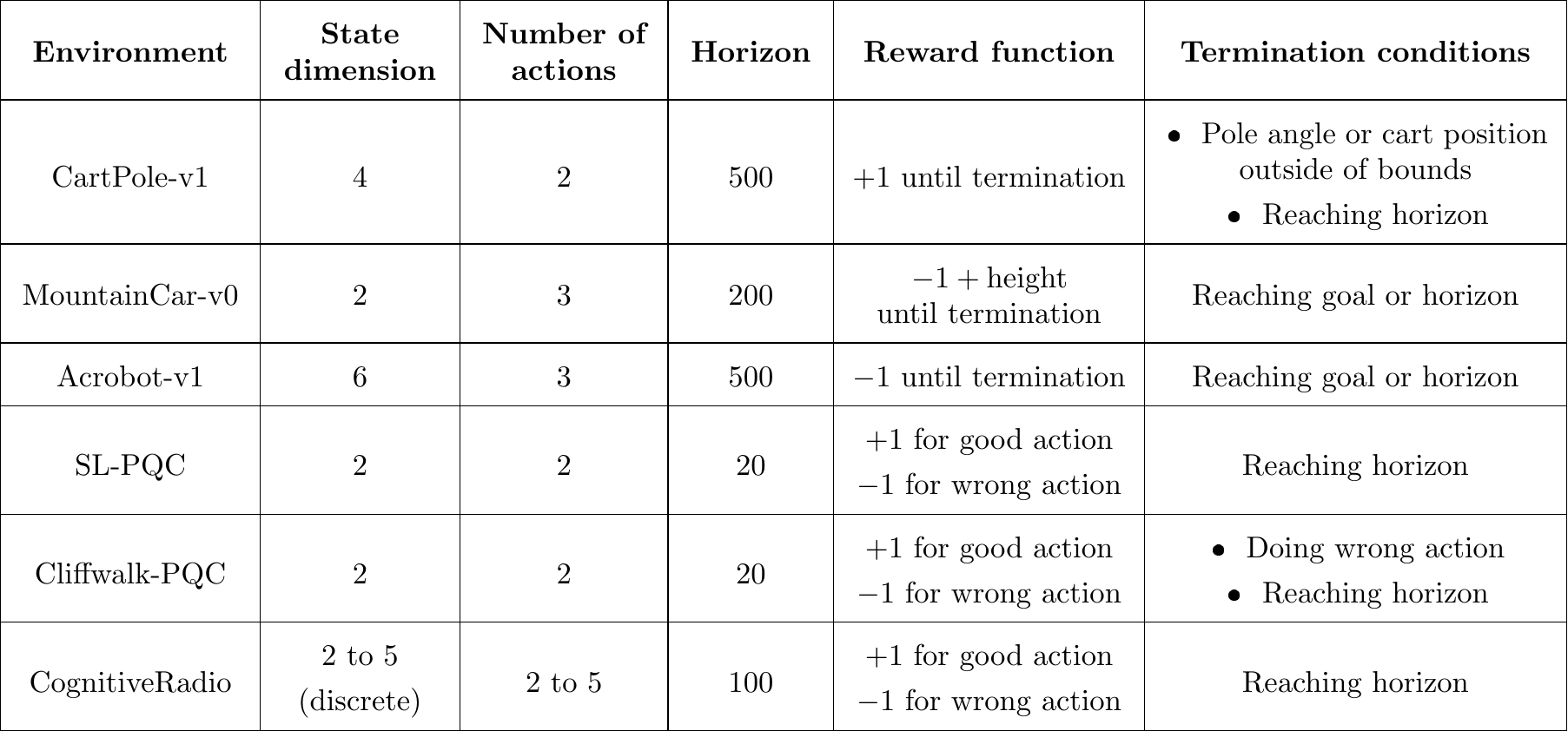}
	\vspace{-1em}
	\caption{\textbf{Environments specifications. }The reward function of Mountaincar-v0 has been modified compared to the standard specification of OpenAI Gym \cite{brockman16}, similarly to Ref.\ \cite{duan16}.}
	\label{table:env-desc}
\end{table*}

In Tables \ref{table:hyperparam-1} and \ref{table:hyperparam-2}, we list the hyperparameters used to train our agents on the various environments we consider. All agents use an ADAM optimizer. For the plots presented in this manuscript, all quantum circuits were implemented using the Cirq library \cite{cirq18} in Python and simulated using a Qulacs backend \cite{suzuki20} in C\texttt{++}. For the tutorial \cite{tfq21}, the TensorFlow Quantum library \cite{broughton20} was used.\\
All simulations were run on the LEO cluster (more than 3000 CPUs) of the University of Innsbruck, with an estimated total compute time (including hyperparametrization) of 20 000 CPU-hours.

\section{Deferred plots and shape of policies learned by PQCs v.s.\ DNNs\label{sec:PQC-DNN}}

\subsection{Influence of architectural choices on \textsc{raw-PQC} agents}

In Fig.\ \ref{fig:architecture-compare2}, we run a similar experiment to that of Sec.\ \ref{sec:architecture-sim} in the main text, but on \textsc{raw-PQC} agents instead of \textsc{softmax-PQC} agents. We observe that both increasing the depth of the PQCs and training the scaling parameters $\lambdas$ have a similar positive influence on the learning performance, and even more pronounced than for \textsc{softmax-PQC} agents. Nonetheless, we also observe that, even at greater depth, the final performance, as well as the speed of convergence, of \textsc{raw-PQC} agents remain limited compared to that of \textsc{softmax-PQC} agents.

\begin{figure}[h]
	\subfloat{\label{fig:cartpole-raw}\includegraphics[width=0.33\linewidth, valign=c]{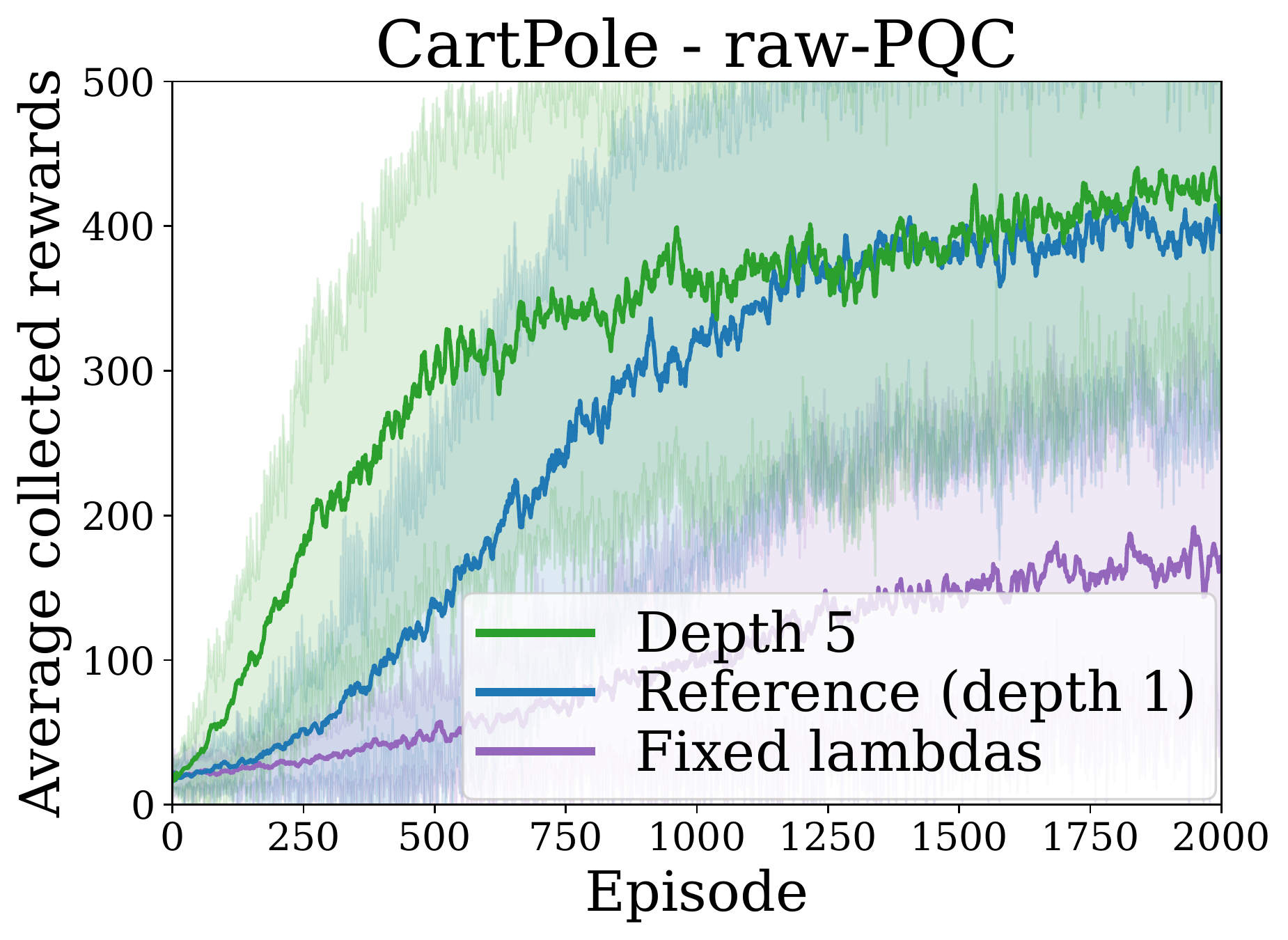}}\hspace{0em}%
	\subfloat{\label{fig:mountaincar-raw}\includegraphics[width=0.33\linewidth, valign=c]{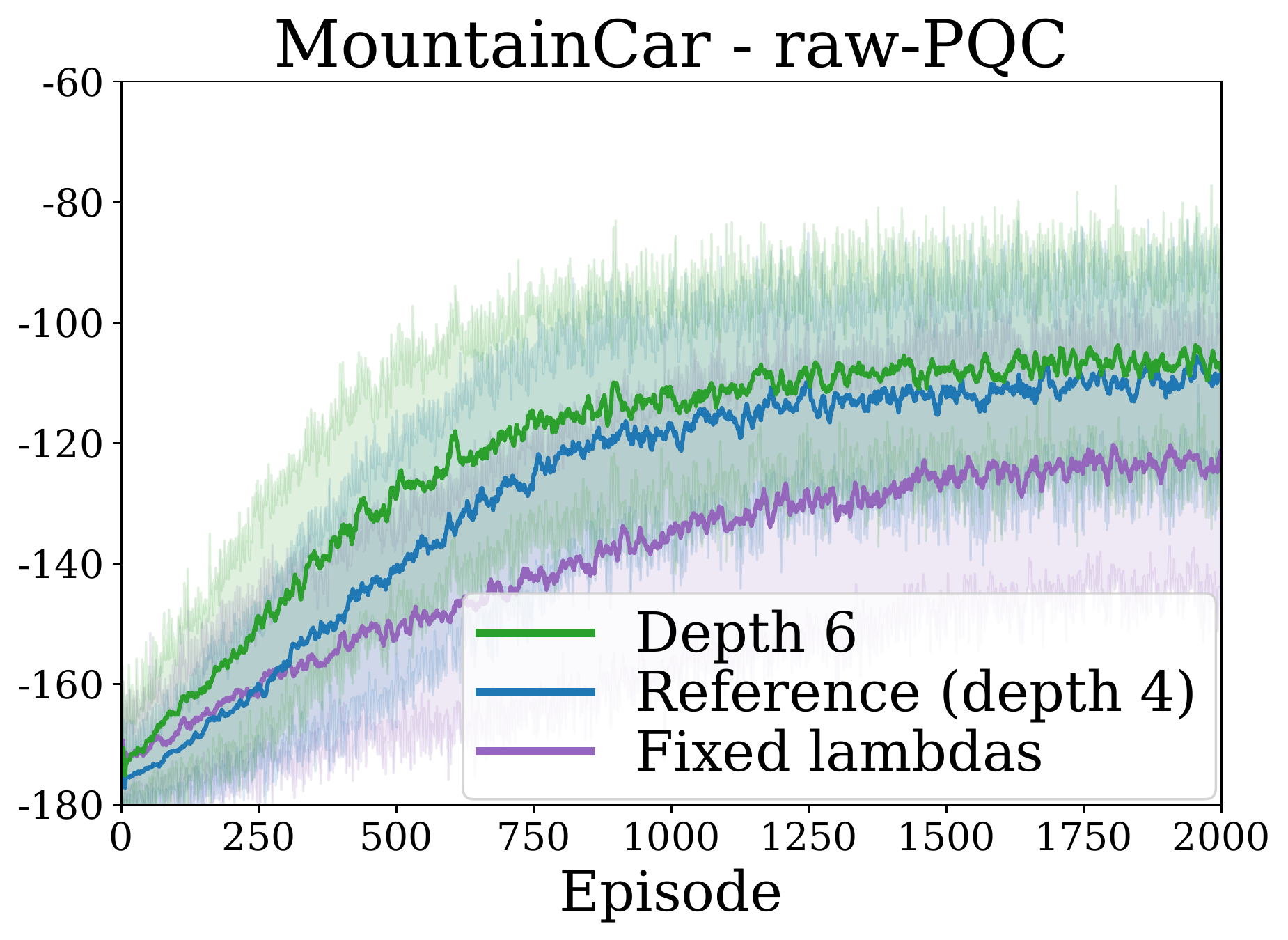}}\hspace{0em}%
	\subfloat{\label{fig:acrobot-raw}\includegraphics[width=0.33\linewidth, valign=c]{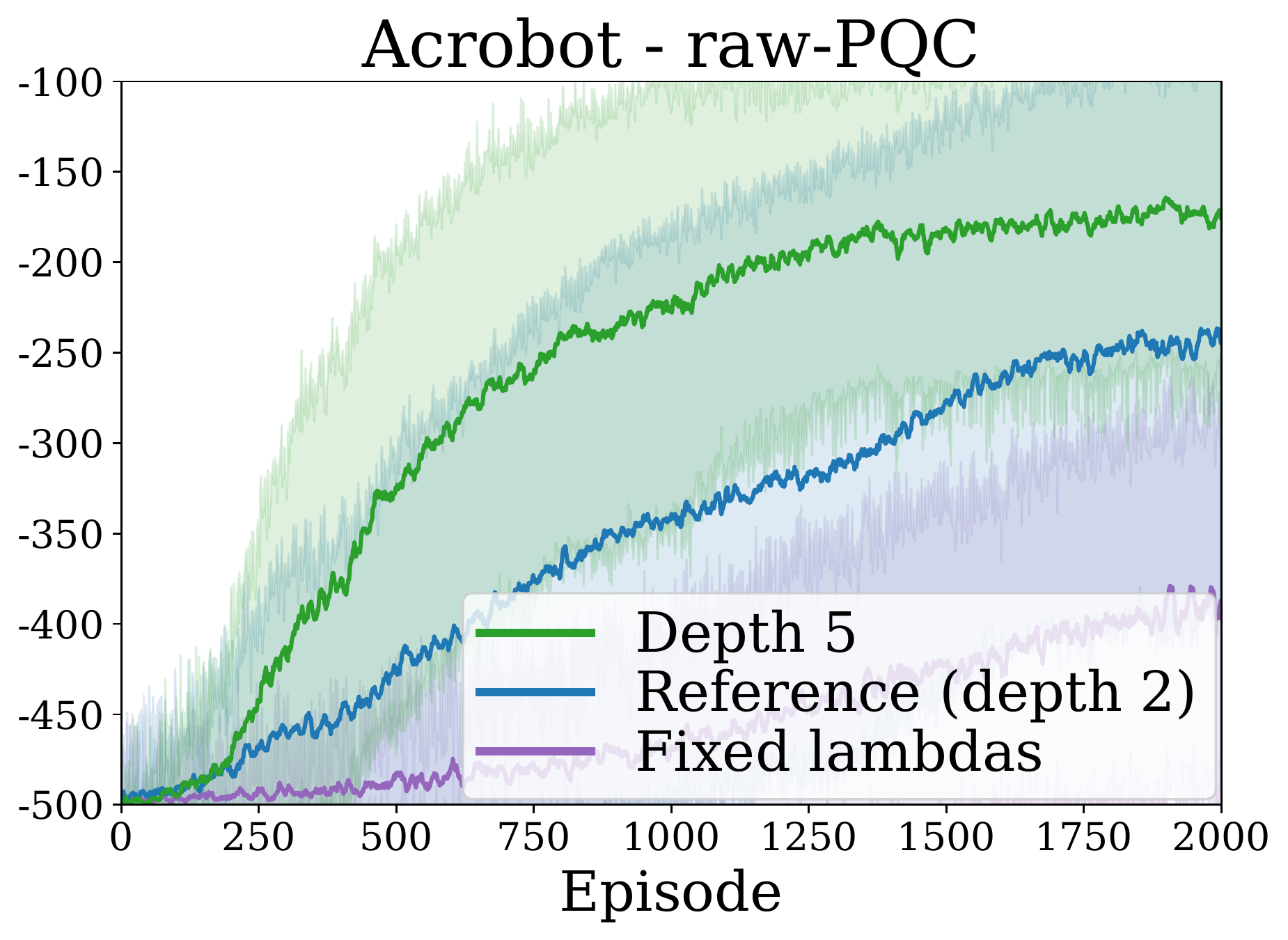}}\hspace{0em}
  \caption{\textbf{Influence of the model architecture for \textsc{raw-PQC} agents. }The blue curves in each plot correspond to the learning curves from Fig.\ \ref{fig:softmax-vs-raw} and are taken as a reference. }
  \label{fig:architecture-compare2}
\end{figure}

\begin{table*}[t]
	\centering
	\includegraphics[width=0.98\textwidth]{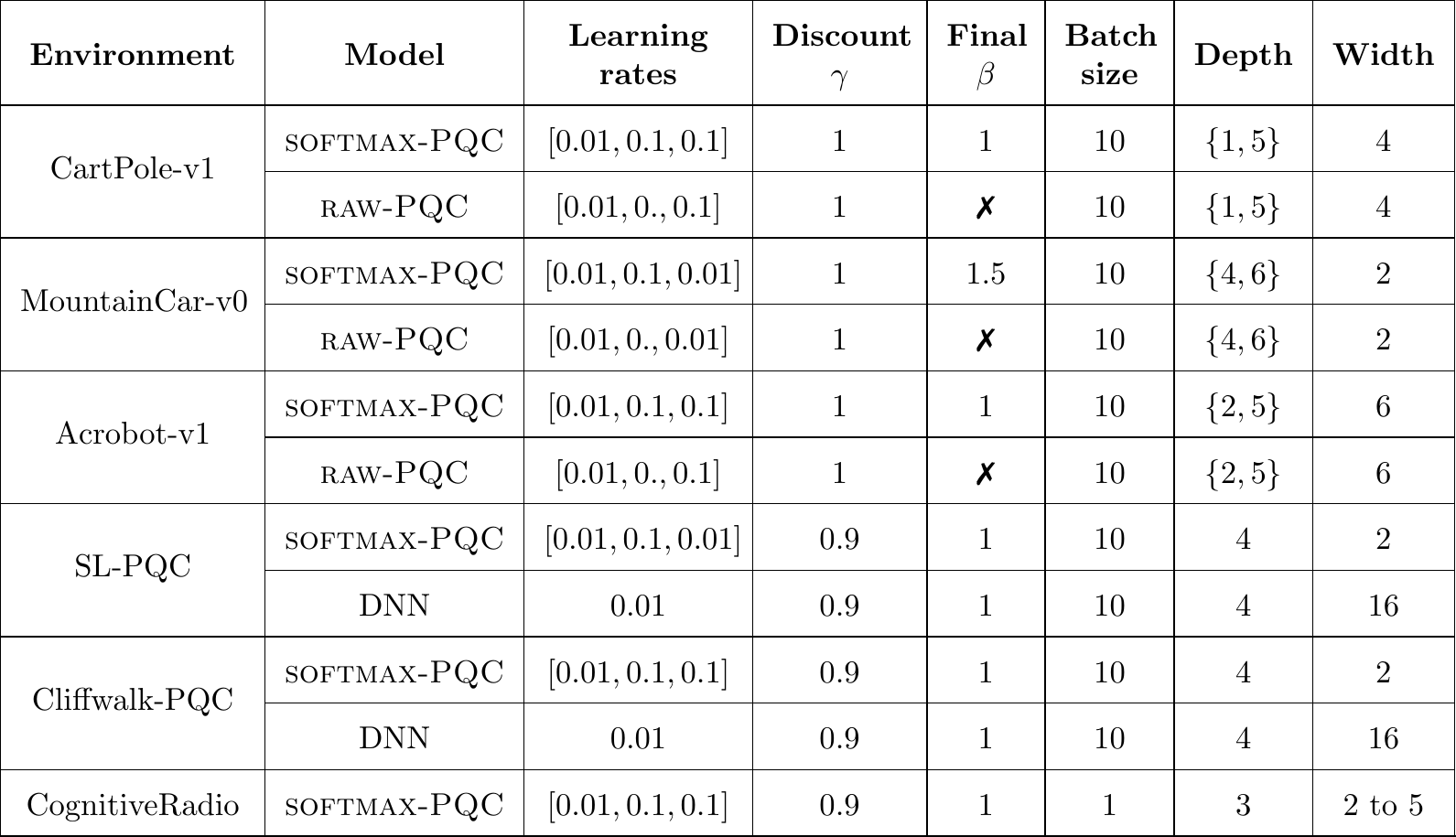}
	\vspace{-0.1em}
	\caption{\textbf{Hyperparmeters 1/2. }For PQC policies, we choose 3 distinct learning rates $[\alpha_{\phis}, \alpha_{\weights}, \alpha_{\lambdas}]$ for rotation angles $\phis$, observable weights $\weights$ and scaling parameters $\lambdas$, respectively. For \textsc{softmax-PQC}s, we take a linear annealing schedule for the inverse temperature parameter $\beta$ starting from $1$ and ending up in the final $\beta$. The batch size is counted in number of episodes used to evaluate the gradient of the value function. Depth indicates the number of encoding layers $D_\text{enc}$ for PQC policies, or the number of hidden layers for a DNN policy. Width corresponds to the number of qubits $n$ on which acts a PQC (also equal to the dimension $d$ of the environment's state space), or the number of units per hidden layer for a DNN.}
	\label{table:hyperparam-1}
\end{table*}

\begin{table*}[t]
	\centering
	\includegraphics[width=0.98\textwidth]{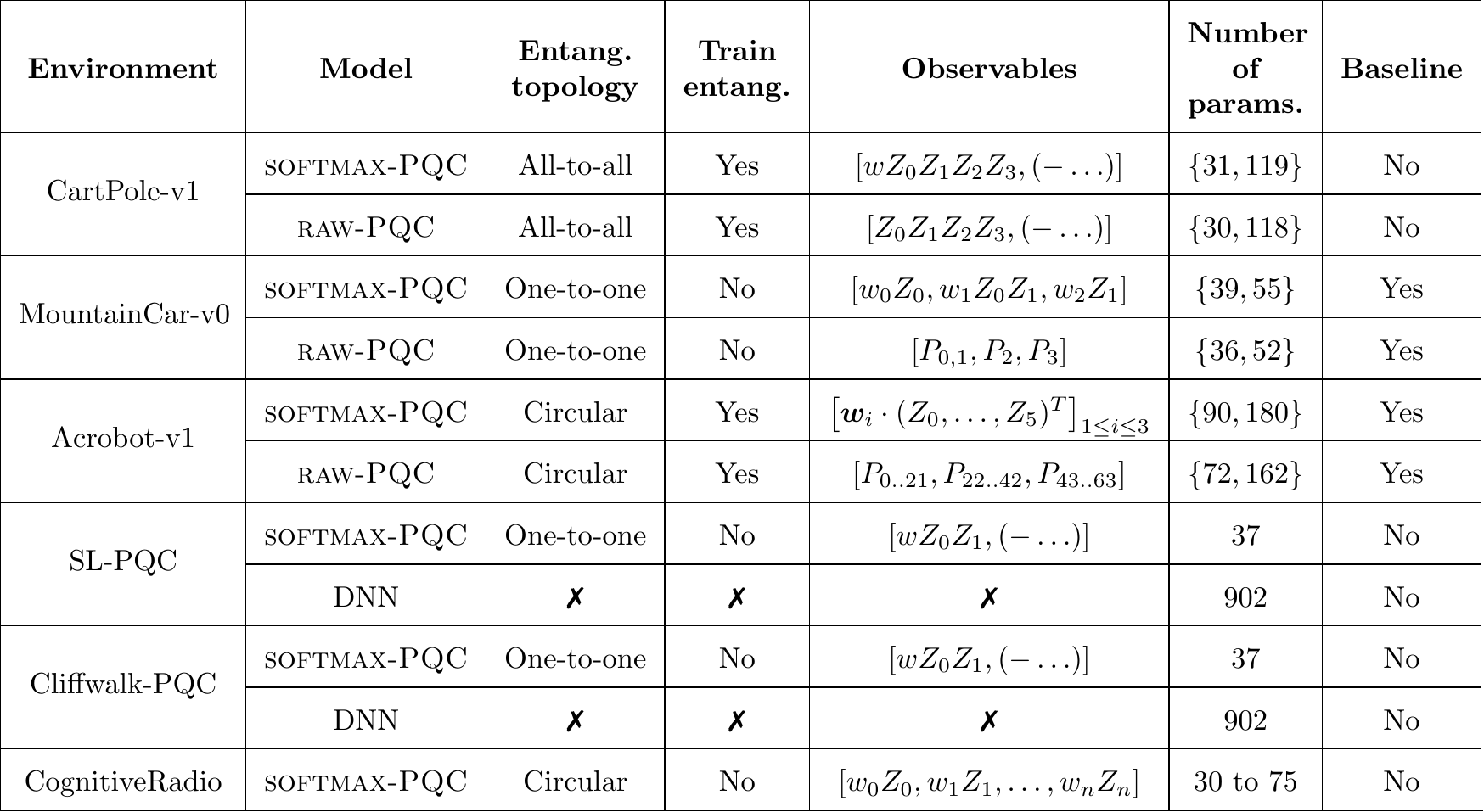}
	\vspace{-0.1em}
	\caption{\textbf{Hyperparmeters 2/2. }We call entangling layer a layer of 2-qubit gates in the PQC. Circular and all-to-all topologies of entangling layers are equivalent for $n=2$ qubits, so we call them one-to-one in that case. When trained, entangling layers are composed of $R_{zz} = e^{-i\theta (Z\otimes Z)/2}$ rotations, otherwise, they are composed of Ctrl-$Z$ gates. For policies with $2$ actions, the same observable, up to a sign change, is used for both actions. $Z_i$ refers to a Pauli-$Z$ observable acting on qubit $i$, while $P_{i..j}$ indicates a projection on basis states $i$ to $j$. In the experiments of Sec.\ \ref{sec:architecture-sim}, when the weights of the \textsc{softmax-PQC} are kept fixed, the observables used for MountainCar-v0 and Acrobot-v1 are $[Z_0, Z_0Z_1,Z_1]$, and those used for CartPole-v1 are $[Z_0Z_1Z_2Z_3, -Z_0Z_1Z_2Z_3]$. The different number of parameters in a given row correspond to the different depths in that same row in Table \ref{table:hyperparam-1}.}
	\label{table:hyperparam-2}
\end{table*}

\clearpage

\subsection{Shape of the policies learned by PQCs v.s.\ DNNs}

\textbf{In CartPole-v1 }
The results of the Sec.\ \ref{sec:benchmark} demonstrate that our PQC policies can be trained to good performance in benchmarking environments. To get a feel of the solutions found by our agents, we compare the \textsc{softmax-PQC} policies learned 
on CartPole to those learned by standard DNNs (with a softmax output layer), which are known to easily learn close-to-optimal behavior on this task. More specifically, we look at the functions learned by these two models, prior to the application of the softmax normalization function (see Eq.\ (\ref{eq:softmax-PQC})). Typical instances of these functions are depicted in Figure \ref{fig:pqc-vs-nn-cartpole}. We observe that, while DNNs learn simple, close to piece-wise linear functions of their input state space, PQCs tend to naturally learn very oscillating functions that are more prone to instability. While the results of Schuld \emph{et al.}\ \cite{schuld19} already indicated that these highly oscillating functions would be natural for PQCs, it is noteworthy to see that these are also the type of functions naturally learned in a direct-policy RL scenario. Moreover, our enhancements to standard PQC classifiers show how to make these highly oscillating functions more amenable to real-world tasks.

\textbf{In PQC-generated environments }
Fig.\ \ref{fig:pqc-vs-nn-2} shows the analog results to Fig.\ \ref{fig:pqc-vs-nn} in the main text but with two different random initializations of the environment-generating PQC. Both confirm our observations. In Fig.\ \ref{fig:pqc-vs-nn-policies}, we compare the policies learned by prototypical \textsc{softmax-PQC} and DNN agents in these PQC-generated environments. We observe that the typical policies learned by DNNs are rather simple, with up to $2$ (or $3$) regions, delimited by close-to-linear boundaries, as opposed to the policies learned by \textsc{softmax-PQC}s, which delimit red from blue regions with wide margins. These observations highlight the inherent flexibility of \textsc{softmax-PQC} policies and their suitability to these PQC-generated environments, as opposed to the DNN (and \textsc{raw-PQC}) policies we consider.

\subsection{Additional numerical simulation on the CognitiveRadio environment}

In a related work on value-based RL with PQCs, the authors of Ref.~\cite{chen20} introduced the CognitiveRadio environment as a benchmark to test their RL agents. In this environment, the agent is presented at each interaction step with a binary vector $(0,0,0,1,0)$ of size $n$ that describes the occupation of $n$ radio channels. Given this state, the agent must select one of the $n$ channels as its communication channel, such as to avoid collision with occupied channels (a $\pm 1$ reward reflects these collisions). The authors of Ref.~\cite{chen20} consider a setting where, in any given state, only one channel is occupied, and its assignment changes periodically over time steps, for an episode length of $100$ steps. While this constitutes a fairly simple task environment with discrete state and action spaces, it allows to test the performance of PQC agents on a family of environments described by their system size $n$ and make claims on the parameter complexity of the PQCs as a function of $n$. As to reproduce the findings of Ref.~\cite{chen20} in a policy-gradient setting, we test the performance of our \textsc{softmax-PQC} agents on this environment. We find numerically (see Fig.~\ref{fig:radio-exp}) that these achieve a very similar performance to the PQC agents of Ref.~\cite{chen20} on the same system sizes they consider ($n=2$ to $5$), using PQCs with the same scaling of number of parameters, i.e., $\mathcal{O}(n)$. 

\begin{figure}[h]
	\subfloat{\includegraphics[width=0.25\linewidth, valign=c]{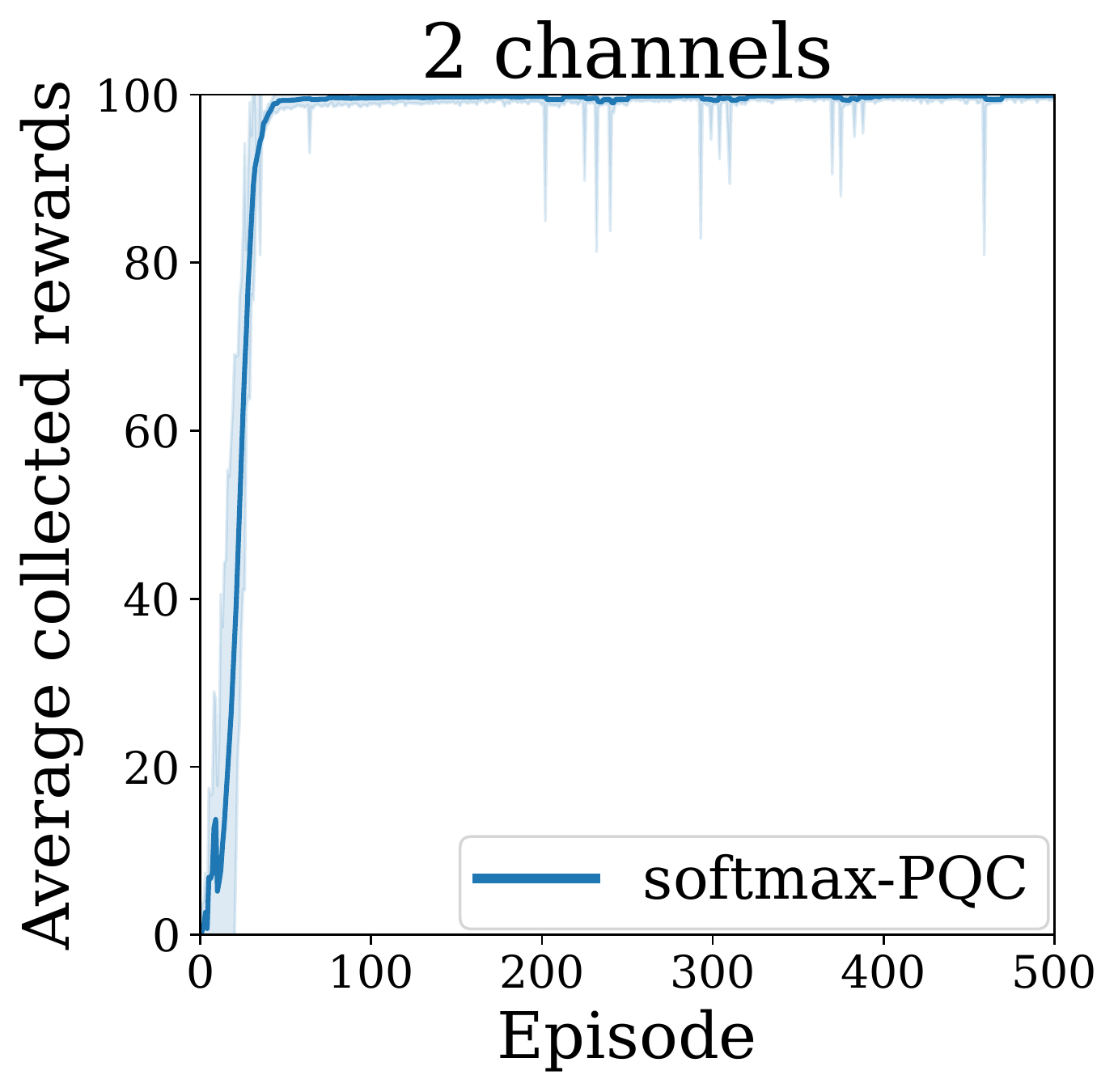}}\hspace{0em}%
	\subfloat{\includegraphics[width=0.25\linewidth, valign=c]{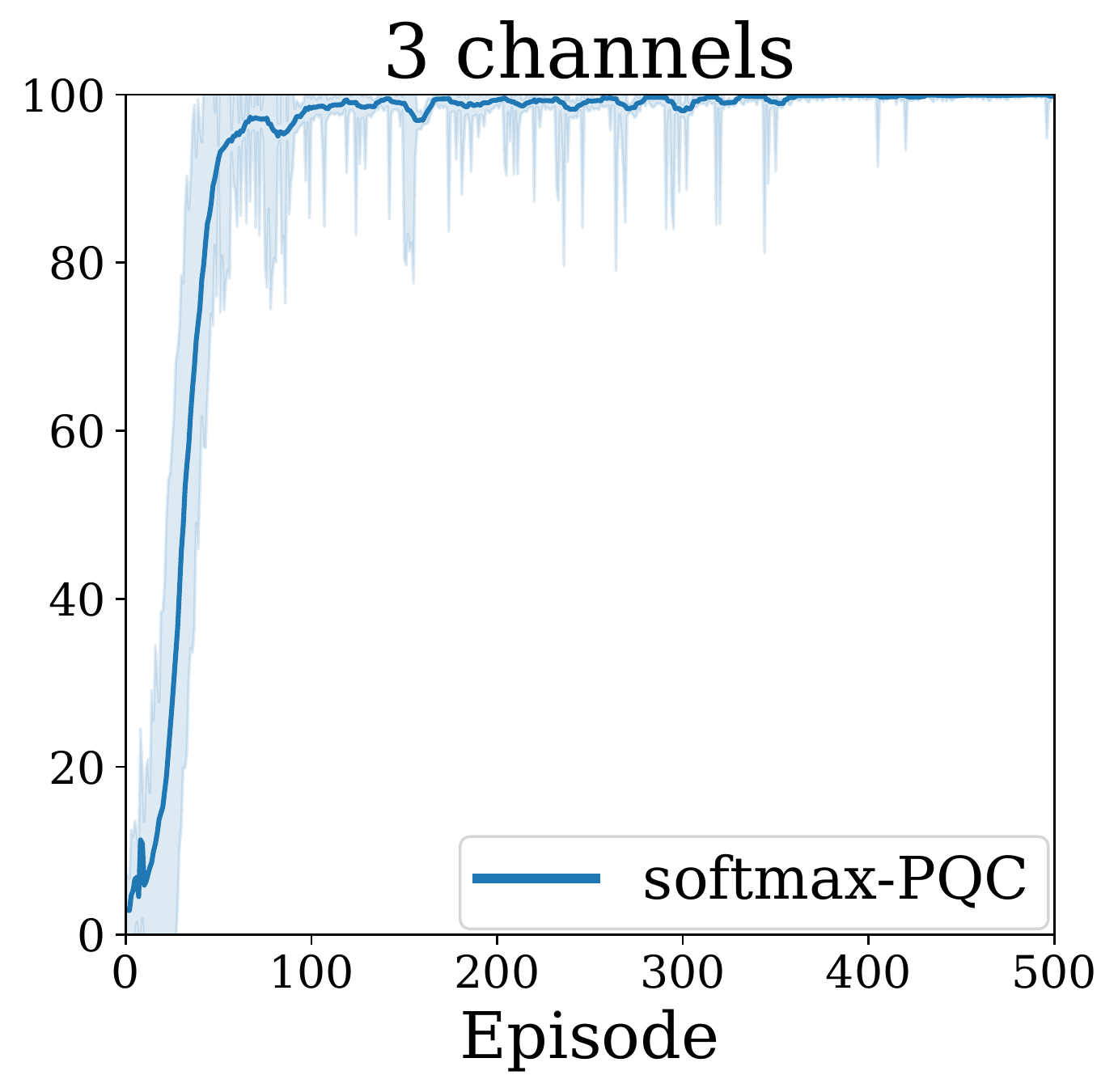}}\hspace{0em}%
	\subfloat{\includegraphics[width=0.25\linewidth, valign=c]{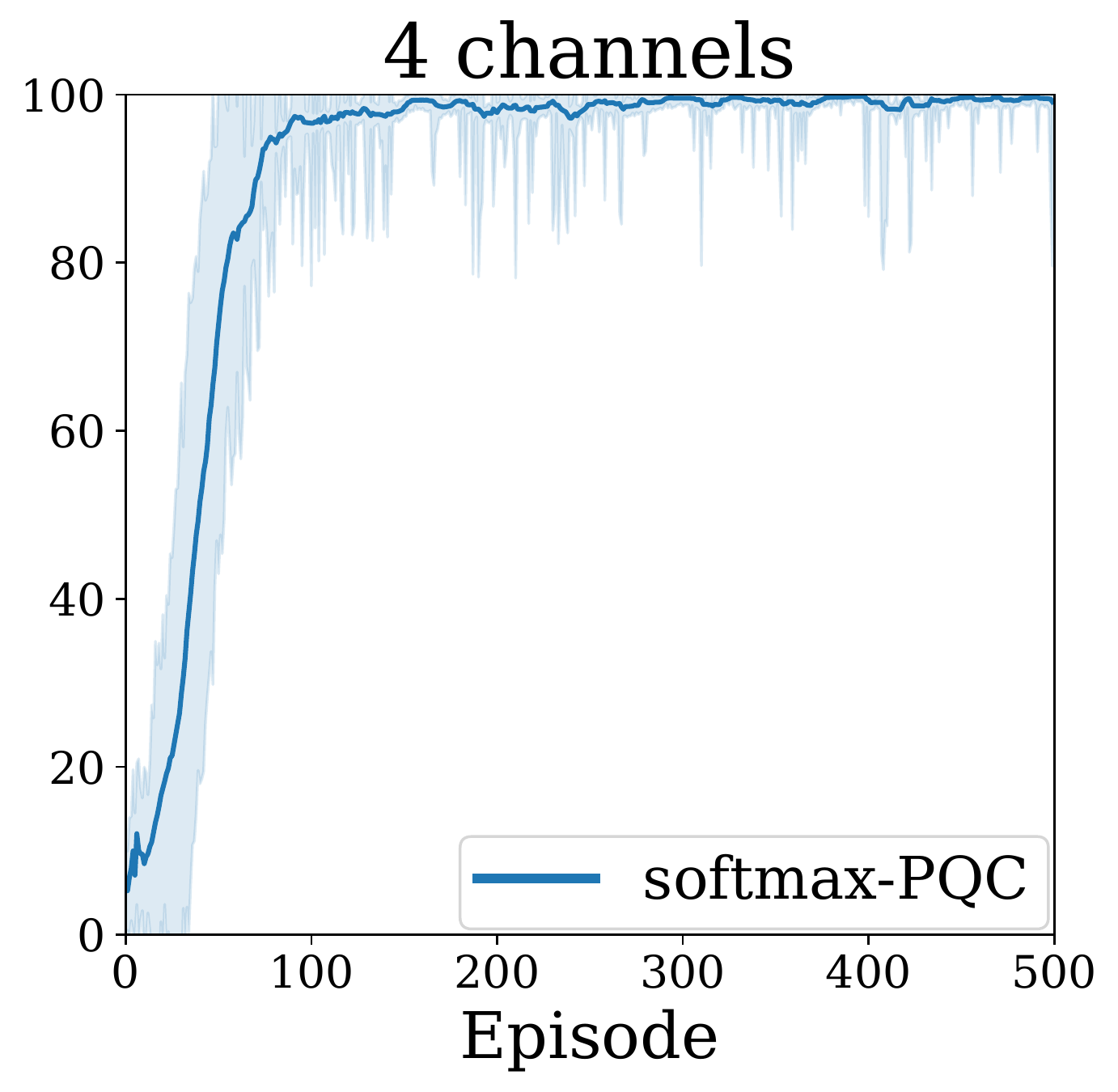}}\hspace{0em}%
	\subfloat{\includegraphics[width=0.25\linewidth, valign=c]{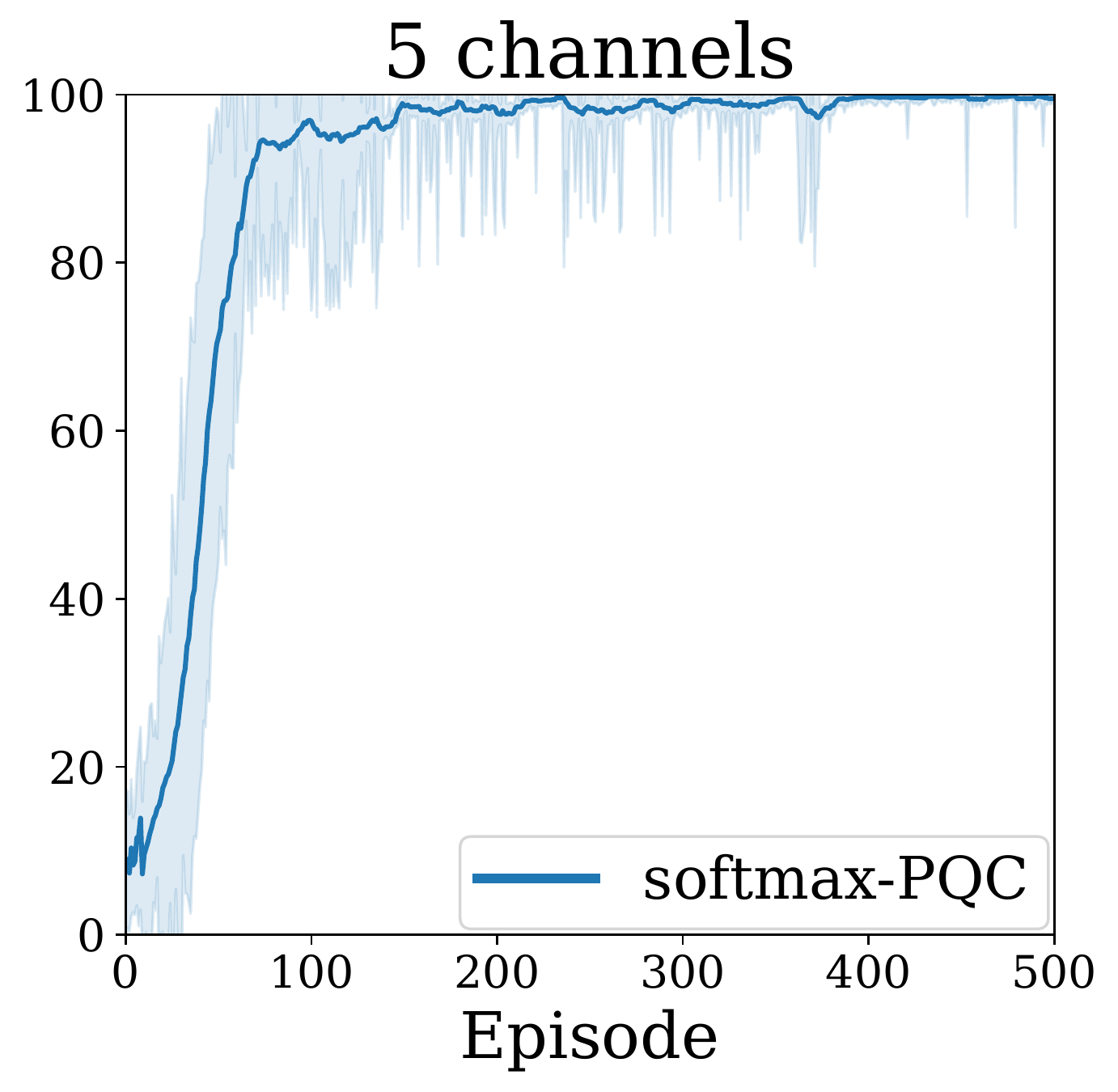}}\hspace{0em}
  \caption{\textbf{Performance of our \textsc{softmax-PQC} agents on the CognitiveRadio environment proposed in Ref.~\cite{chen20}. }Average performance of $20$ agents for system sizes (and number of qubits) $n=2$ to $5$.}
  \label{fig:radio-exp}
\end{figure}

\begin{figure*}[t]
	\centering
	\subfloat[][]{\includegraphics[width=0.96\textwidth, valign=c]{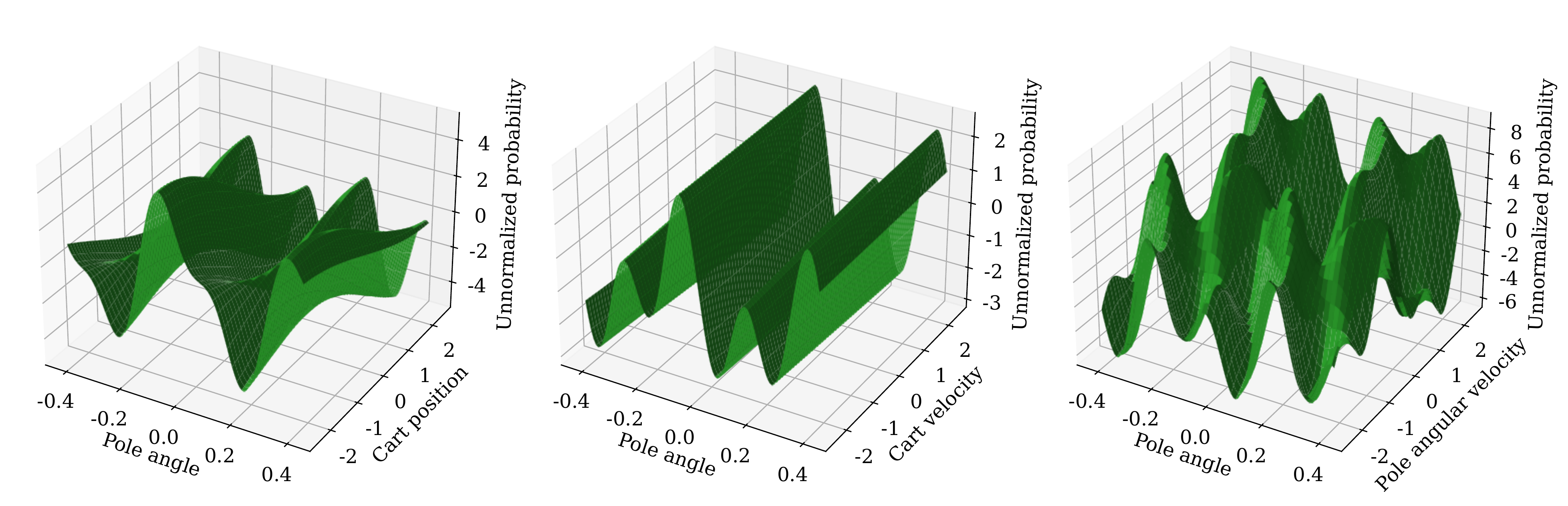}}\hspace{0em}\\ \vspace{-0.75em}
	\subfloat[][]{\includegraphics[width=0.96\textwidth, valign=c]{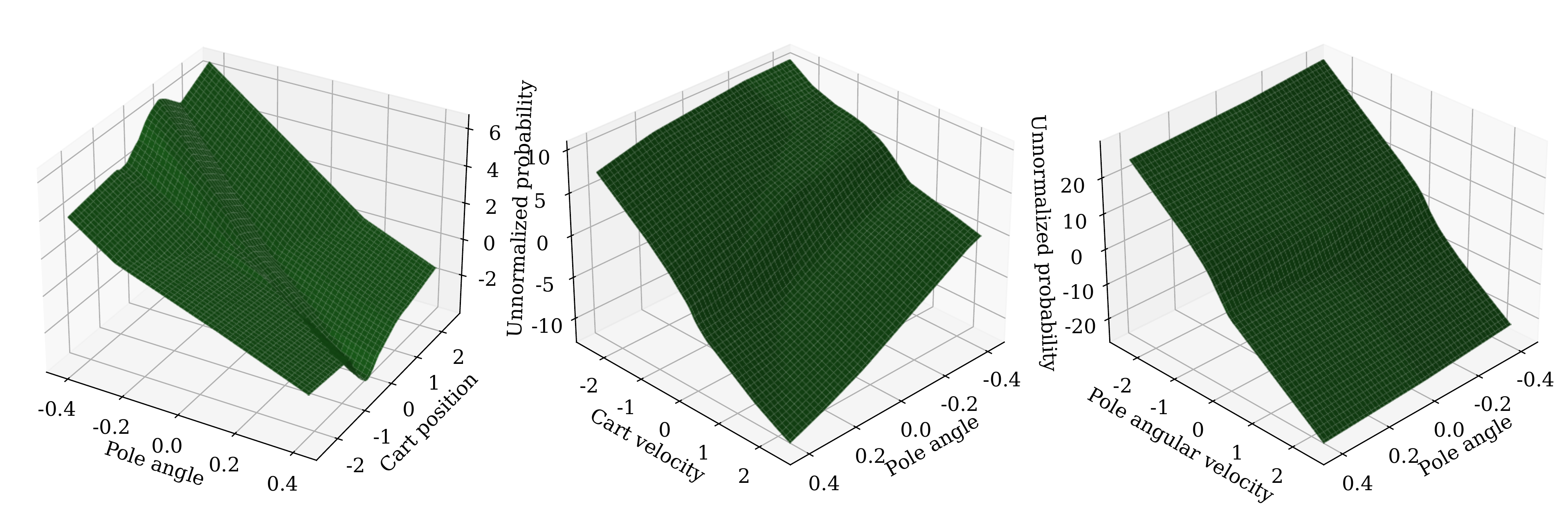}}\hspace{0em}
  \caption{\textbf{Prototypical unnormalized policies learned by \textsc{softmax-PQC} agents and DNN agents in CartPole. }Due to the $4$ dimensions of the state space in CartPole, we represent the unnormalized policies learned by (a) \textsc{softmax-PQC} agents and (b) DNN agents on $3$ subspaces of the state space by fixing unrepresented dimensions to $0$ in each plot. To get the probability of the agent pushing the cart to the left, one should apply the logistic function (i.e., 2-dimensional softmax) $1/(1+exp(-z))$ to the $z$-axis values of each plot.}
  \label{fig:pqc-vs-nn-cartpole}
  \vspace{-0.5em}
\end{figure*}

\begin{figure*}
	\subfloat[][]{\includegraphics[width=0.31\textwidth, valign=c]{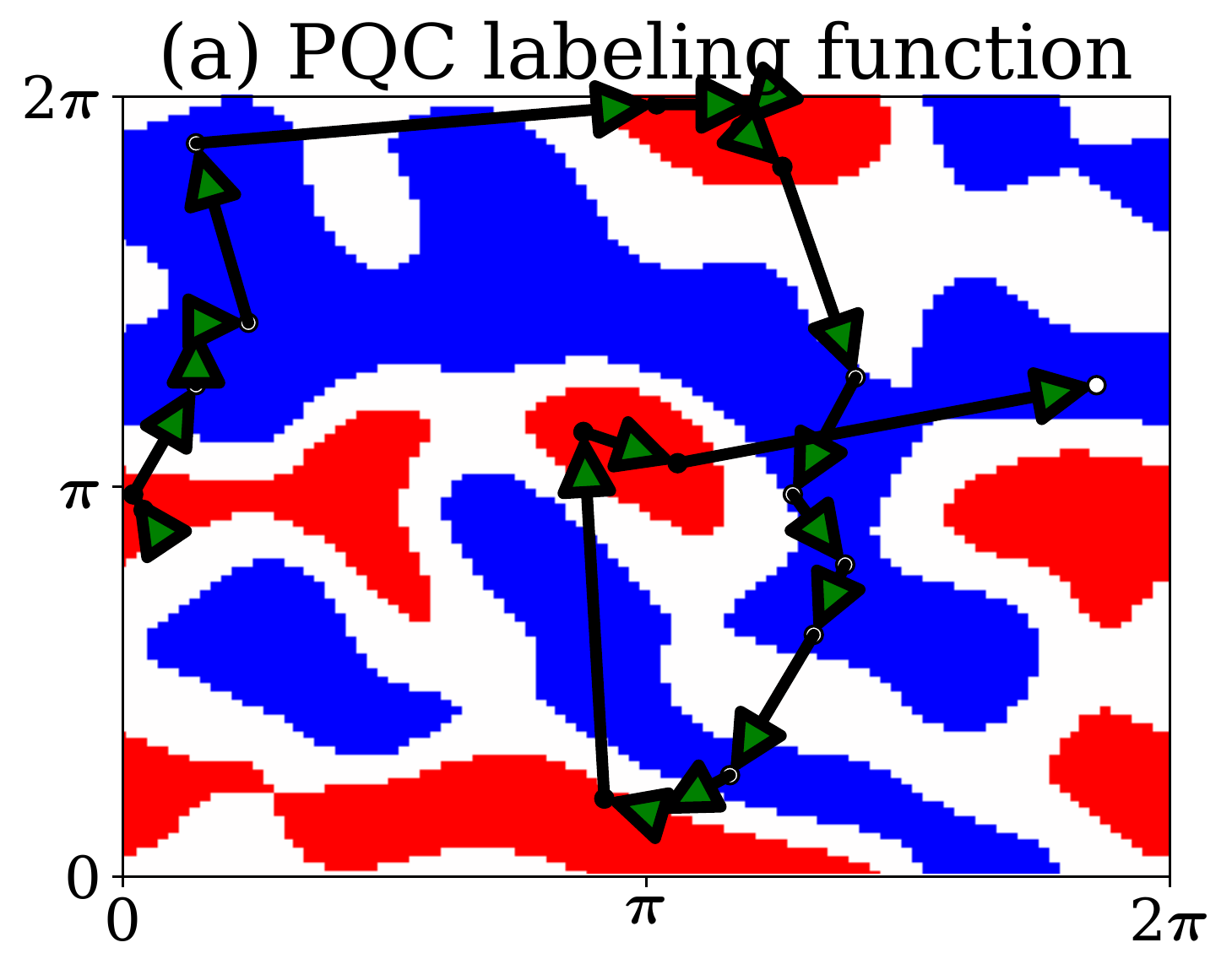}}\hspace{0em}%
	\subfloat[][]{\includegraphics[width=0.34\textwidth, valign=c]{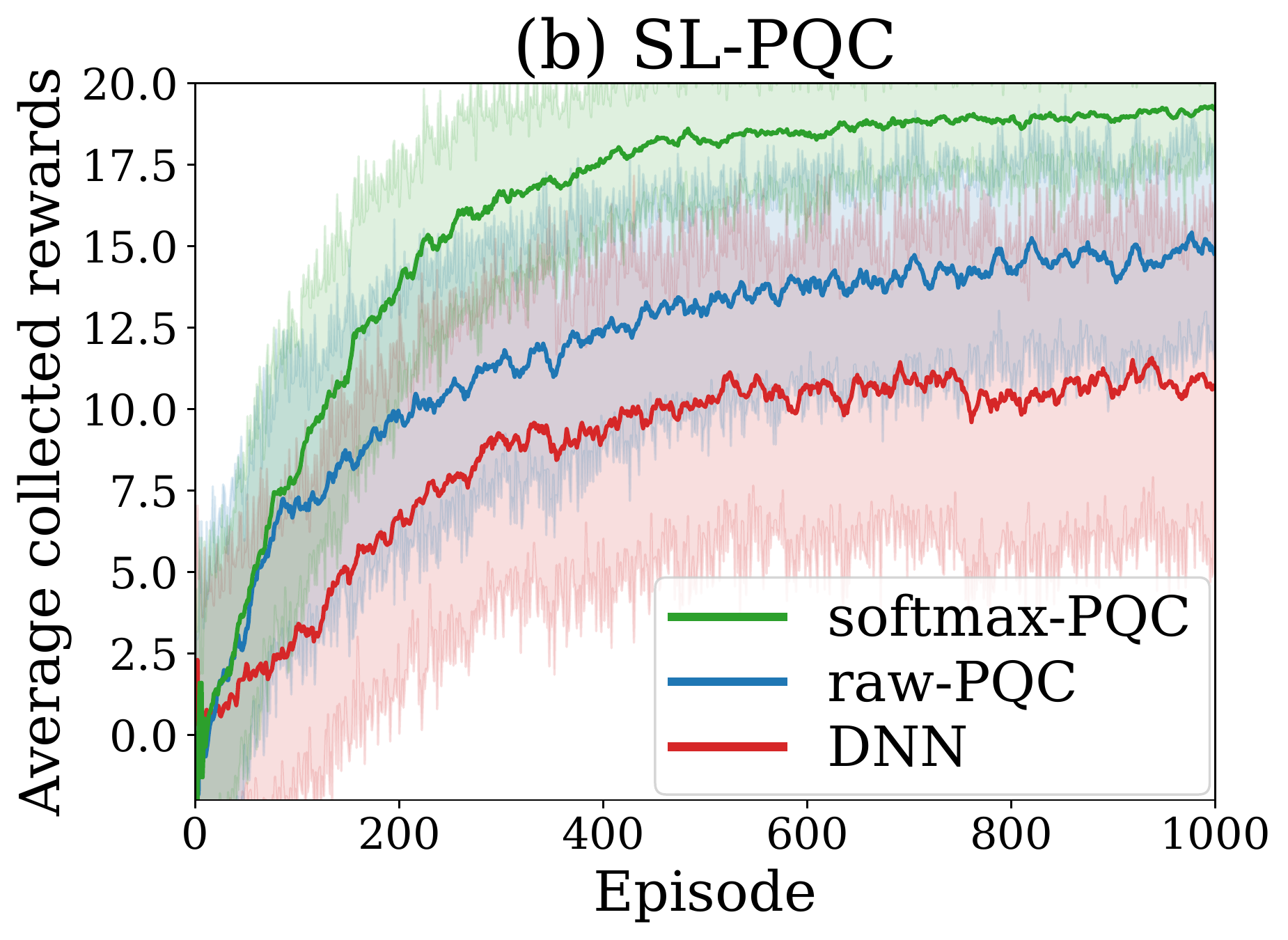}}\hspace{0em}%
	\subfloat[][]{\includegraphics[width=0.34\textwidth, valign=c]{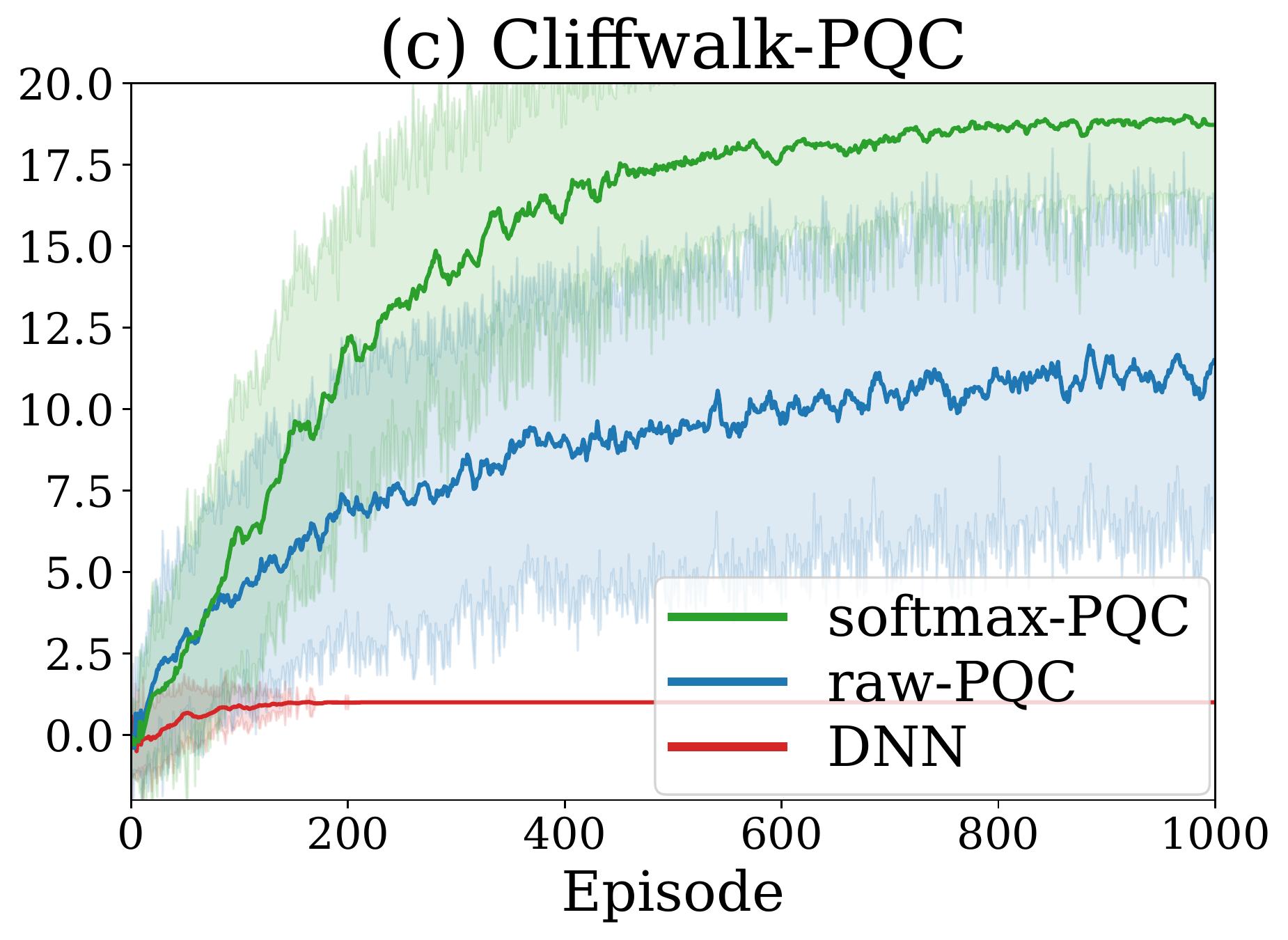}}\hspace{0em}\\%
	\subfloat[][]{\includegraphics[width=0.31\textwidth, valign=c]{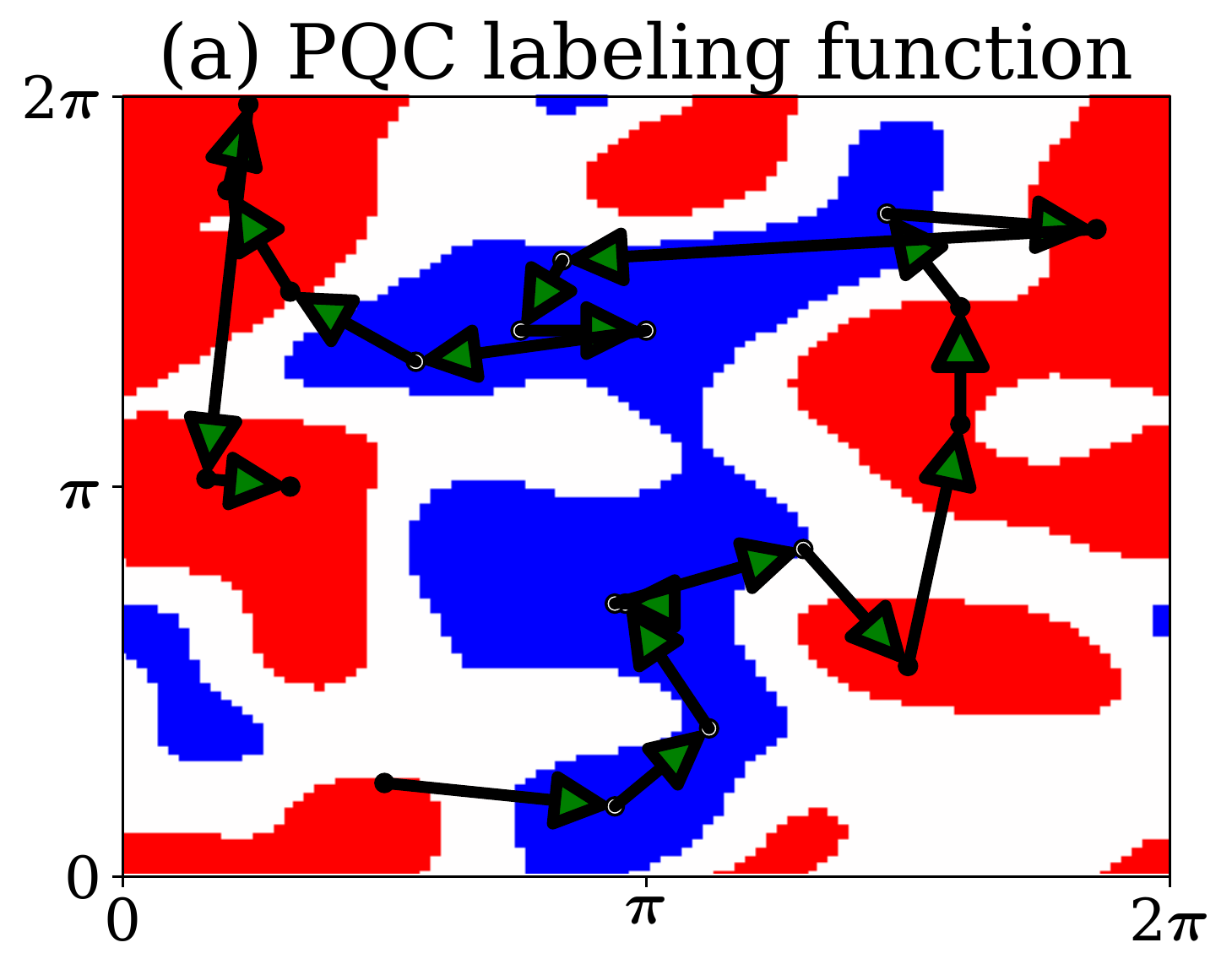}}\hspace{0em}%
	\subfloat[][]{\includegraphics[width=0.34\textwidth, valign=c]{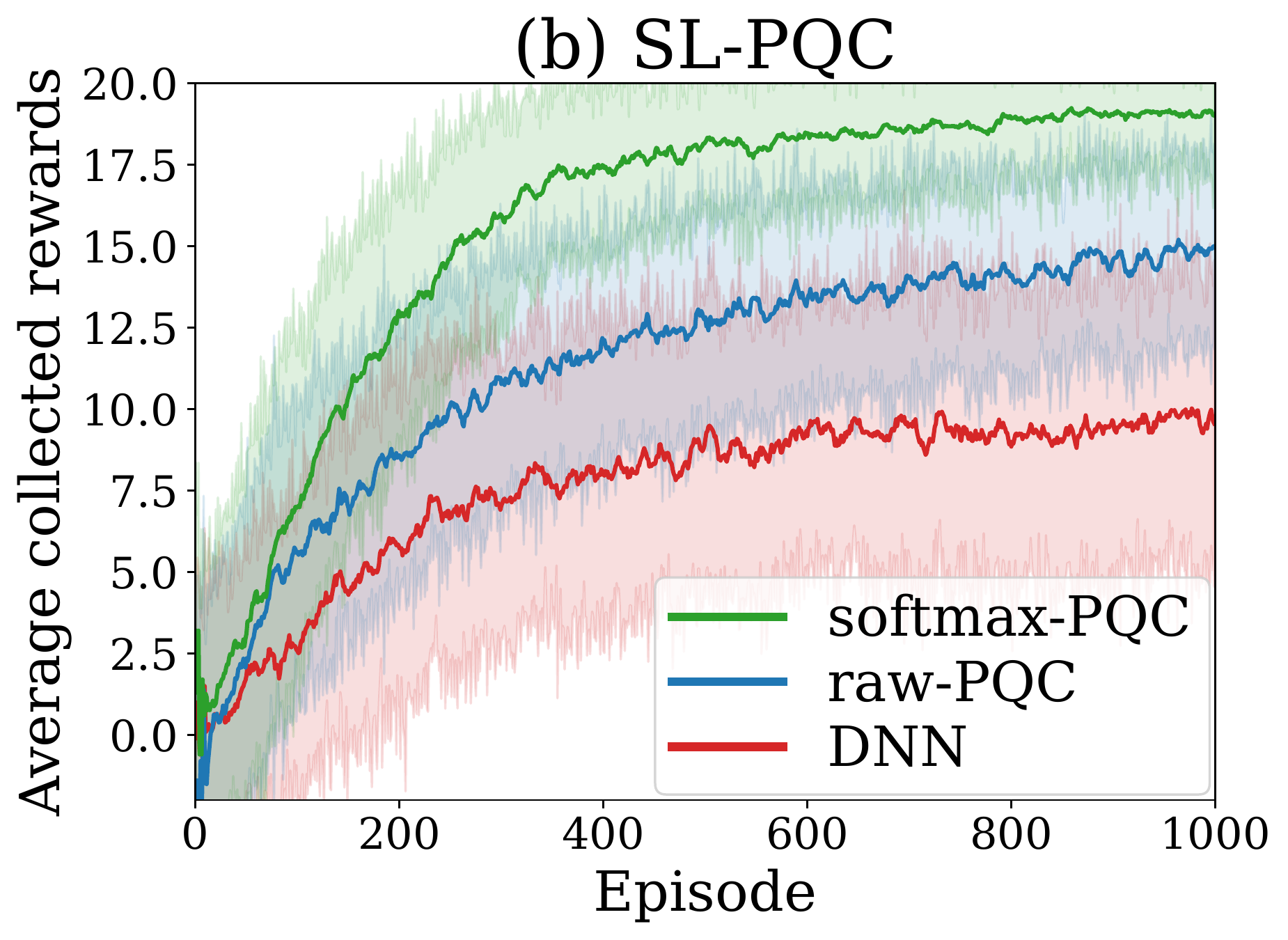}}\hspace{0em}%
	\subfloat[][]{\includegraphics[width=0.34\textwidth, valign=c]{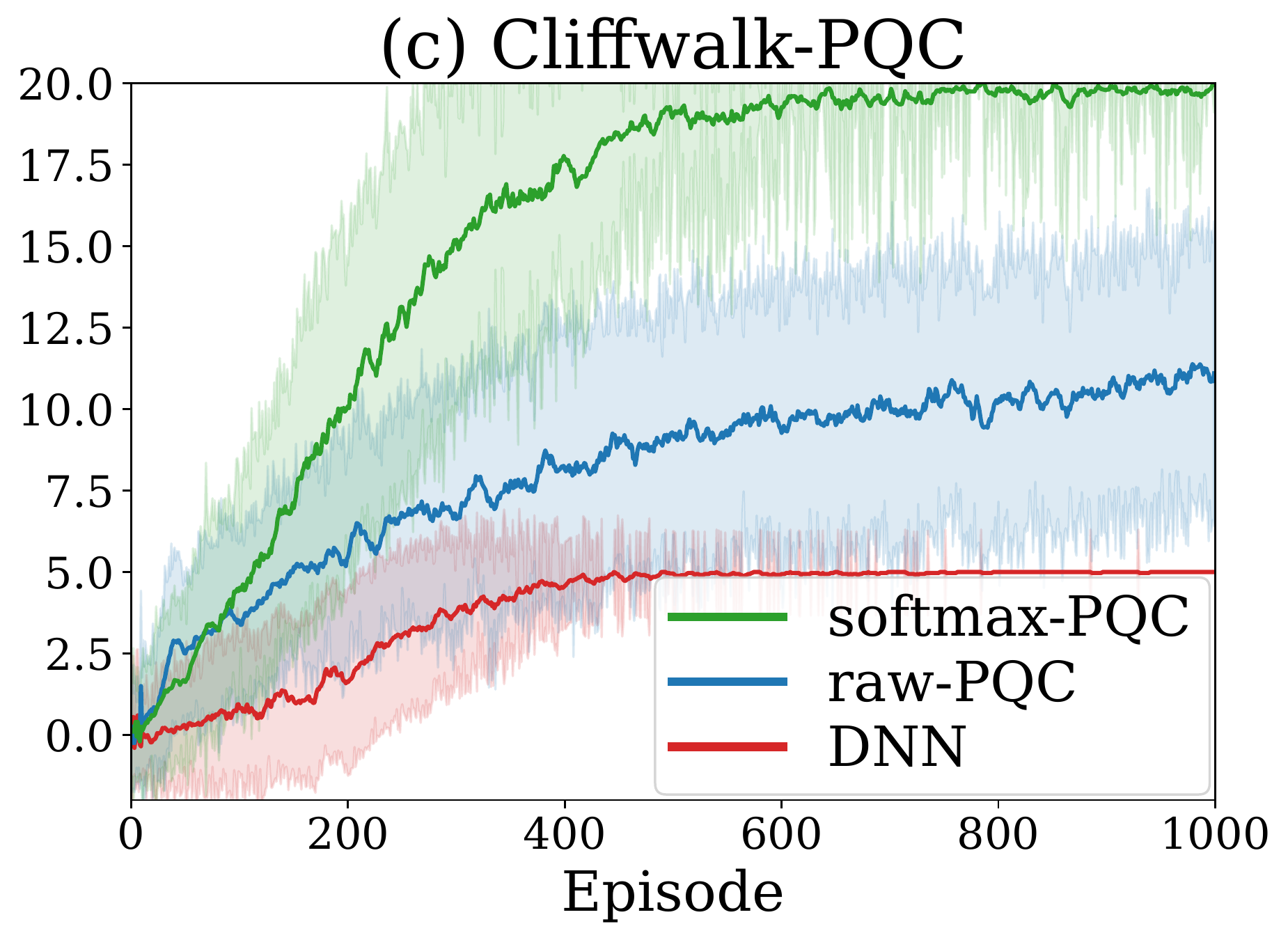}}\hspace{0em}
  \caption{\textbf{Different random initializations of PQC-generated environments and their associated learning curves. }See Fig.\ \ref{fig:pqc-vs-nn} for details. The additional learning curves (20 agents per curve) of randomly-initialized \textsc{raw-PQC} agents highlight the hardness of these environments for PQC policies drawn from the same family as the environment-generating PQCs.}
  \label{fig:pqc-vs-nn-2}
\end{figure*}

\begin{figure*}
	\centering
	\subfloat[][]{\includegraphics[width=0.31\textwidth, valign=c]{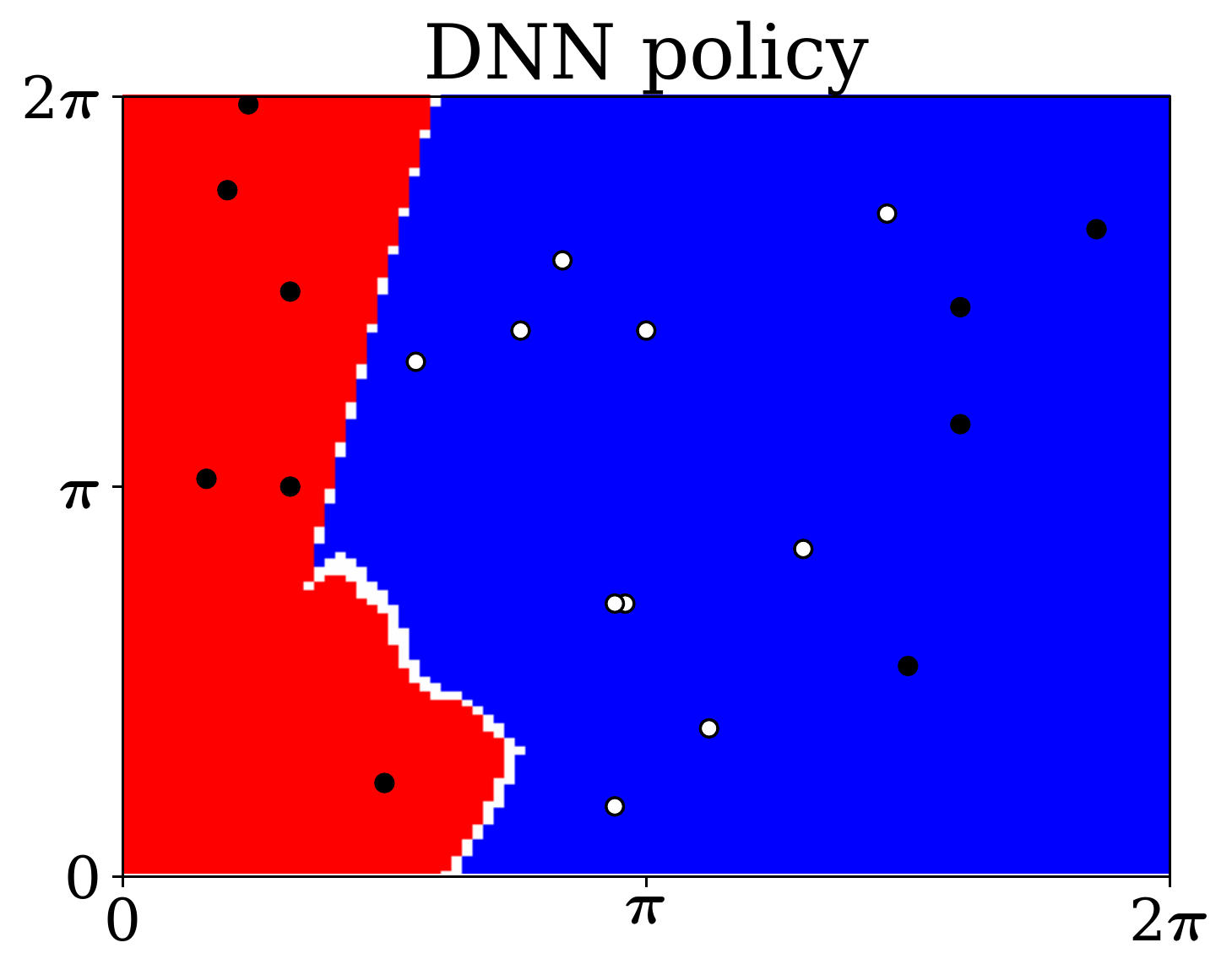}}\hspace{0em}%
	\subfloat[][]{\includegraphics[width=0.31\textwidth, valign=c]{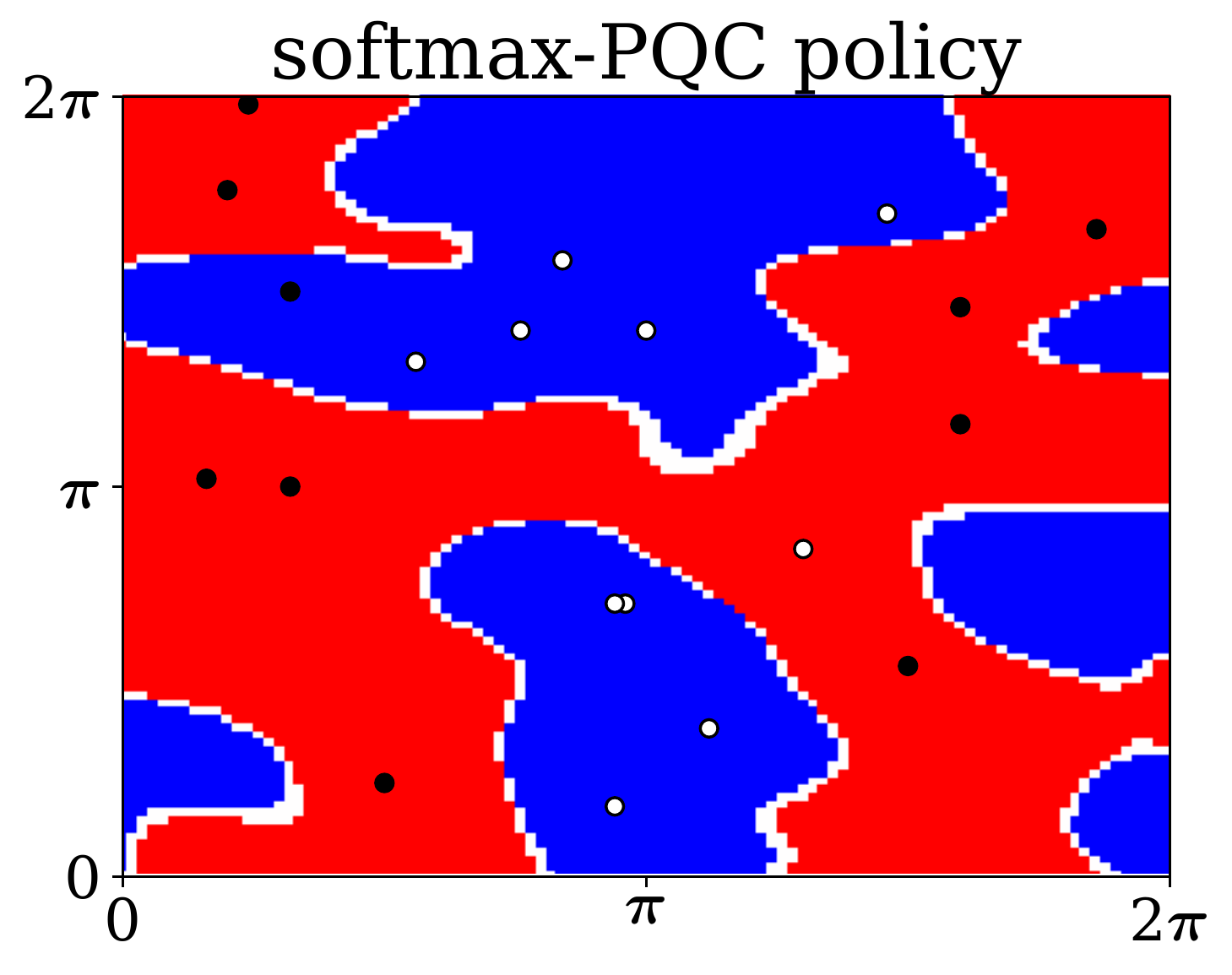}}\hspace{0em}\\
	\subfloat[][]{\includegraphics[width=0.31\textwidth, valign=c]{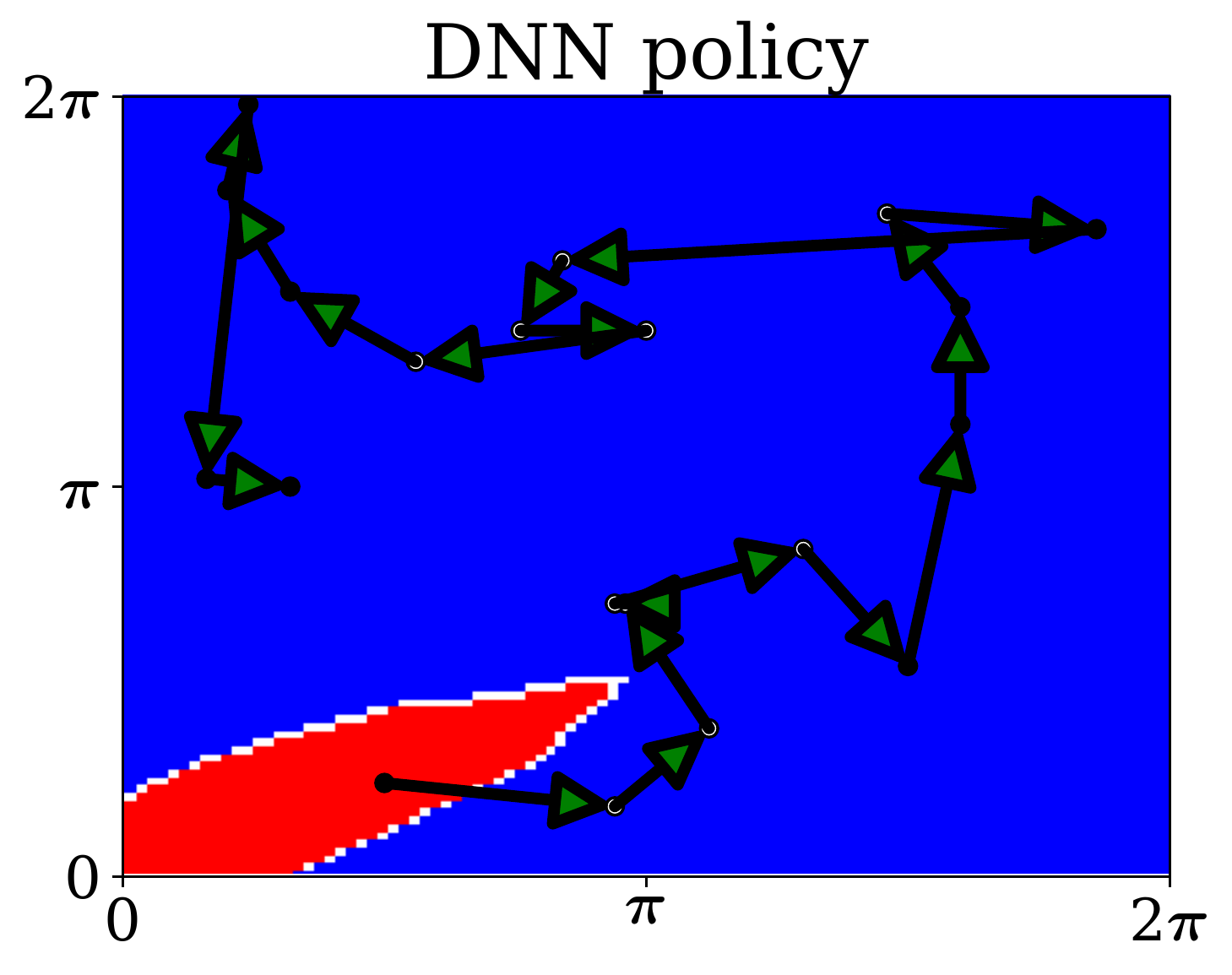}}\hspace{0em}%
	\subfloat[][]{\includegraphics[width=0.31\textwidth, valign=c]{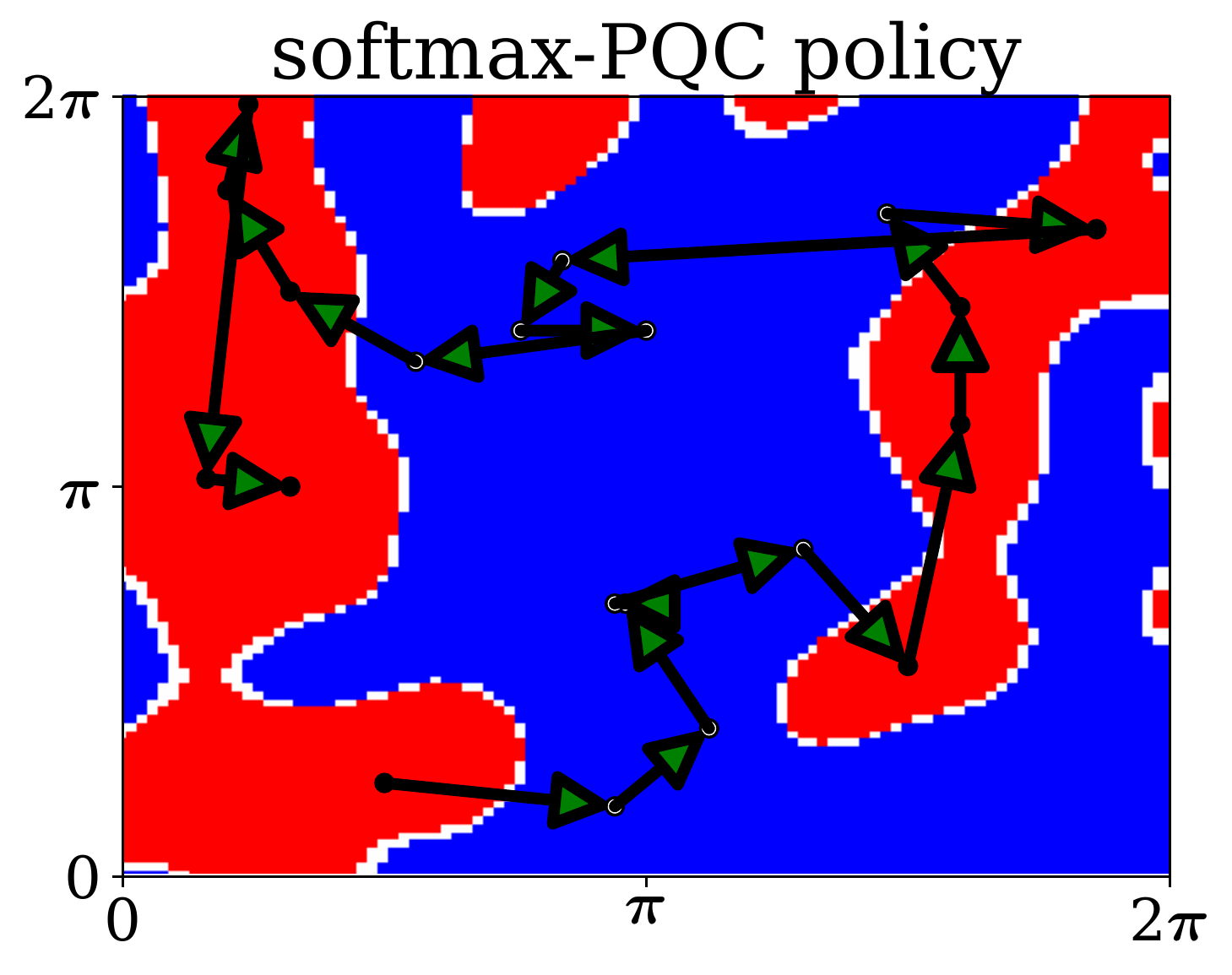}}\hspace{0em}
  \caption{\textbf{Prototypical policies learned by \textsc{softmax-PQC} agents and DNN agents in PQC-generated environments. }All policies are associated to the labeling function of Fig.\ \ref{fig:pqc-vs-nn-2}.d. Policies (a) and (b) are learned in the SL-PQC environment while policies (c) and (d) are learned in the Cliffwalk-PQC environment.}
  \label{fig:pqc-vs-nn-policies}
\end{figure*}

\clearpage

\section{Supervised learning task of Liu \emph{et al.}\label{sec:DLP-task}}

Define $p$ a large prime number, $n = \lceil \log_2(p-1) \rceil$, and $g$ a generator of $\mathbb{Z}_p^*=\{1,2,\dots,p-1\}$ (i.e., a $g \in \mathbb{Z}_p^*$ such that $\{g^y, y\in \mathbb{Z}_{p-1}\} = \mathbb{Z}_p^*$). The DLP consists in computing $\log_g x$ on input $x \in \mathbb{Z}_p^*$. Based on DLP, Liu \emph{et al.}\ \cite{liu20} define a concept class $\mathcal{C}=\{f_s\}_{s\in \mathbb{Z}_{p-1}}$ over the input space $\mathcal{X}=\mathbb{Z}_p^*$, where each labeling function of this concept class is defined as follows:
\begin{equation}\label{eq:labeling-fct-DLP}
    f_s(x)=\begin{cases}+1,&\text{if }\log_g x\in[s,s+\frac{p-3}{2}],\\ -1, &\text{otherwise.}\end{cases}
\end{equation}
Each function $f_s:\mathbb{Z}_p^*\to \{-1,1\}$ hence labels half the elements in $\mathbb{Z}_p^*$ with a label $+1$ and the other half with a label $-1$. We refer to Figure 1 in Ref.\ \cite{liu20} for a good visualization of all these objects.\\
The performance of a classifier $f$ is measured in terms of its testing accuracy
\begin{equation*}
\text{Acc}_f(f_s) = \text{Pr}_{x\sim\mathcal{X}}[f(x) = f_s(x)].
\end{equation*}

\section{Proof of Theorem \ref{thm:separations-DLP}\label{sec:proof-thm-separations}}

In the following, we provide constructions of \emph{a)} fully random, \emph{b)} partially random and \emph{c)} fully deterministic environments satisfying the properties of Theorem \ref{thm:separations-DLP}. We consider the three families of environments separately and provide individual lemmas specifying their exact separation properties.

\paragraph{Fully random: the SL-DLP environment.}
This result is near-trivially obtained by noting that any classification problem can be easily mapped to a (degenerate) RL problem. For this, the environment will be an MDP defined as follows: its state space is the input space of the classification problem, its action space comprises all possible labels, rewards are trivially $+1$ for assigning a correct label to an input state and $-1$ otherwise, and the initial and next-state transition probabilities are state-independent and equal to the input distribution of the classification task. The optimal policy of this MDP is clearly the optimal classifier of the corresponding SL task. Consider now the classification task of Liu \emph{et al.}, defined in detail in Appendix \ref{sec:DLP-task}: the input distribution is taken to be uniform on the state space, i.e., $P(s_t) = \frac{1}{|S|}$, and the performance of a classifier $f$ w.r.t.\ a labeling (or ground truth) function $f^*$ is measured in terms of a testing accuracy
\begin{equation}
\text{Acc}_f(f^*) = \frac{1}{|S|} \sum_{s} \text{Pr}[f(s) = f^*(s)].
\end{equation}
For the MDP associated to this classification task and length-$1$ episodes of interaction, the value function of any policy $\pi(a|s)$ is given by
\begin{align*}
V_\pi(s_0) &= \frac{1}{|S|}\sum_{s_0} \left(\pi(f^*(s_0)|s_0) - \pi(-f^*(s_0)|s_0)\right) \\
&= \frac{1}{|S|}\sum_{s_0} 2\pi(f^*(s_0)|s_0) - 1\\
&= 2 \text{Acc}_\pi(f^*)-1,
\end{align*}
which is trivially related to the testing accuracy of this policy on the classification task. Note that we also have $V_\text{rand}(s_0) = 0$ and $V_\text{opt}(s_0) = 1$.\\
Since these observations hold irrespectively of the labeling function $f^*$, we can show the following result:

\begin{lemma}[Quantum advantage in SL-DLP]\label{lemma:separation-SL-DLP}
There exists a uniform family of SL-DLP MDPs, each derived from a labeling function $f^*$ of the DLP concept class $\mathcal{C}$ (see Appendix \ref{sec:DLP-task}),  for which classical hardness and quantum learnability holds. More specifically, the performance of any classical learner is upper bounded by $1/\textnormal{poly}(n)$, while that of a class of quantum agents is lower bounded by $0.98$ with probability above $2/3$ (over the randomness of their interaction with the environment and noise in their implementation).
\end{lemma}
\begin{proof}
Classical hardness is trivially obtained by contraposition: assuming no classical polynomial-time algorithm can solve DLP, then using Theorem $1$ of Liu \emph{et al.}, any classical policy would have testing accuracy $\text{Acc}_\pi(f^*) \leq 1/2 + 1/\text{poly}(n)$, and hence its value function would be $V_\pi(s_0) \leq 1/\text{poly}(n)$.

For quantum learnability, we define an agent that first collects $\text{poly}(n)$ random length-$1$ interactions (i.e., a random state $s_0$ and its associated reward for an action $+1$, from which the label $f^*(s_0)$ can be inferred), and use Theorem $2$ of Liu \emph{et al.}\ to train a classifier that has test accuracy at least $0.99$ with probability at least $2/3$ (this process can be repeated $\mathcal{O}\left(\log(\delta^{-1})\right)$ times to increase this probability to $1-\delta$ via majority voting). This classifier has a value function $V_\pi(s_0) \geq 0.98$.
\end{proof}
Note that this proof trivially generalizes to episodes of interaction with length greater than $1$, when preserving the absence of temporal correlation in the states experienced by the agents. For episodes of length $H$, the only change is that the value function of any policy, and hence the bounds we achieve, get multiplied by a factor of $\frac{1-\gamma^{H}}{1-\gamma}$ for a discount factor $\gamma<1$ and by a factor $H$ for $\gamma=1$.

\paragraph{Partially random: the Cliffwalk-DLP environment.}
One major criticism to the result of Lemma \ref{lemma:separation-SL-DLP} is that it applies to a very degenerate, fully random RL environment. In the following, we show that similar results can be obtained in environments based on the same classification problem, but while imposing more temporal structure and less randomness (such constructions were introduced in Ref.\ \cite{dunjko17b}, but for the purpose of query separations between RL and QRL). For instance, one can consider cliffwalk-type environments, inspired by the textbook ``cliff walking'' environment of Sutton \& Barto \cite{sutton98}. This class of environments differs from the previous SL-DLP environments in its state and reward structure: in any episode of interaction, experienced states follow a fixed ``path'' structure (that of the cliff) for correct actions, and a wrong action yields to immediate ``death'' (negative reward and episode termination).
We slightly modify this environment to a ``slippery scenario'' in which, with a $\delta$ probability, any action may lead to a uniformly random position on the cliff. This additional randomness allows us to prove the following separation: 
\begin{lemma}[Quantum advantage in Cliffwalk-DLP]\label{lemma:separation-cliffwalk-DLP}
There exists a uniform family of Cliffwalk-DLP MDPs with arbitrary slipping probability $\delta \in [0.86,1]$ and discount factor $\gamma \in [0,0.9]$, each derived from a labeling function $f^*$ of the DLP concept class $\mathcal{C}$, for which classical hardness and quantum learnability holds. More specifically, the performance of any classical learner is upper bounded by $V_\textnormal{rand}(s_0)+0.1$, while that of a class of quantum agents is lower bounded by $V_\textnormal{opt}(s_0)-0.1$ with probability above $2/3$ (over the randomness of their interaction with the environment and noise in their implementation). Since $V_\textnormal{rand}(s_0) \leq -\frac{1}{2}$ and $V_\textnormal{opt}=0$, we always have a classical-quantum separation.
\end{lemma}
The proof of this lemma is deferred to Appendix \ref{sec:proof-cliffwalk-DLP} for clarity.

\paragraph{Fully deterministic: the Deterministic-DLP environment.}
The simplest example of a deterministic RL environment where separation can be proven is a partially observable MDP (POMDP) defined as follows: it constitutes a 1-D chain of states of length $k+2$, where $k$ is $\text{poly}(n)$. We refer to the first $k$ states as ``training states", and we call the last two states ``test'' and ``limbo'' states, respectively.
The training states are of the form $(x, f_s(x))$, i.e., a point uniformly sampled and its label. The actions are $+1,-1$, and both lead to the same subsequent state on the chain (since the same $(x, f_s(x))$ can appear twice in the chain, this is the reason why the environment is partially observable), and no reward is given for the first $k$ states. In the test state, the agent is only given a point $x$ with no label. A correct action provides a reward of $1$ and leads to the beginning of the chain, while an incorrect action leads to the limbo state, which self-loops for both actions and has no rewards. In other words, after poly-many examples where the agent can learn the correct labeling, it is tested on one state. Failure means it will never obtain a reward.

For each concept $f_s$, we define exponentially many environments obtained by random choices of the states appearing in the chain. In a given instance, call $T=(x_0, \ldots, x_{k-1})$ the training states of that instance, $x_k$ its testing state and $l$ its limbo state. The interaction of an agent with the environment is divided into episodes of length $k+1$, but the environment keeps memory of its state between episodes. This means that, while the first episode starts in $x_0$, depending on the performance of the agent, later episodes start either in $x_0$ or in $l$. For a policy $\pi$, we define the value $V_\pi(s_0)$ as the expected reward\footnote{Note that we assume here a discount factor $\gamma=1$, but our results would also hold for an arbitrary $\gamma>0$, if we scale the reward of the testing state to $\gamma^{-k}$.} of this policy in any episode of length $k+1$ with an initial state $s_0 \in\{x_0,l\}$. Since the testing state $x_k$ is the only state to be rewarded, we can already note that $V_\pi(x_0) = \pi(f^*(x_k) | T, x_{k}) $, that is, the probability of the policy correctly labeling the testing state $x_k$ after having experienced the training states $T$. Also, since $s_0 \in\{x_0,l\}$ and $V_\pi(l) = 0$, we have $V_\pi(x_0) \geq V_\pi(s_0)$.

With this construction, we obtain the following result: 
\begin{lemma}[Quantum advantage in Deterministic-DLP]\label{lemma:separation-deterministic-DLP}
There exists a uniform family of Deterministic-DLP POMDPs (exponentially many instances for a given concept $f_s$ of the DLP classification problem) where:\\
    1) (classical hardness) if there exists a classical learning agent which, when placed in a randomly chosen instance of the environment, has value $V_c(s_0) \geq 1/2+1/\textnormal{poly}(n)$ (that is, $1/\textnormal{poly}(n)$ better than a random agent), with probability at least $0.845$ over the choice of environment and the randomness of its learning algorithm, then there exists an efficient classical algorithm to solve DLP,\\
    2) (quantum learnability) there exists a class of quantum agents that attains a value $V_q(s_0) = 1$ (that is, the optimal value) with probability at least $0.98$ over the choice of environment and randomness of the learning algorithm.
\end{lemma}
The proof of this lemma is deferred to Appendix \ref{sec:proof-deterministic-DLP} for clarity.

By combining our three lemmas, and taking the weakest separation claim for the cases \emph{ii)} and \emph{iii)}, we get Theorem \ref{thm:separations-DLP}.
For the interested reader, we list the following remarks, relating to the proofs of these lemmas:
\begin{itemize}[leftmargin=3mm]
	\item SL-DLP and Deterministic-DLP are the two closest environments to the DLP classification task of Liu \emph{et al.} While the value function in SL-DLP is trivially equivalent to the accuracy of the classification problem, we find the value function in Deterministic-DLP to be \emph{weaker} than this accuracy. Namely, a high accuracy trivially leads to a high value while a high (or non-trivial) value does not necessarily lead to a high (or non-trivial) accuracy (in all these cases, the high probability over the randomness of choosing the environments and of the learning algorithms is implied). This explains why the classical hardness statement for Deterministic-DLP is weaker than in SL-DLP.
	\item In Cliffwalk-DLP, it is less straightforward to relate the testing accuracy of a policy to its performance on the deterministic parts of the environment, which explains why we trivially upper bound this performance by $0$ on these parts. We believe however that these deterministic parts will actually make the learning task much harder, since they strongly restrict the part of the state space the agents can see. This claim is supported by our numerical experiments in Sec.\ \ref{sec:PQC-env}. Also, since we showed classical hardness for fully deterministic environments, it would be simple to construct a variant of Cliffwalk-DLP where these deterministic parts would be provably hard as well.
\end{itemize}

\section{Proof of Lemma \ref{lemma:separation-cliffwalk-DLP}\label{sec:proof-cliffwalk-DLP}}

Consider a slippery cliffwalk environment defined by a labeling function $f^*$ in the concept class $\mathcal{C}$ of Liu \emph{et al.} This cliffwalk has $p-1$ states ordered, w.l.o.g., in their natural order, and correct actions (the ones that do not lead to immediate ``death") $f^*(i)$ for each state $i \in \mathbb{Z}_p^*$. For simplicity of our proofs, we also consider circular boundary conditions (i.e, doing the correct action on the state $p-1$ of the cliff leads to the state $1$), random slipping at each interaction step to a uniformly sampled state on the cliff with probability $\delta>0$, an initialization of each episode in a uniformly sampled state $i\in \mathbb{Z}_p^*$, and a $0$ ($-1$) reward for doing the correct (wrong) action in any given state.

\subsection{Upper bound on the value function}

The value function of any policy $\pi$ which has probability $\pi(i)$ (we abbreviate $\pi(f^*(i)|i)$ to $\pi(i)$) of doing the correct action in state $i \in \mathbb{Z}_p^*$ is given by:
\begin{equation}\label{eq:value-fct-cliffwalk-DLP}
    V_\pi(i) = \pi(i) \gamma  \left((1-\delta)V_\pi(i+1) + \delta \frac{1}{p-1}\sum_{j=1}^{p-1} V_\pi(j)\right) - (1-\pi(i))
\end{equation}
Since this environment only has negative rewards, we have that $V_\pi(i)\leq 0$ for any state $i$ and policy $\pi$, which allows us to write the following inequality:
\begin{equation*}
    V_\pi(i) \leq \pi(i) \gamma \left(\delta \frac{1}{p-1}\sum_{j=1}^{p-1} V_\pi(j)\right) - (1-\pi(i))
\end{equation*}
We use this inequality to bound the following term:
\begin{align*}
    \frac{1}{p-1}\sum_{i=1}^{p-1} V_\pi(i) &\leq \frac{1}{p-1} \sum_{i=1}^{p-1}\left( \pi(i)  \frac{\gamma \delta}{p-1}\sum_{j=1}^{p-1} V_\pi(j) - (1-\pi(i))\right)\\
    &= \left(\frac{1}{p-1} \sum_{i=1}^{p-1} \pi(i)\right)\left(  \frac{\gamma\delta}{p-1}\sum_{j=1}^{p-1} V_\pi(j) + 1\right) - 1
\end{align*}
We note that the first factor is exactly the accuracy of the policy $\pi$ on the classification task of Liu \emph{et al.}:
\begin{equation*}
    \text{Acc}_\pi(f^*) = \frac{1}{p-1} \sum_{i=1}^{p-1} \pi(i).
\end{equation*}
We hence have:
\begin{equation*}
    \frac{1}{p-1}\sum_{i=1}^{p-1} V_\pi(i) \leq \text{Acc}_\pi(f^*) \left( \gamma\delta\frac{1}{p-1}\sum_{j=1}^{p-1} V_\pi(j) + 1 \right) - 1
\end{equation*}
which is equivalent to:
\begin{equation*}
    \frac{1}{p-1}\sum_{i=1}^{p-1} V_\pi(i) \leq \frac{\text{Acc}_\pi(f^*)-1}{1-\text{Acc}_\pi(f^*)\gamma\delta}
\end{equation*}
when $\text{Acc}_\pi(f^*)\gamma\delta < 1$.\\
We now note that this average value function is exactly the value function evaluated on the initial state $s_0$ of the agent, since this state is uniformly sampled from $\mathbb{Z}_p^*$ for every episode. Hence,
\begin{equation}\label{eq:upper-bound-value-fct}
    V_\pi(s_0) \leq \frac{\text{Acc}_\pi(f^*)-1}{1-\text{Acc}_\pi(f^*)\gamma\delta}
\end{equation}

\subsection{Lower bound on the value function}

Again, by noting in Eq.\ (\ref{eq:value-fct-cliffwalk-DLP}) that we have $V_\pi(i)\leq 0$ and $\pi(i)\leq 1$ for any policy $\pi$ and state $i \in \mathbb{Z}_p^*$, we have:
\begin{equation*}
    V_\pi(i) \geq \gamma \left((1-\delta)V_\pi(i+1) +  \frac{\delta}{p-1}\sum_{j=1}^{p-1} V_\pi(j)\right) - (1-\pi(i))
\end{equation*}
We use this inequality to bound the value function at the initial state $s_0$:
\begin{align*}
    V_\pi(s_0) &= \frac{1}{p-1}\sum_{i=1}^{p-1} V_\pi(i)\\
     &\geq\gamma \left(\frac{1-\delta}{p-1}\sum_{i=1}^{p-1} V_\pi(i+1) + \frac{\delta}{p-1}\sum_{j=1}^{p-1} V_\pi(j)\right)+ \frac{1}{p-1}\sum_{i=1}^{p-1} \pi(i) - 1\\
     &= \gamma \left((1-\delta)V_\pi(s_0) + \delta V_\pi(s_0)\right) + \text{Acc}_\pi(f^*) - 1\\
     &= \gamma V_\pi(s_0) + \text{Acc}_\pi(f^*) -1
\end{align*}
by using the circular boundary conditions of the cliffwalk in the third line.\\
This inequality is equivalent to:
\begin{equation}\label{eq:lower-bound-value-fct}
    V_\pi(s_0) \geq \frac{\text{Acc}_\pi(f^*) -1}{1-\gamma}
\end{equation}
when $\gamma < 1$.

\subsection{Bounds for classical hardness and quantum learnability}
We use the bounds derived in the two previous sections to prove classical hardness and quantum learnability of this task environment, as stated in Lemma \ref{lemma:separation-cliffwalk-DLP}.

For this, we start by noting the following expression for the value function of a random policy (one that does random actions in all states):
\begin{align*}
    V_\text{rand}(s_0) &= \frac{\gamma}{2}\left(\frac{1-\delta}{p-1}\sum_{i=1}^{p-1} V_\text{rand}(i+1) + \frac{\delta}{p-1}\sum_{j=1}^{p-1} V_\text{rand}(j)  \right) -\frac{1}{2}\\
    &= \frac{\gamma}{2}V_\text{rand}(s_0) - \frac{1}{2} = -\frac{1}{2-\gamma}
\end{align*}
again due to the circular boundary conditions of the cliffwalk and the resulting absence of termination conditions outside of ``death".\\
As for the value function of the optimal policy, this is trivially $V_\text{opt} = 0$.

\subsubsection{Proof of classical hardness}

For any policy $\pi$, we define the function $g(x,\delta,\gamma) = V(x,\delta,\gamma) - V_\text{rand}(\gamma)$, where we adopt the short-hand notation $x = \text{Acc}_\pi(f^*)$ and call $V$ the upper bound on the value function $V_\pi(s_0)$ of $\pi$. The expression of $g(x,\delta,\gamma)$ (for $(x, \delta, \gamma) \neq (1,1,1)$) is given by:
\begin{equation}
    g(x,\delta,\gamma) = \frac{x-1}{1-\delta\gamma x} + \frac{1}{2-\gamma}
\end{equation}
To prove classical hardness, it is sufficient to show that $x \leq 0.51$ implies $g(x,\delta,\gamma) \leq 0.1$ for $\delta \in [\delta_0,1]$, $\gamma \in [0,\gamma_1]$ and a $\{\delta_0,\gamma_1\}$ pair of our choosing. To see this, notice that the contraposition gives $x = \text{Acc}_\pi(f^*) > 0.51$ which is sufficient to construct an efficient algorithm that solves DLP. To achieve this result, we show the three following inequalities, $\forall\ x \leq 0.51$ and $\forall\ (\delta,\gamma) \in [\delta_0,1]\times[0,\gamma_1]$:
\begin{equation*}
	 g(x,\delta,\gamma) \overset{(i)}{\leq} g(0.51,\delta,\gamma) \overset{(ii)}{\leq} g(0.51,\delta_0,\gamma)\overset{(iii)}{\leq} g(0.51,\delta_0,\gamma_1)
\end{equation*}
where $\delta_0$ and $\gamma_1$ are chosen such that $g(0.51,\delta_0,\gamma_1) \leq 0.1$.
\begin{proof}[Proof of (i)]
    We look at the derivative of $g$ w.r.t.\ $x$:
    \begin{equation*}
        \frac{\partial g(x,\delta,\gamma)}{\partial x} = \frac{1-\delta\gamma}{(1-\delta\gamma x)^2} \geq 0\quad \forall (x,\delta,\gamma) \in [0,1]^3\backslash(1,1,1)
    \end{equation*}
    and hence $g$ is an increasing function of $x$, which gives our inequality.
\end{proof}
\begin{proof}[Proof of (ii)]
    We look at the derivative of $g$ w.r.t.\ $\delta$:
    \begin{equation*}
        \frac{\partial g(x,\delta,\gamma)}{\partial \delta} = \frac{\gamma(x-1)x}{(1-\delta\gamma x)^2} \leq 0\quad \forall (x,\delta,\gamma) \in [0,1]^3\backslash(1,1,1)
    \end{equation*}
    and hence $g$ is a decreasing function of $\delta$, which gives our inequality.
\end{proof}
\begin{proof}[Proof of (iii)]
    We look at the derivative of $g$ w.r.t.\ $\gamma$:
    \begin{equation*}
        \frac{\partial g(x,\delta,\gamma)}{\partial \gamma} = \frac{\delta(x-1)x}{(1-\delta\gamma x)^2} + \frac{1}{(2-\gamma)^2} \quad \forall (x,\delta,\gamma) \in [0,1]^3\backslash(1,1,1)
    \end{equation*}
    We have:
    \begin{equation*}
        \begin{gathered}
        \frac{\partial g(x,\delta,\gamma)}{\partial \gamma} \geq 0 \Leftrightarrow \left((\delta x)^2+\delta(x^2-x)\right)\gamma^2 - 2\delta(2x^2-x)\gamma +4\delta(x^2-x)+1 \geq 0
        \end{gathered}
    \end{equation*}
    By setting $x=0.51$ and $\delta=0.86$, we find 
    \begin{equation*}
        \frac{\partial g(0.51,0.86,\gamma)}{\partial \gamma} \geq 0 \quad \forall \gamma \in [0,1] 
    \end{equation*}
    since the roots of the second-degree polynomial above are approximately $\{-2.91, 2.14\}$ and we have $(\delta x)^2+\delta(x-1)x \approx -0.0225 < 0$.\\
    Hence $g(0.51,\delta_0,\gamma)$ is an increasing function of $\gamma$, which gives our inequality.
    \end{proof}
    
Given that $g(0.51,0.86,0.9) \approx 0.0995 < 0.1$, we then get our desired result for $\delta_0 = 0.86$ and $\gamma_1 = 0.9$. Noting that $V_\pi(s_0)-V_\text{rand}(\gamma) \leq g(x,\delta,\gamma) \leq 0.1$ from Eq.\ (\ref{eq:upper-bound-value-fct}), we hence have classical hardness $\forall\ (\delta,\gamma) \in [\delta_0,1]\times[0,\gamma_1]$.

\subsubsection{Proof of quantum learnability}
Proving quantum learnability is more trivial, since, for $\text{Acc}_\pi(f^*) \geq 0.99$ and $\gamma \leq 0.9$, we directly have, using Eq.\ (\ref{eq:lower-bound-value-fct}):
\begin{equation*}
    V_\pi(s_0) \geq -0.1 = V_\text{opt} -0.1
\end{equation*}
To conclude this proof, we still need to show that we can obtain in this environment a policy $\pi$ such that $\text{Acc}_\pi(f^*) \geq 0.99$ with high probability. For that, we use agents that first collect $\text{poly}(n)$ \emph{distinct} samples (states $s$ and their inferred labels $f^*(s)$) from the environment (distinct in order to avoid biasing the distribution of the dataset with the cliffwalk temporal structure). This can be done efficiently in $\text{poly}(n)$ interactions with the environment, since each episode is initialized in a random state $s_0 \in \mathbb{Z}_p^*$. We then use the learning algorithm of Liu \emph{et al.}\ to train a classifier $\pi$ with the desired accuracy, with high probability.

\section{Proof of Lemma \ref{lemma:separation-deterministic-DLP}\label{sec:proof-deterministic-DLP}}

\subsection{Proof of classical hardness}

Suppose that a polynomial-time classical agent achieves a value $V_c(s_0) \geq \frac{1}{2} + \frac{1}{\text{poly}(n)}$ with probability $(1-\delta)$ over the choice of environment and the randomness of its learning algorithm. We call ``success" the event $V_c(s_0) \geq \frac{1}{2} + \frac{1}{\text{poly}(n)}$ and $S_{\delta}$ the subset of the instances $S = \{ T,x_k\}$ for which, theoretically, a run of the agent would ``succeed" (this is hence a set that depends on the randomness of the agent).

Note that, on every instance in $S_\delta$, $\pi(f^*(x_k) | T, x_{k}) = V_c(x_0) \geq V_c(s_0) \geq \frac{1}{2} + \frac{1}{\text{poly}(n)}$. Since this probability is bounded away from $1/2$ by an inverse polynomial, this means that we can ``boost" it to a larger probability $(1-\varepsilon)$. More specifically, out of the policy $\pi$ obtained after interacting for $k$ steps with the environment,  we define a classifier $f_c$ acting on $x_k$ such that we sample $\mathcal{O}\left(\log(\varepsilon^{-1})\right)$-many times from $\pi(a | T, x_{k})$ and label $x_k$ by majority vote. For the instances in $S_\delta$, the probability of correctly labeling $x_k$ is $\text{Pr}\left[ f_c(x_k) = f^*(x_k)\right] \geq 1-\varepsilon$.

Define $P(T)=\text{Pr}[\text{T}=T]$ and $P(x_k)=\text{Pr}[\text{x}_\text{k}=x_k]$ the probabilities of sampling certain training states $T$ and a testing state $x_k$, when choosing an instance of the environment. We now look at the following quantity:
\begin{align*}
    \mathbb{E}_{P(T)}\left[ \text{Acc}_{f_c}(T) \right] &= \sum_{T} P(T) \sum_{x_k} P(x_k) \text{Pr}\left[ f_c(x_k) = f^*(x_k) | T,x_k\right]\\
    &= \sum_{T,x_k} P(T,x_k) \text{Pr}\left[ f_c(x_k) = f^*(x_k) | T,x_k\right]\\
    &\geq \sum_{T,x_k} P(T,x_k) \text{Pr}\big[\text{success}|T,x_k\big] \times \text{Pr}\big[ f_c(x_k) = f^*(x_k) | T,x_k, \text{success}\big]\\\
    &\geq (1-\delta) (1-\varepsilon)
\end{align*}
since $\text{Pr}\left[ f_c(x_k) = f^*(x_k) | T, x_k\right] \geq 1-\varepsilon$ for instances in $S_\delta$ and $\sum_{T,x_k} P(T,x_k) \text{Pr}\big[\text{success}|T,x_k\big] \allowbreak\geq 1-\delta$ by definition.\\
In the following, we set $1-\varepsilon = 0.999$ and $1-\delta \geq 0.845$ (the reason for this becomes apparent below), such that:
\begin{equation}\label{eq:bound-avg-accuracy}
    \mathbb{E}_{P(T)}\left[ \text{Acc}_{f_c}(T) \right] \geq 0.844155 > \frac{5}{6} + \frac{1}{96}
\end{equation}

Now, consider the following learning algorithm: given a training set $T$, construct a Deterministic-DLP environment that uses this $T$ and a randomly chosen $x_k$, and define the classifier $f_c$ that boosts the $\pi(a | T, x_{k})$ obtained by running our classical agent on this environment (as explained above). We want to show that $f_c$ has accuracy $\text{Acc}_{f_c}(T) \geq \frac{1}{2} + \frac{1}{\text{poly}(n)}$ with probability at least $2/3$ over the choice of $T$ and the randomness of its construction, which is sufficient to solve DLP classically. For that, we show a stronger statement. Call $\mathcal{T}_\text{succ}$ the subset of all instances of training states $\mathcal{T} = \left\{ T \right\}$ for which $\text{Acc}_{f_c}(T) \geq \frac{1}{2} + \frac{1}{\text{poly}(n)}$. We prove by contradiction that $\abs{\mathcal{T}_\text{succ}} \geq \frac{2\abs{\mathcal{T}}}{3}$:\\
Assume $\abs{\mathcal{T}_\text{succ}} < \frac{2\abs{\mathcal{T}}}{3}$, then
\begin{align*}
    \mathbb{E}_{P(T)}\left[ \text{Acc}_{f_c}(T) \right] &= \sum_{T} P(T) \text{Acc}_{f_c}(T) \\
    &= \frac{1}{\abs{\mathcal{T}}} \left( \sum_{T\in \mathcal{T}_\text{succ}} \text{Acc}_{f_c}(T) + \sum_{T\notin \mathcal{T}_\text{succ}} \text{Acc}_{f_c}(T) \right)\\
    &< \frac{\abs{\mathcal{T}_\text{succ}}}{\abs{\mathcal{T}}} \times 1 + \frac{\abs{\mathcal{T}}-\abs{\mathcal{T}_\text{succ}}}{\abs{\mathcal{T}}}\left( \frac{1}{2} + \frac{1}{\text{poly}(n)} \right)\\
    &< \frac{5}{6} + \frac{1}{3\text{poly}(n)} < 0.844155
\end{align*}
for large enough $n$, in contradiction with Eq.\ (\ref{eq:bound-avg-accuracy}).

Hence, with probability at least $2/3$ over the choice of training states and the randomness of the learning algorithm, our constructed classifier has accuracy $\text{Acc}_{f_c}(T) \geq \frac{1}{2} + \frac{1}{\text{poly}(n)}$. By using Theorem 8, Remark 1 of Liu \emph{et al.}, this is sufficient to construct an efficient classical algorithm that solves DLP.

\subsection{Proof of quantum learnability}
Using the learning algorithm of Liu \emph{et al.}, we can construct a quantum classifier that achieves accuracy $\text{Acc}_{q}(T) \geq 0.99$ with probability at least $2/3$ over the randomness of the learning algorithm and the choice of training states $T$, of length $\abs{T} = \text{poly}(n)$. Now define instead training states $T$ of length $\abs{T}=M\text{poly}(n)$, for $M=\mathcal{O}\left(\log(\delta'^{-1})\right)$ (hence $\abs{T}$ is still polynomial in $n$), and use each of the $M$ segments of $T$ to train $M$ independent quantum classifiers. Define $f_q$ as a classifier that labels $x_k$ using a majority vote on the labels assigned by each of these classifiers. This constructed classifier has accuracy $\text{Acc}_{f_q}(T) \geq 0.99$ with now probability $1-\delta'$ over the choice of training states $T$ and the randomness of the learning algorithm.

We then note that, by calling ``success" the event $\text{Acc}_{f_q}(T) \geq 0.99$, we have:
\begin{align*}
    \sum_{T,x_k} P(T,x_k) &\text{Pr}\big[V_q(x_0)=1 | T,x_k\big]\\ 
    &\geq  \sum_T P(T) \sum_{x_k} P(x_k) \text{Pr}\big[ \text{success} | T \big] \times \text{Pr}\big[V_q(x_0)=1 | T,x_k, \text{success}\big]\\
    &= \sum_T P(T)\text{Pr}\big[ \text{success} | T \big] \sum_{x_k} P(x_k) \times \text{Pr}\big[f_q(x_k) = f^*(x_k) | T,x_k, \text{success}\big]\\
    &=  \sum_T P(T)\text{Pr}\big[ \text{success} | T \big] \text{Acc}_{f_q}(T)\\
    &\geq (1-\delta')\times0.99
\end{align*}
which means that our constructed agent achieves a value $V_q(x_0)=1$ (which also implies $V_q(s_0)=1$) with probability at least $(1-\delta')\times0.99$ over the choice of environment and the randomness of the learning algorithm. By setting $(1-\delta') = 0.98/0.99$, we get our statement.

\section{Construction of a PQC agent for the DLP environments\label{sec:PQC-agent-DLP}}

In the two following appendices, we construct a PQC classifier that can achieve close-to-optimal accuracy in the classification task of Liu \emph{et al.} \cite{liu20} (see Appendix \ref{sec:DLP-task}), and can hence also be used as a learning model in the DLP environments defined in Sec.\ \ref{sec:DLP-RL}.

\subsection{Implicit v.s.\ explicit quantum SVMs}

To understand the distinction between the quantum learners of Liu \emph{et al.} and the PQC policies we are constructing here, we remind the reader of the two models for quantum SVMs defined in Ref.\ \cite{schuld19b}: the explicit and the implicit model. Both models share a feature-encoding unitary $U(x)$ that encodes data points $x$ into feature state $\ket{\phi(x)}=U(x)\ket{0^{\otimes n}}$.\\
In the implicit model, one first evaluates the kernel values
\begin{equation}
	K(x_i,x_j) = \abs{\braket{\phi(x_i)}{\phi(x_j)}}^2
\end{equation}
for the feature states associated to every pair of data points $\{x_i,x_j\}$ in the dataset, then uses the resulting kernel matrix in a classical SVM algorithm. This algorithm returns a hyperplane classifier in feature space, defined by its normal vector $\bra{\bm{w}} = \sum_{i}\alpha_i \bra{\phi(x_i)}$ and bias $b$, such that the sign of $\abs{\braket{\bm{w}}{\phi(x)}}^2 + b$ gives the label of $x$.\\
In the explicit model, the classifier is instead obtained by training a parametrized $\ket{\bm{w}_{\bm{\theta}}}$. Effectively, this classifier is implemented by applying a variational unitary $V(\bm{\theta})$ on the feature states $\ket{\phi(x)}$ and measuring the resulting quantum states using a fixed observable, with expectation value $\abs{\braket{\bm{w}_{\bm{\theta}}}{\phi(x)}}^2$.

In the following sections, we describe how the implicit quantum SVMs of Liu \emph{et al.} can be transformed into explicit models while guaranteeing that they can still represent all possible optimal policies in the DLP environments. And in Appendix \ref{sec:training-SL-DLP}, we show that, even under similar noise considerations as Liu \emph{et al.}, these optimal policies can also be found using $\text{poly}(n)$ random data samples.

\subsection{Description of the PQC classifier}

As we just described, our classifier belongs to a family of so-called explicit quantum SVMs. It is hence described by a PQC with two parts: a feature-encoding unitary $U(x)$, which creates features $\ket{\phi(x)} = U(x) \ket{0^{\otimes n}}$ when applied to an all-0 state, followed by a variational circuit $V(\params)$ parametrized by a vector $\params$. The resulting quantum state is then used to measure the expectation value $\expval{O}_{x,\params}$ of an observable $O$, to be defined. We rely on the same feature-encoding unitary $U(x)$ as the one used by Liu \emph{et al.}, i.e., the unitary that creates feature states of the form
\begin{equation}
    \ket{\phi({x})} = \frac{1}{\sqrt{2^k}} \sum^{2^k - 1}_{i = 0} \ket{x \cdot g^{i}}
\end{equation}
for $k=n - t \log(n)$, where $t$ is a constant defined later, under noise considerations. This feature state can be seen as the uniform superposition of the image (under exponentiation $s' \mapsto g^{s'}$) of an interval of integers $[\log_{g}(x), \log_{g}(x)+ 2^k - 1]$ in log-space. Note that $U(x)$ can be implemented in $\widetilde{\mathcal{O}}({n^3})$ operations \cite{liu20}.

By noting that every labeling functions $f_s \in \mathcal{C}$ to be learned (see Eq.\ (\ref{eq:labeling-fct-DLP})) is delimiting two equally-sized intervals of $\log(\mathbb{Z}_p^*)$, we can restrict the decision boundaries to be learned by our classifier to be half-space dividing hyperplanes in log-space. In feature space, this is equivalent to learning separating hyperplanes that are normal to quantum states of the form:
\begin{equation}
    \ket{\phi_{s'}} = \frac{1}{\sqrt{(p-1)/2}} \sum^{(p-3)/2}_{i = 0} \ket{g^{s'+i}}.
\end{equation}
Noticeably, for input points $x$ such that $\log_g(x)$ is away from some delimiting regions around $s$ and $s+\frac{p-3}{2}$, we can notice that the inner product $\abs{\braket{\phi(x)}{\phi_s}}^2$ is either $\Delta = \frac{2^{k+1}}{p-1}$ or $0$, whenever $x$ is labeled $+1$ or $-1$ by $f_s$, respectively. This hence leads to a natural classifier to be built, assuming overlaps of the form $\abs{\braket{\phi(x)}{\phi_{s'}}}^2$ can be measured:
\begin{equation}\label{eq:PQC-classifier}
h_{s'}(x)=
    \begin{cases}
    1, &\text{if } \abs{\braket{\phi(x)}{\phi_{s'}}}^2 /\Delta \geq 1/2, \\
    -1, &\text{otherwise}
    \end{cases}
\end{equation}
which has an (ideal) accuracy $1-\Delta$ whenever $s'=s$.

To complete the construction of our PQC classifier, we should hence design the composition of its variational part $V(\params)$ and measurement $O$ such that they result in expectation values of the form $\expval{O}_{x,\params} = \abs{\braket{\phi(x)}{\phi_{s'}}}^2$. To do this, we note that, for $\ket{\phi_{s'}} = \hat{V}(s') \ket{0}$, the following equality holds:
\begin{align*}
	 \abs{\braket{\phi(x)}{\phi_{s'}}}^2 &= \abs{\bra{0^{\otimes n}} \hat{V}(s')^\dagger U({{x_i}}) \ket{0^{\otimes n}}}^2\\
	 	&= \text{Tr}\left[ \ket{0^{\otimes n}}\bra{0^{\otimes n}}  \rho(x,s')  \right]
\end{align*}
where $\rho(x,s') = \ket{\psi(x,s')}\bra{\psi(x,s')}$ is the density matrix of the quantum state $\ket{\psi(x,s')} = \hat{V}(s')^\dagger U({{x_i}}) \ket{0^{\otimes n}}$. Hence, an obvious choice of variational circuit is $V(\params) = \hat{V}(s')$, combined with a measurement operator $O =  \ket{0^{\otimes n}}\bra{0^{\otimes n}}$. Due to the similar nature of $\ket{\phi_s'}$ and $\ket{\phi(x)}$, it is possible to use an implementation for $\hat{V}(s')$ that is similar to that of $U(x_i)$ (take $x_i = g^{s'}$ and $k\approx n/2$).\footnote{Note that we write $\hat{V}(s')$ and $U_{s'}$ to be parametrized by $s'$ but the true variational parameter here is $g^{s'}$, since we work in input space and not in log-space.} We also note that, for points $x$ such that $\log_g(x)$ is $(p-1)\Delta/2$ away from the boundary regions of $h_{s'}$, the non-zero inner products $\abs{\braket{\phi(x)}{\phi_{s'}}}^2$ are equal to $\Delta = \mathcal{O} (n^{-t})$. These can hence be estimated efficiently to additive error, which allows to efficiently implement our classifier $h_{s'}$ (Eq.\ (\ref{eq:PQC-classifier})).

\subsection{Noisy classifier}

In practice, there will be noise associated with the estimation of the inner products  $\abs{\braket{\phi(x)}{\phi_{s'}}}^2$, namely due to the additive errors associated to sampling. Similarly to Liu \emph{et al.}, we model noise by introducing a random variable $e_{is'}$ for each data point $x_i$ and variational parameter $g^{s'}$, such that the estimated inner product is $\abs{\braket{\phi(x_i)}{\phi_{s'}}}^2+ e_{is'}$. This random variable satisfies the following equations:
\begin{equation*}
	\begin{cases}
    \ e_{is'} \in [-\Delta, \Delta] \\
    \ \mathbb{E}[e_{is'}] = 0 \\
    \ \text{Var}[e_{is'}] \leq 1/R
    \end{cases}
\end{equation*}
where $R$ is the number of circuit evaluations used to estimate the inner product. We hence end up with a noisy classifier:
\begin{equation*}
\widetilde{h}_{s'}(x_i)=
    \begin{cases}
    1, &\text{if } \left(\abs{\braket{\phi(x_i)}{\phi_{s'}}}^2 +e_{is'}\right)/\Delta \geq 1/2, \\
    -1, &\text{otherwise}
    \end{cases}
\end{equation*}

The noise has the effect that some points which would be correctly classified by the noiseless classifier have now a non zero probability of being misclassified. To limit the overall decrease in classification accuracy, we focus on limiting the probability of misclassifying points $x_i$ such that $\log_g(x_i)$ is $(p-1)\Delta/2$ away from the boundary points $s'$ and $s' + \frac{p-3}{2}$ of $g_{s'}$. We call $I_{s'}$ the subset of $\mathbb{Z}_p^*$ comprised of these points. For points in $I_{s'}$, the probability of misclassification is that of having $|e_{is'}|\geq\Delta/2$. We can use Chebyshev's inequality to bound this probability:
\begin{equation}\label{eq:prb-mismatch}
    \text{Pr}\left(\abs{e_{is'}}\geq \frac{\Delta}{2}\right)\leq\frac{4}{\Delta^2 R}
\end{equation}
since $\mathbb{E}[e_{is'}] = 0$ and $\text{Var}[e_{is'}] \leq 1/R$.

\section{Proof of trainability of our PQC agent in the SL-DLP environment\label{sec:training-SL-DLP}}

In this Appendix, we describe an optimization algorithm to train the variational parameter $g^{s'}$ of the PQC classifier we defined in Appendix \ref{sec:PQC-agent-DLP}. This task is non-trivial for three reasons: 1) even by restricting the separating hyperplanes accessible by our classifier, there are still $p-1$ candidates, which makes an exhaustive search for the optimal one intractable; 2) noise in the evaluation of the classifier can potentially heavily perturb its loss landscape, which can shift its global minimum and 3) decrease the testing accuracy of the noisy classifier. Nonetheless, we show that all these considerations can be taken into account for a simple optimization algorithm, such that it returns a classifier with close-to-optimal accuracy with high probability of success. More precisely, we show the following Theorem:
\begin{theorem}\label{thm:perf-PQC}
For a training set of size $n^c$ such that $c\geq \max\left\{\log_n(8/\delta), \log_n\left(\frac{\log(\delta/2)}{\log(1-2n^{-t})}\right)\right\}$ for $t\geq\max\left\{3\log_n(8/\delta),\log_n(16/\varepsilon)\right\}$ in the definition of $\Delta$, and a number of circuit evaluations per inner product $R \geq \max\left\{ \frac{4n^{2(t+c)}}{\delta}, \frac{128}{\varepsilon^3} \right\}$, then our optimization algorithm returns a noisy classifier $\widetilde{h}_{s'}$ with testing accuracy $\textnormal{Acc}_{\widetilde{h}_{s'}}(f_s)$ on the DLP classification task of \emph{Liu} et al.\ such that
\begin{equation*}
	\textnormal{Pr}\left(\textnormal{Acc}_{\widetilde{h}_{s'}}(f_s) \geq 1-\varepsilon \right) \geq 1-\delta.
\end{equation*}
\end{theorem}
The proof of this Theorem is detailed below.

Given a training set $X \subset \mathcal{X}$ polynomially large in $n$, i.e., $\abs{X} = n^c$, define the training loss:
\begin{equation*}\label{eq:loss-supervised}
    \mathcal{L}(s') = \frac{1}{2\abs{X}} \sum_{x \in X} \abs{h_{s'}(x)-f_s(x)}
\end{equation*}
and its noisy analog:
\begin{equation*}\label{eq:loss-supervised-noisy}
    \widetilde{\mathcal{L}}(s') = \frac{1}{2\abs{X}} \sum_{x \in X} \abs{\widetilde{h}_{s'}(x)-f_s(x)}
\end{equation*}
Our optimization algorithm goes as follows: using the noisy classifier $\widetilde{h}_{s'}$, evaluate the loss function $\widetilde{\mathcal{L}}\left(\log_g(x)\right)$ for each variational parameter $g^{s'} = x \in X$, then set \[g^{s'} = \text{argmin}_{x \in X} \widetilde{\mathcal{L}}(\log_g(x)).\] This algorithm is efficient in the size of the training set, since it only requires $\abs{X}^2$ evaluations of $\widetilde{h}_{s'}$.\\
To prove Theorem \ref{thm:perf-PQC}, we show first that we can enforce $\text{argmin}_{x \in X} \widetilde{\mathcal{L}}(\log_g(x)) = \text{argmin}_{x \in X} \mathcal{L}(\log_g(x))$ with high probability (Lemma \ref{lemma:no-mismatch}), and second, that this algorithm also leads to $s'$ close to the optimal $s$ in log-space with high probability (Lemma \ref{lemma:PAC-proof}).

\begin{lemma}\label{lemma:no-mismatch}
For a training set of size $n^c$ such that $c\geq \log_n(8/\delta)$, a $t\geq3c$ in the definition of $\Delta$, and a number of circuit evaluations per inner product $R \geq \frac{4n^{2(t+c)}}{\delta}$, we have
\begin{equation*}
     \textnormal{Pr}\left( \underset{x \in X}{\textnormal{argmin}}\ \widetilde{\mathcal{L}}(\log_g(x)) = \underset{x \in X}{\textnormal{argmin}}\ \mathcal{L}(\log_g(x)) \right) \geq 1-\frac{\delta}{2}
\end{equation*} 
\end{lemma}
\begin{proof}
In order for the minima of the two losses to be obtained for the same $x\in X$, it is sufficient to ensure that the classifiers $h_{\log_g(x_i)}$ and $\widetilde{h}_{\log_g(x_i)}$ agree on all points $x_j$, for all $(x_i,x_j) \in X$. This can be enforced by having:
\begin{equation*}
	\left(\bigcap_{\underset{i\neq j}{i,j}} x_i \in I_{\log_g(x_j)}\right) \cap \left(\bigcap_{i,s'} |e_{i,s'}|\leq\frac{\Delta}{2}\right)
\end{equation*}
that is, having for all classifiers $h_{\log_g(x_j)}$ that all points $x_i\in X$, $x_i \neq x_j$, are away from its boundary regions in log-space, and that the labels assigned to these points are all the same under noise.\\ We bound the probability of the negation of this event:
\begin{equation*}
	\text{Pr}\left( \bigcup_{\underset{i\neq j}{i,j}} x_i \notin I_{\log_g(x_j)} \cup \bigcup_{i,s'} |e_{i,s'}|\geq\frac{\Delta}{2}\right) \leq \text{Pr}\left( \bigcup_{\underset{i\neq j}{i,j}} x_i \notin I_{\log_g(x_j)}\right) + \text{Pr}\left( \bigcup_{i,s'} |e_{i,s'}|\geq\frac{\Delta}{2}\right)
\end{equation*}
using the union bound.\\
We start by bounding the first probability, again using the union bound:
\begin{align*}
\text{Pr}\left( \bigcup_{\underset{i\neq j}{i,j}} x_i \notin I_{\log_g(x_j)}\right) &\leq \sum_{\underset{i\neq j}{i,j}} \text{Pr}\left( x_i \notin I_{\log_g(x_j)} \right)\\
&= \sum_{\underset{i\neq j}{i,j}} \frac{\Delta}{2} \leq \frac{n^{2c}\Delta}{2}
\end{align*}
By setting $t\geq3c$, we have $\Delta \leq 4n^{-t} \leq 4n^{-3c}$, which allows us to bound this first probability by $\delta/4$ when $c\geq \log_n(8/\delta)$.\\
As for the second probability above, we have
\begin{align*}
	\text{Pr}\left( \bigcup_{i,s'} |e_{i,s'}|\geq\frac{\Delta}{2}\right) &\leq \sum_{i,s'} \text{Pr}\left(|e_{i,s'}|\geq\frac{\Delta}{2}\right)\\
	& \leq \frac{4n^{2c}}{\Delta^2R}
\end{align*}
using the union bound and Eq.\ (\ref{eq:prb-mismatch}). By setting $R \geq \frac{4n^{2(t+c)}}{\delta} \geq \frac{16n^{2c}}{\Delta^2\delta}$ (since $\Delta \geq 2n^{-t}$), we can bound this second probability by $\delta/4$ as well, which gives:
\begin{align*}
     \textnormal{Pr}\left( \underset{x \in X}{\textnormal{argmin}}\ \widetilde{\mathcal{L}}(\log_g(x)) = \underset{x \in X}{\textnormal{argmin}}\ \mathcal{L}(\log_g(x)) \right)  &\geq 1 - \text{Pr}\left( \bigcup_{\underset{i\neq j}{i,j}} x_i \notin I_{\log_g(x_j)} \cup \bigcup_{i,s'} |e_{i,s'}|\geq\frac{\Delta}{2}\right)\\
     &\geq 1 - \delta/2 \qedhere
\end{align*}
\end{proof}

\begin{lemma}\label{lemma:PAC-proof}
For a training set of size $n^c$ such that $c\geq \log_n\left(\frac{\log(\delta/2)}{\log(1-2\varepsilon)}\right)$, then $s'=\log_g\left( \textnormal{argmin}_{x \in X} \mathcal{L}(\log_g(x)) \right)$ is within $\varepsilon$ distance of the optimal $s$ with probability:
\begin{equation*}
    \textnormal{Pr}\left(\frac{\abs{s' - s}}{p-1} \leq \varepsilon\right) \geq 1 - \frac{\delta}{2}
\end{equation*}
\end{lemma}
\begin{proof}
We achieve this result by proving:
\begin{equation*}
   \text{Pr}\left(\frac{\abs{s'-s}}{p-1} \geq \varepsilon\right) \leq \frac{\delta}{2} 
\end{equation*}
This probability is precisely the probability that no $\log_g(x) \in \log_g(X)$ is within $\varepsilon$ distance of $s$, i.e.,
\begin{equation*}
   \text{Pr}\left(\bigcap_{x\in X}\log(x) \notin [s-\varepsilon(p-1), s + \varepsilon(p-1)]\right)
\end{equation*}
As the elements of the training set are all i.i.d., we have that this probability is equal to
\begin{equation*}
	\text{Pr}\left(\log(x) \notin [s-\varepsilon(p-1), s + \varepsilon(p-1)]\right) ^{\abs{X}}
\end{equation*}
Since all the datapoints are uniformly sampled from $\mathbb{Z}_p^*$, the probability that a datapoint is in any region of size $2\varepsilon(p-1)$ is just $2\varepsilon$. With the additional assumption that $\abs{X} = n^c \geq \log_{1-2\varepsilon}(\delta/2)$ (and assuming $\varepsilon<1/2$), we get:
\begin{equation*}
    \text{Pr}\left(\frac{\abs{s'-s}}{p-1} \geq \varepsilon\right) \leq (1-2\varepsilon)^{\log_{1-2\varepsilon}(\delta/2)} = \frac{\delta}{2} \qedhere
\end{equation*}
\end{proof}

Lemma \ref{lemma:no-mismatch} and Lemma \ref{lemma:PAC-proof} can be used to prove:
\begin{corollary}\label{cor:PAC-optim}
For a training set of size $n^c$ such that $c\geq \max\left\{\log_n(8/\delta), \log_n\left(\frac{\log(\delta/2)}{\log(1-2\varepsilon)}\right)\right\}$, a $t\geq3c$ in the definition of $\Delta$, and a number of circuit evaluations per inner product $R \geq \frac{4n^{2(t+c)}}{\delta}$, then our optimization algorithm returns a variational parameter $g^{s'}$ such that
\begin{equation*}
    \textnormal{Pr}\left(\frac{\abs{s' - s}}{p-1} \leq \varepsilon\right) \geq 1 - \delta
\end{equation*}
\end{corollary}
From here, we notice that, when we apply Corollary \ref{cor:PAC-optim} for $\varepsilon' \leq \frac{\Delta}{2}$, our optimization algorithm returns an $s'$ such that, with probability $1-\delta$, the set $I_{s'}$ is equal to $I_{s}$ and is of size $(p-1)(1-2\Delta)$. In the event where $\abs{s' - s}/(p-1) \leq \varepsilon'\leq \frac{\Delta}{2}$, we can hence bound the accuracy of the noisy classifier:
\begin{align*}
	\textnormal{Acc}_{\widetilde{h}_{s'}}(f_s) &= \frac{1}{p-1} \sum_{x \in \mathcal{X}} \text{Pr}\left(\widetilde{h}_{s'}(x) = f_s(x)\right)\\
	&\geq \frac{1}{p-1} \sum_{x \in I_{s}}  \text{Pr}\left(\widetilde{h}_{s'}(x) = f_s(x)\right)\\
	&\geq (1-2\Delta) \min_{x_i \in I_{s}} \text{Pr}\left(\abs{e_{i,s'}}\leq \frac{\Delta}{2}\right)\\
	&\geq (1-2\Delta) \left( 1-\frac{4}{\Delta^2R} \right)\\
	&= 1-\left(2\Delta\left( 1-\frac{4}{\Delta^2R} \right)+ \frac{4}{\Delta^2R}\right)
\end{align*}
with probability $1-\delta$.\\
We now set $t\geq\max\left\{3\log_n(8/\delta),\log_n(16/\varepsilon)\right\}$, $\varepsilon' = n^{-t}$ and $R\geq \max\left\{ \frac{4n^{2(t+c)}}{\delta}, \frac{128}{\varepsilon^3} \right\}$, such that $2\varepsilon' = 2n^{-t} \leq \Delta \leq 4n^{-t} \leq \frac{\varepsilon}{4}$, $\left( 1-\frac{4}{\Delta^2R} \right) \leq 1$ and $\frac{4}{\Delta^2R}\leq\frac{\varepsilon}{2}$.\\
Using these inequalities, we get 
\begin{equation*}
\textnormal{Acc}_{\widetilde{h}_{s'}}(f_s) \geq 1 - \varepsilon
\end{equation*}
with probability $1-\delta$, which proves Theorem \ref{thm:perf-PQC}.

\end{document}